\pdfoutput=1

\documentclass{fac}

\ifprodtf \else \usepackage{latexsym}\fi

\newcommand{\bhline}[1]{\noalign{\hrule height #1}}
\newcommand\black{\ensuremath{\blacktriangleright}}
\newcommand\white{\ensuremath{\vartriangleright}}

\newif\ifamsfontsloaded
\ifprodtf
  \newcommand\whbl{\white\kern-.1em--\kern-.1em\black}
  \newcommand\blwh{\black\kern-.1em--\kern-.1em\white}
  \newcommand\blbl{\black\kern-.1em--\kern-.1em\black}
  \newcommand\whwh{\white\kern-.1em--\kern-.1em\white}
  \amsfontsloadedtrue
\else
  \checkfont{msam10}
  \iffontfound
    \IfFileExists{amssymb.sty}
      {\usepackage{amssymb}\amsfontsloadedtrue
       \newcommand\whbl{\white\kern-.125em--\kern-.125em\black}%
       \newcommand\blwh{\black\kern-.125em--\kern-.125em\white}%
       \newcommand\blbl{\black\kern-.125em--\kern-.125em\black}%
       \newcommand\whwh{\white\kern-.125em--\kern-.125em\white}}
      {}
  \fi
\fi

\usepackage{amsfonts}
\newdimen\tbaselineshift
\usepackage{url}
\usepackage{mathpartir}
\usepackage{listings}
\usepackage{amsfonts,xspace}
\usepackage[lined,boxed,linesnumbered]{algorithm2e}
\usepackage[utf8]{inputenc}
\usepackage{textcomp} 
\usepackage{wrapfig}
\usepackage{floatflt} 
\usepackage{multicol}
\usepackage{multirow}
\usepackage[inference]{semantic} 
\usepackage[T1]{fontenc}
\usepackage[cmex10]{amsmath}
\usepackage{amssymb}
\usepackage{mathtools} 
\usepackage{bm} 
\usepackage{ushort} 
\usepackage[babel]{csquotes} 
\usepackage{microtype} 
\usepackage{xfrac}
\usepackage{nicefrac}
\usepackage{ifpdf}
\ifpdf 
  \usepackage{pgf}
  \usepackage{tikz}
\else 
  \usepackage{graphicx}
  \DeclareGraphicsExtensions{.eps,.bmp}
  \usepackage{pgf}
  \usepackage{tikz}
  \usepackage{pstricks}
\fi
\usetikzlibrary{calc}
\usepackage{todonotes} 
\usepackage{times}
\usepackage{booktabs}
\usepackage{subfig}
\usepackage{tcolorbox}
\usepackage{thmtools}
\usepackage{thm-restate}
\usepackage[scaled=.89]{helvet}
\usepackage{varwidth}
\usepackage{pgfplots}
\usepackage{breqn}
\usepackage{wrapfig}

\newif\ifcommentson\commentsonfalse
\ifcommentson
\newcommand{\commentAL}[1]{\begin{center} \parbox{.9\textwidth}{\textbf{\textcolor{black}{Comment AL.}} \textcolor{red}{#1 }}\end{center}}
\newcommand{\commentCP}[1]{\begin{center} \parbox{.9\textwidth}{\textbf{\textcolor{black}{Comment CP.}} \textcolor{red}{#1 }}\end{center}}
\newcommand{\commentFB}[1]{\begin{center} \parbox{.9\textwidth}{\textbf{\textcolor{black}{Comment FB.}} \textcolor{red}{#1 }}\end{center}}
\newcommand{\commentYK}[1]{\begin{center} \parbox{.9\textwidth}{\textbf{\textcolor{black}{Comment YK.}} \textcolor{red}{#1} }\end{center}}
\newcommand{\replyAL}[1]{\begin{center} \parbox{.8\textwidth}{\textbf{Reply AL.} \textcolor{blue}{#1} }\end{center}}
\newcommand{\replyCP}[1]{\begin{center} \parbox{.8\textwidth}{\textbf{Reply CP.} \textcolor{blue}{#1} }\end{center}}
\newcommand{\replyFB}[1]{\begin{center} \parbox{.8\textwidth}{\textbf{Reply FB.} \textcolor{blue}{#1} }\end{center}}
\newcommand{\replyYK}[1]{\begin{center} \parbox{.8\textwidth}{\textbf{Reply YK.} \textcolor{blue}{#1} }\end{center}}
\newcommand{\commentA}[1]{\marginpar{\footnotesize \color{red} {\bf A:} \textsf{\scriptsize #1}}}
\newcommand{\commentC}[1]{\marginpar{\footnotesize \color{red} {\bf C:} \textsf{\scriptsize #1}}}
\newcommand{\commentF}[1]{\marginpar{\footnotesize \color{red} {\bf F:} \textsf{\scriptsize #1}}}
\newcommand{\commentY}[1]{\marginpar{\footnotesize \color{red} {\bf Y:} \textsf{\scriptsize #1}}}
\newcommand{\replyA}[1]{\marginpar{\footnotesize \color{red} {\bf A:} \textsf{\scriptsize #1}}}
\newcommand{\replyC}[1]{\marginpar{\footnotesize \color{red} {\bf C:} \textsf{\scriptsize #1}}}
\newcommand{\replyF}[1]{\marginpar{\footnotesize \color{red} {\bf F:} \textsf{\scriptsize #1}}}
\newcommand{\replyY}[1]{\marginpar{\footnotesize \color{red} {\bf Y:} \textsf{\scriptsize #1}}}
\else
\newcommand{\commentAL}[1]{}
\newcommand{\commentCP}[1]{}
\newcommand{\commentFB}[1]{}
\newcommand{\commentYK}[1]{}
\newcommand{\replyAL}[1]{}
\newcommand{\replyCP}[1]{}
\newcommand{\replyFB}[1]{}
\newcommand{\replyYK}[1]{}
\newcommand{\commentA}[1]{}
\newcommand{\commentC}[1]{}
\newcommand{\commentF}[1]{}
\newcommand{\commentY}[1]{}
\newcommand{\replyA}[1]{}
\newcommand{\replyC}[1]{}
\newcommand{\replyF}[1]{}
\newcommand{\replyY}[1]{}
\fi
\newtheorem{theorem}{Theorem}[section]
\newtheorem{lemma}[theorem]{Lemma}
\newtheorem{proposition}[theorem]{Proposition}
\newtheorem{corollary}[theorem]{Corollary}
\newcommand{\toolname}{HyLeak\xspace}
\newcommand{\LeakWatch}[0]{\mbox{\textsf{LeakWatch}}}
\newcommand{\leakiEst}[0]{\mbox{\textsf{leakiEst}}}

\newcommand{\eqdef}{\ensuremath{\stackrel{\mathrm{def}}{=}}}

\newcommand{\Expect}[1]{\mathbb{E}\!\left[#1\right]}

\newcommand{\Variance}[1]{\mathbb{V}\!\left[#1\right]}

\newcommand{\set}[1]{\{#1\}}

\newcommand{\Cov}[1]{\mathit{Cov}\!\left[#1\right]}
\renewcommand{\O}[0]{\mathcal{O}}

\newcommand{\random}[0]{\textsf{random}}
\newcommand{\randombit}[0]{\textsf{randombit}}
\newcommand{\xor}[0]{\mathbin{\textsf{xor}}}

\newcommand{\Sys}{\mathcal{S}}
\newcommand{\TP}{\mathbf{P}}
\newcommand{\tr}{\mathit{tr}}

\newcommand{\I}{\mathcal{I}}
\newcommand{\J}{\mathcal{J}}
\newcommand{\II}{\mathcal{I}^{\star}}

\newcommand{\pe}{\mathit{pe}}
\newcommand{\X}{\mathcal{X}}
\newcommand{\Y}{\mathcal{Y}}
\newcommand{\Z}{\mathcal{Z}}

\newcommand{\Di}{\mathit{D_i}}
\newcommand{\Dxi}{\mathit{D_{\!Xi}}}
\newcommand{\Dyi}{\mathit{D_{\!Yi}}}

\newcommand{\Dhi}{\hat{\mathit{D_i}}}
\newcommand{\Rh}{\hat{R}}
\newcommand{\fxy}{\mathit{f_{xy}}}
\newcommand{\fy}{\mathit{f_{y}}}
\newcommand{\Ly}[1]{L_{#1\cdot y}}

\newcommand{\Mixy}{M_{ixy}}

\newcommand{\ph}{\hat{P}}

\newcommand{\Qxy}{q_{xy}}

\newcommand{\Ky}[1]{K_{\!#1\cdot y}}
\newcommand{\Kiy}{K_{\!i\cdot y}}
\newcommand{\Kioy}{\overline{K_{\!i\cdot y}}}
\newcommand{\Kjy}{K_{\!j\cdot y}}
\newcommand{\Kjoy}{\overline{K_{\!j\cdot y}}}

\newcommand{\BKxy}{\bm{K}_{\!xy}}
\newcommand{\BKoxy}{\overline{\BKxy}}
\newcommand{\BKy}{\bm{K}_{\!y}}
\newcommand{\BKoy}{\overline{\BKy}}
\newcommand{\Kxy}[1]{K_{\!#1xy}}
\newcommand{\Koxy}[1]{\overline{\Kxy{#1}}}
\newcommand{\Kixy}{\Kxy{i}}
\newcommand{\Kioxy}{\Koxy{i}}
\newcommand{\Kjxy}{\Kxy{j}}
\newcommand{\Kjoxy}{\Koxy{j}}

\newcommand{\Bxy}[1]{B_{#1xy}}
\newcommand{\Bixy}{\Bxy{i}}
\newcommand{\By}[1]{B_{#1\cdot y}}
\newcommand{\Biy}{\By{i}}
\newcommand{\Liy}{L_{i\cdot y}}
\newcommand{\Ljy}{L_{j\cdot y}}
\newcommand{\Lioy}{\overline{L_{i\cdot y}}}
\newcommand{\Ljoy}{\overline{L_{j\cdot y}}}

\newcommand{\D}{\mathcal{D}}

\newcommand{\B}{\mathcal{B}}

\lstdefinelanguage{quail}{
  morekeywords={assign, if, else, fi, return, observable, public, secret, random, int1, int3, randombit, then, goto, const, array, int32, while, do, od, for, in, of},
  morecomment=[l]//
}
\lstdefinestyle{inline}{
    mathescape=false,
    breaklines=true,
    keywordstyle=,            
    keywordstyle=[2],
    extendedchars=true,
    basicstyle=\ttfamily\small
}

\definecolor{pblue}{rgb}{0.13,0.13,1}
\definecolor{pgreen}{rgb}{0,0.5,0}
\definecolor{pred}{rgb}{0.9,0,0}
\definecolor{pgrey}{rgb}{0.46,0.45,0.48}
\tikzstyle{prosumer}=[fill=blue!30,draw=blue,thick]
\tikzstyle{aggregator}=[fill=orange!30,draw=orange]

\allowdisplaybreaks

\title[Formal Aspects of Computing]
      {Hybrid Statistical Estimation of Mutual Information and its Application to Information Flow
}

\author[Fabrizio Biondi, Yusuke Kawamoto, Axel Legay, Louis-Marie Traonouez]
    {Fabrizio Biondi$^1$, Yusuke Kawamoto$^2$, Axel Legay$^3$, Louis-Marie Traonouez$^3$\\
     $^1$CentraleSup\'elec Rennes, France\\
     $^2$AIST, Japan \\
     $^3$Inria, France}

\correspond{Yusuke Kawamoto,
            AIST Tsukuba Central 1, 1-1-1 Umezono, Tsukuba, Ibaraki 305-8560 JAPAN.\\
            Fabrizio Biondi, Inria Rennes, Campus de Beaulieu, 263 Avenue G\'en\'eral Leclerc, 35000 Rennes, France.}

\pubyear{2018}
\pagerange{\pageref{firstpage}--\pageref{lastpage}}

\begin{document}
\label{firstpage}

\makecorrespond

\maketitle

\begin{abstract}
Analysis of a probabilistic system often requires to learn the joint probability distribution of its random variables.
The computation of the exact distribution is usually an exhaustive \emph{precise analysis} on all executions of the system. 
To avoid the high computational cost of such an exhaustive search, \emph{statistical analysis} has been studied to efficiently obtain approximate estimates by analyzing only a small but representative subset of the system's behavior.
In this paper we propose a \emph{hybrid statistical estimation method} that combines precise and statistical analyses to estimate mutual information, Shannon entropy, and conditional entropy, together with their confidence intervals.
We show how to combine the analyses on different components of a discrete system with different accuracy to obtain an estimate for the whole system.
The new method performs weighted statistical analysis with different sample sizes over different components and dynamically finds their optimal sample sizes.
Moreover, it can reduce sample sizes by using prior knowledge about systems and a new \emph{abstraction-then-sampling} technique based on qualitative analysis.
To apply the method to the source code of a system, we show how to decompose the code into components and to determine the analysis method for each component by overviewing the implementation of those techniques in the \toolname{} tool.
We demonstrate with case studies that the new method outperforms the state of the art in quantifying information leakage.
\end{abstract}

\begin{keywords}
Mutual information;
Statistical estimation;
Quantitative information flow;
Hybrid method;
Confidence interval;
Statistical model checking
\end{keywords}

\vspace{-1ex}

\section{Introduction}
\label{sec:intro}

In modeling and analyzing software and hardware systems, the statistical approach is often useful to evaluate quantitative aspects of the behaviors of the systems.
In particular, probabilistic systems with complicated internal structures can be approximately and efficiently modeled and analyzed.
For instance, statistical model checking has widely been used to verify quantitative properties of many kinds of probabilistic systems~\cite{DBLP:conf/rv/LegayDB10}.

The \emph{statistical analysis} of a probabilistic system is usually considered as a black-box testing approach in which the analyst does not require prior knowledge of the internal structure of the system.
The analyst runs the system many times and records the execution traces to construct an approximate model of the system.
Even when the formal specification or precise model of the system is not provided to the analyst, statistical analysis can be directly applied to the system if the analyst can execute the black-box implementation.
Due to this random sampling of the systems, statistical analysis provides only approximate estimates.
However, it can evaluate the precision and accuracy of the analysis for instance by providing the confidence intervals of the estimated values.

One of the important challenges in statistical analysis is to estimate entropy-based properties in probabilistic systems.
For example, statistical methods~\cite{DBLP:conf/tacas/ChatzikokolakisCG10,DBLP:conf/cav/ChothiaKN13,DBLP:conf/csfw/ChothiaKNP13,DBLP:conf/esorics/ChothiaKN14,DBLP:conf/isw/BorealeP14} have been studied for \emph{quantitative information flow analysis}
~\cite{DBLP:journals/tcs/ClarkHM01,DBLP:conf/ccs/KopfB07,DBLP:conf/popl/Malacaria07,DBLP:journals/iandc/ChatzikokolakisPP08},
which estimates an entropy-based property to quantify the leakage of confidential information in a system.
More specifically, the analysis estimates \emph{mutual information} or other properties between two random variables on the secrets and on the observable outputs in the system to measure the amount of information that is inferable about the secret by observing the output. 
The main technical difficulties in the estimation of entropy-based properties are:
\begin{enumerate}
\item to efficiently compute large matrices that represent probability distributions, and
\item to provide a statistical method for correcting the bias of the estimate and computing a confidence interval to evaluate the accuracy of the estimation.
\end{enumerate}

\begin{center}
\begin{figure*}
  \centering
  \subfloat[A probabilistic program composed of 3 components.]{\includegraphics[width=0.45\textwidth]{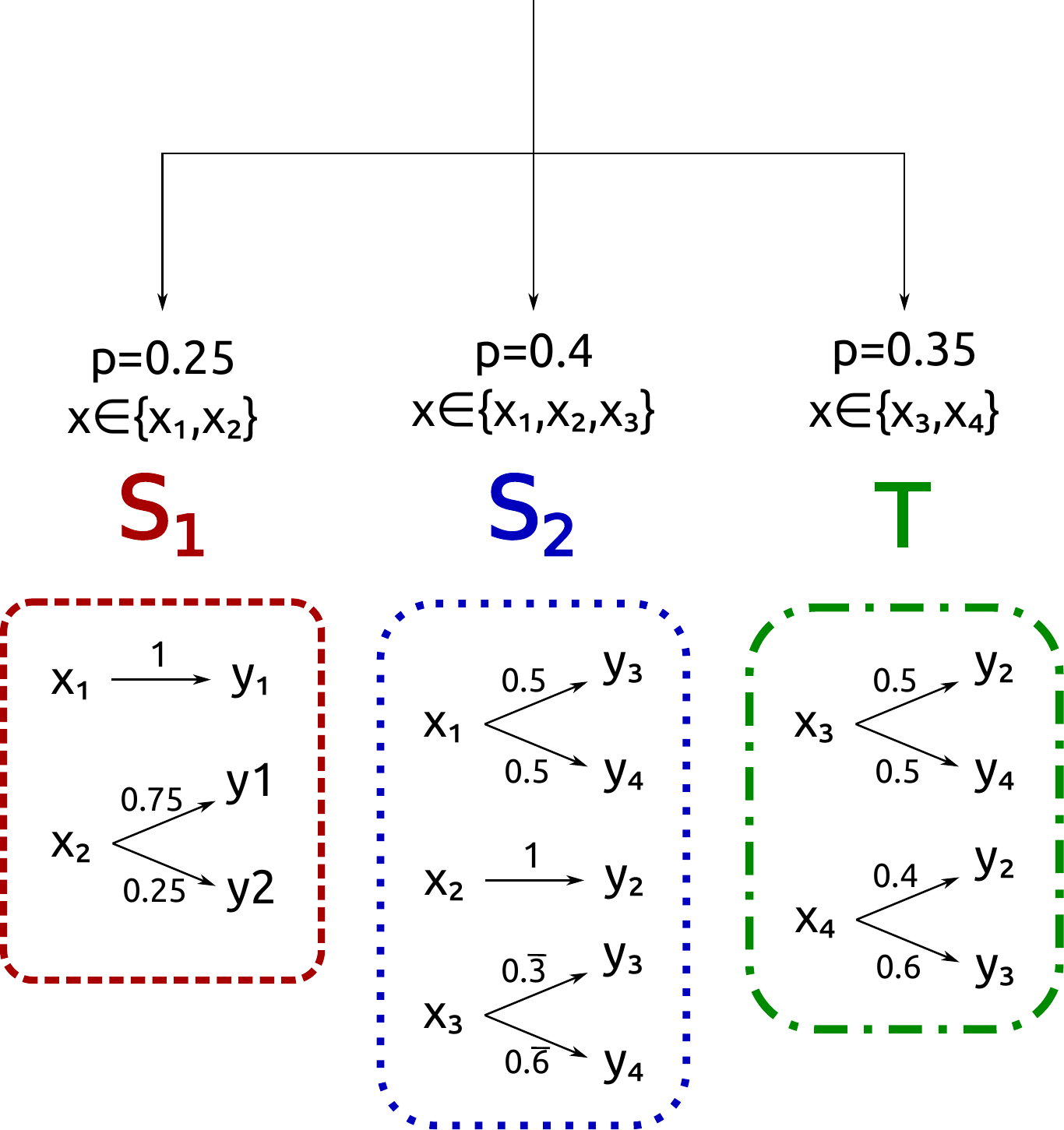}\label{fig:joint_flow}}\hfill
  \subfloat[Joint distribution $P_{XY}$ composed of 3 components.]{\includegraphics[width=0.45\textwidth]{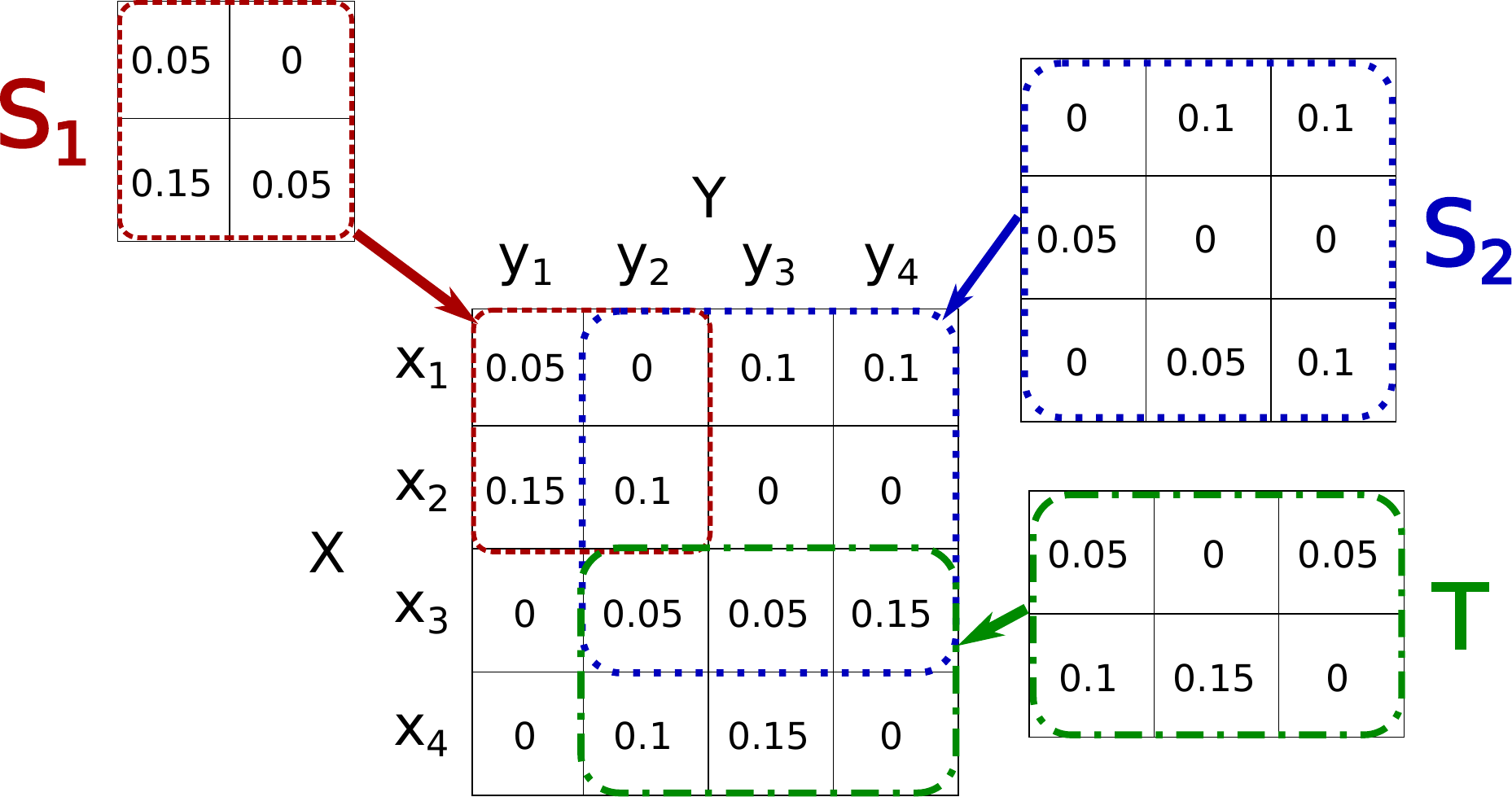}\label{fig:joint}}
  \caption{Decomposition of a probabilistic program with input $x\in\{x_1,x_2,x_3,x_4\}$ and output $y\in\{y_1,y_2,y_3,y_4\}$ into components $S_1$, $S_2$, and $T$ for compositional analysis.
These components have a probability of being visited of $0.25$, $0.4$, and $0.35$ respectively, accounting for both input distribution and internal randomness.
Each component can be visited on specific inputs, and probabilistically produces an output according to the input value. 
For instance, the component $S_2$ is chosen with the probability $0.4$ and given input $x_1$, it outputs $y_3$ or $y_4$ with the equal probability $0.5$.
The division into components can be used to compositionally produce a joint distribution $P_{XY}$ on the input-output behavior of the program.}  \label{fig:joint_prog}
\end{figure*}
\end{center}

To overcome these difficulties, we propose a method for statistically estimating mutual information, one of the most popular entropy-based properties,
in discrete systems.
The new method, called \emph{hybrid statistical estimation method}, integrates black-box statistical analysis and white-box \emph{precise analysis}, exploiting the advantages of both.
More specifically, this method employs some prior knowledge on the system and performs precise analysis (e.g., static analysis of the source code or specification) on some components of the system.
Since precise analysis computes the exact sub-probability distributions of the components, the hybrid method using precise analysis is more accurate than statistical analysis alone.

Moreover, the new method can combine multiple statistical analyses on different components of the system to improve the accuracy and efficiency of the estimation.
This is based on our new theoretical results that extend and generalize previous work~\cite{mod1,brill1,DBLP:conf/tacas/ChatzikokolakisCG10} on purely statistical estimation.
As far as we know this is the first work on a hybrid method for estimating entropy-based properties and their confidence intervals.
Note that our approach assumes that the system has discrete inputs and outputs and behaves deterministically or probabilistically, while it can also be applied to non-deterministic systems if the non-determinism has been resolved by schedulers.


To illustrate the method we propose, Fig.~\ref{fig:joint_prog} presents an example of a probabilistic program (Fig.~\ref{fig:joint_flow}) having input ranging over $X=\{x_1,x_2,x_3,x_4\}$ and output over $Y=\{y_1,y_2,y_3,y_4\}$, built up from three overlapping components $S_1$, $S_2$ and $T$, and the corresponding joint probability distribution $P_{XY}$ (Fig.~\ref{fig:joint}).
To estimate the full joint distribution $P_{XY}$, the analyst separately computes the joint sub-distribution for the component $T$ by precise analysis, estimates those for $S_1$ and $S_2$ by statistical analysis, and then combines these sub-distributions.
Since the statistical analysis is based on the random sampling of execution traces, the empirical sub-distributions for $S_1$ and $S_2$ are different from the true ones, while the sub-distribution for $T$ is exact.
From these approximate and precise sub-distributions, the proposed method can estimate the mutual information for the entire system and evaluate its accuracy by providing a confidence interval. 
Owing to the combination of different kinds of analyses (with possibly different parameters such as sample sizes),
the computation of the bias and confidence interval of the estimate is more complicated than the previous work on statistical analysis.

\subsection{Contributions}
The contributions of this paper are as follows:
\begin{itemize}
\item We propose a new method, called hybrid statistical estimation, that combines statistical and precise analysis on the estimation of mutual information (which can also be applied to Shannon entropy and conditional Shannon entropy).
Specifically, we show theoretical results on compositionally computing the bias and confidence interval of the estimate from multiple results obtained by statistical and precise analysis;
\item We present a weighted statistical analysis method with different sample sizes over different components and a method for adaptively optimizing the sample sizes by evaluating the accuracy and cost of the analysis;
\item We show how to reduce the sample sizes by using prior knowledge about systems, including an abstraction-then-sampling technique based on qualitative analysis.
In particular, we point out that the state-of-the-art statistical analysis tool \LeakWatch{}~\cite{DBLP:conf/esorics/ChothiaKN14} incorrectly computes the bias in estimation using the knowledge of the prior distribution, and explain how the proposed approach fixes this problem;
\item We show that the proposed method can be applied not only to composed systems but also to the source codes of a single system by decomposing it into components and determining the analysis method for each component;
\item We provide a practical implementation of the method in the \toolname tool~\cite{hyleak:www}, and show how the techniques in this paper can be applied to multiple benchmarks; 
\item We evaluate the quality of the estimation in this method, showing that the estimates are more accurate than statistical analysis alone for the same sample size, and that the new method outperforms the state-of-the-art statistical analysis tool \LeakWatch{};
\item We demonstrate the effectiveness of the hybrid method in case studies on the quantification of information leakage.
\end{itemize}

A preliminary version of this paper, without proofs, appeared in~\cite{DBLP:conf/fm/KawamotoBL16}. 
Also a preliminary version of the implementation description (Sections~\ref{sec:functioning} and~\ref{subsec:result-randomwalk}), without details, appeared in the tool paper describing \toolname~\cite{Biondi:2017:ATVA}.
In this paper we add the estimation of Shannon entropy (Propositions~\ref{lem:entropy:expectation},~\ref{lem:entropy:variance} and~\ref{prop:adaptive-sampling:shannon}) and that of conditional entropy (Propositions~\ref{lem:cond-entropy:expectation} and~\ref{lem:cond-entropy:variance}).
We also show the formulas for the adaptive analysis using knowledge of prior distributions (Proposition~\ref{prop:adaptive-sampling-known-prior}) and 
using the abstraction-then-sampling technique (Theorem~\ref{thm:adaptive-sampling-ATS}).
Furthermore, we provide detailed explanation on the implementation in the \toolname{} tool in Section~\ref{sec:functioning}, including how to decompose the source code of a system into components.
We also present more experimental results with details in Section~\ref{sec:shortexamples}.
Finally, we add Appendix~\ref{sec:app:proofs} to present the detailed proofs.

The rest of the paper is structured as follows.
Section~\ref{sec:background} introduces background in quantification of information and compares precise analysis with statistical analysis for the estimation of mutual information.
Section~\ref{sec:overview:hybrid} overviews our new method for estimating mutual information.
Section~\ref{sec:compositional} describes the main results of this paper: the statistical estimation of mutual information for the hybrid method, including the method for optimizing sample sizes for different components.
Section~\ref{sec:estimate-known-system} presents how to reduce sample sizes by using prior knowledge about systems, including the abstraction-then-sampling technique with qualitative analysis.
Section~\ref{sec:adaptive} shows the optimal assignment of samples to components to be samples statistically to improve the accuracy of the estimate.
Section~\ref{sec:functioning} overviews the implementation of the techniques in the \toolname tool, including how to decompose the source code of a system into components and to determine the analysis method for each component.
Section~\ref{sec:evaluation} evaluates the proposed method and illustrates its effectiveness against the state of the art, and 
Section~\ref{sec:conclusion} concludes the paper.
Detailed proofs can be found in Appendix~\ref{sec:app:proofs}.

\subsection{Related Work}\label{sec:related}

The information-theoretical approach to program security dates back to the work of Denning \cite{DBLP:journals/cacm/Denning76}
and Gray \cite{DBLP:conf/sp/Gray91}.
Clark et al.~\cite{DBLP:journals/tcs/ClarkHM01,DBLP:journals/jcs/ClarkHM07} presented techniques to automatically compute mutual information of an imperative language with loops. 
For a deterministic program, leakage can be computed from the equivalence relations on the secret induced by the possible outputs, and such relations can be automatically quantified \cite{DBLP:conf/sp/BackesKR09}.
Under- and over-approximation of leakage based on the observation of some traces have been studied for deterministic programs~\cite{DBLP:conf/pldi/McCamantE08,DBLP:conf/pldi/NewsomeMS09}.
As an approach without replying on information theory McCamant et al.~\cite{DBLP:conf/ndss/KangMPS11} developed tools implementing dynamic quantitative taint analysis techniques for security.

Fremont and Seshia \cite{DBLP:conf/smt/FremontS14} present a polynomial time algorithm to approximate the weight of traces of deterministic programs with possible application to quantitative information leakage.
Progress in randomized program analysis includes a scalable algorithm for uniform generation of sample from a distribution defined as constraints  \cite{DBLP:conf/tacas/ChakrabortyFMSV15,DBLP:conf/cp/ChakrabortyMV13}, with applications to constrained-random program verification.

The statistical approach to quantifying information leakage has been studied since the seminal work by Chatzikokolakis et al.~\cite{DBLP:conf/tacas/ChatzikokolakisCG10}. 
Chothia et al. have developed this approach in tools \leakiEst{}~\cite{DBLP:conf/cav/ChothiaKN13,leakiest:www} 
and \LeakWatch{}~\cite{DBLP:conf/esorics/ChothiaKN14,leakwatch:www}. 
The hybrid statistical method in this paper can be considered as their extension with the inclusion of component weighting and adaptive priors inspired by the 
importance sampling in statistical model checking~\cite{DBLP:conf/tacas/BarbotHP12,DBLP:conf/atva/ClarkeZ11}.
To the best of our knowledge, no prior work has applied weighted statistical analysis to the estimation of mutual information or any other leakage measures.

The idea on combining static and randomized approaches to quantitative information flow was first proposed by K\"opf and Rybalchenko~\cite{DBLP:conf/csfw/KopfR10} 
while our approach takes a different approach relying on statistical estimation to have better precision and accuracy and is general enough to deal with probabilistic systems under various prior information conditions.
In related fields, the hybrid approach combining precise and statistical analysis have been proven to be effective, for instance in concolic analysis~\cite{DBLP:conf/icse/MajumdarS07,DBLP:conf/icse/LiuCFS14}, where it is shown that input generated by hybrid techniques leads to greater code coverage than input generated by both fully random and concolic generation.
After the publication of the preliminary version~\cite{DBLP:conf/fm/KawamotoBL16} of this paper, a few papers on quantitative information flow combining symbolic and statistical approaches have been published.
Malacaria et al.~\cite{Malacaria:18:CSF} present an approach that performs Monte Carlo sampling over symbolic paths while the prior distributions are restricted to be uniform.
Sweet et al.~\cite{Sweet:18:POST} combine abstraction interpretation with sampling and concolic execution for reasoning about Bayes vulnerability.
Unlike our work, these two studies aim at giving only bounds on information leakage and do not use statistical hypothesis testing.

Our tool \toolname{} processes a simple imperative language that is an extension of the language used in the QUAIL tool version 2.0
~\cite{DBLP:conf/spin/BiondiLQ15}. 
The algorithms for precise computation of information leakage used in this paper are based on 
trace analysis \cite{DBLP:journals/tcs/BiondiLMW15}, implemented in the QUAIL 
tool \cite{quail:www,DBLP:conf/cav/BiondiLTW13,DBLP:conf/spin/BiondiLQ15}. 
As remarked above, the QUAIL tool implements only a precise calculation of leakage that examines all executions of programs.
Hence the performance of QUAIL does not scale, especially when the program performs complicated computations that yield a large number of execution traces. The performance of QUAIL as compared to \toolname{} is represented by the ``precise'' analysis approach in Section~\ref{sec:evaluation}.
Since QUAIL does not support the statistical approach or the hybrid approach, it cannot handle large problems that \toolname{} can analyze.

As remarked above, the stochastic simulation techniques implemented in \toolname have also been developed 
in the tools LeakiEst~\cite{DBLP:conf/cav/ChothiaKN13} 
(with its extension~\cite{DBLP:conf/qest/KawamotoCP14}) 
and LeakWatch~\cite{DBLP:conf/csfw/ChothiaKNP13,DBLP:conf/esorics/ChothiaKN14}.
The performance of these tools as compared to \toolname is represented by the ``statistical'' analysis approach in Section~\ref{sec:evaluation}.

The tool Moped-QLeak
~\cite{DBLP:conf/fsttcs/ChadhaMS14} computes the precise information leakage of a program by transforming it into an algebraic decision diagram (ADD).
As noted in~\cite{DBLP:conf/spin/BiondiLQ15}, this technique is efficient when the program under analysis is simple enough to be converted into an ADD, and fails otherwise even when other tools including \toolname can handle it.
In particular, there are simple examples~\cite{DBLP:conf/spin/BiondiLQ15} where Moped-QLeak fails to produce any result but that can be examined by QUAIL and LeakWatch, hence by \toolname.

Many information leakage analysis tools restricted to deterministic input programs have been released, including TEMU~\cite{DBLP:conf/pldi/NewsomeMS09}, 
squifc~\cite{DBLP:conf/ccs/PhanM14}, 
jpf-qif~\cite{DBLP:journals/sigsoft/PhanMTP12}, 
QILURA~\cite{DBLP:conf/spin/PhanMPd14}, 
nsqflow~\cite{DBLP:conf/eurosp/ValEBAH16}, and \textsc{sharpPI}~\cite{DBLP:conf/esorics/Weigl16}. Some of these tools have been proven to scale to programs of thousands of lines written in common languages like C and Java. 
Such tools are not able to compute the Shannon leakage for the scenario of adaptive attacks but only compute the min-capacity of a deterministic program for the scenario of one-try guessing attacks, which give only a coarse upper bound on the Shannon leakage.
More specifically, they compute the logarithm of the number of possible outputs of the deterministic program,
usually by using model counting on a SMT-constraint-based representation of the possible outputs, obtained by analyzing the program. 
Contrary to these tools, \toolname{} can analyze randomized programs\footnote{Some of these tools, like jpf-qif and nsqflow, present case studies on randomized protocols. However, the randomness of the programs is assumed to have the most leaking behavior. E.g., in the Dining Cryptographers this means assuming all coins produce head with probability~1.} and provides a quite precise estimation of the Shannon leakage of the program, not just a coarse upper bound. 
As far as we know, \toolname{} is the most efficient tool that has this greater scope and higher accuracy.

\section{Background}
\label{sec:background}

In this section we introduce the basic concepts used in the paper.
We first introduce some notions in information theory to quantify the amount of some information in probabilistic systems.
Then we compare two previous analysis approaches to quantifying information: precise analysis and statistical analysis.

\subsection{Quantification of Information}
\label{subsec:information-theory}

In this section we introduce some background on information theory, which we use to quantify the amount of information in a probabilistic system.
Hereafter we write $X$ and $Y$ to denote two random variables, and $\X$ and $\Y$ to denote the sets of all possible values of $X$ and $Y$, respectively.
We denote the number of elements of a set $\mathcal{A}$ by $\#\mathcal{A}$.
Given a random variable $A$ we denote by $\Expect{A}$ or by $\overline{A}$ the expected value of $A$, and by $\Variance{A}$ the variance of $A$, i.e., $\Variance{A} = \Expect{ (A - \Expect{A})^2 }$.
The logarithms in this paper are to the base $2$ and we often abbreviate $\log_2$ to $\log$.

\subsubsection{Channels}
In information theory, a \emph{channel} models the input-output relation of a system as a conditional probability distribution of outputs given inputs.
This model has also been used 
to formalize information leakage in a system that processes confidential data:
\emph{inputs} and \emph{outputs} of a channel are respectively regarded as \emph{secrets} and \emph{observables} in the system and the channel represents relationships between the secrets and observables.

A \emph{discrete channel} is a triple $(\mathcal{X},\mathcal{Y},C)$ where $\mathcal{X}$ and $\mathcal{Y}$ are two finite sets of discrete input and output values respectively and $C$ is an $\#\mathcal{X}\times\#\mathcal{Y}$ matrix where each element $C[x, y]$ represents the conditional probability of an output $y$ given an input $x$; i.e., for each $x\in\X$,\, $\sum_{y\in\mathcal{Y}}C[x,y]=1$ and $0\leq C[x,y]\leq 1$ for all $y\in\Y$.

A \emph{prior} is a probability distribution on input values $\X$.
Given a prior $P_X$ over $\X$ and a channel $C$ from $\X$ to $\Y$, the \emph{joint probability distribution} $  P_{XY}$ of $X$ and $Y$ is defined by: $P_{XY}[x, y] = P_X[x] C[x, y]$ for each $x\in\X$ and $y\in\Y$.

\subsubsection{Shannon Entropy}

We recall some information-theoretic measures as follows.
Given a prior $P_X$ on input $X$, the \emph{prior uncertainty} (before observing the system's output $Y$) is defined as
\begin{dmath*}
H(X) = \displaystyle - \sum_{x\in\X} P_X[x] \log_2 P_X[x]
\end{dmath*}
while the \emph{posterior uncertainty} (after observing the system's output $Y$) is defined as
\begin{dmath*}
H(X|Y) = \displaystyle - \sum_{y\in\Y^+} P_Y[y] \sum_{x\in\X} P_{X|Y}[x|y] \log_2 P_{X|Y}[x|y],
\end{dmath*}
where $P_Y$ is the probability distribution on the output $Y$,\, 
$\Y^+$ is the set of outputs in $\Y$ with non-zero probabilities,
and $P_{X|Y}$ is the conditional probability distribution of $X$ given $Y$:
\begin{align*}
P_Y[y]~ =\sum_{x'\in\X} P_{XY}[x', y] \hspace{4ex}
P_{X|Y}[x|y] = \frac{P_{XY}[x,y]}{P_Y[y]} \hspace{2ex} \mbox{ if } P_Y[y] \neq 0.
\end{align*}
$H(X|Y)$ is also called the \emph{conditional entropy} of $X$ given $Y$.

\subsubsection{Mutual Information}
The amount of information gained about a random variable $X$ by knowing a random variable $Y$ is defined as the difference between the uncertainty about $X$ before and after observing $Y$. 
The \emph{mutual information} $I(X; Y)$ between $X$ and $Y$ is one of the most popular measures to quantify the amount of information on $X$ gained by $Y$:
\begin{align*}
I(X; Y) &= 
\displaystyle 
\sum_{x\in {\cal X}, y\in {\cal Y}}\hspace{-1ex} P_{XY}[x, y] \log_2 \left({\displaystyle \frac{P_{XY}[x, y]}{P_X[x] P_Y[y]} }\right)
\end{align*}
where $P_Y$ is the marginal probability distribution defined as $P_Y[y] = \sum_{x\in\X} P_{XY}[x,y]$.

In the security scenario, information-theoretical measures quantify the amount of secret information leaked against some particular attacker: 
the mutual information between two random variables $X$ on the secrets and $Y$ on the observables in a system measures the information that is inferable about the secret by knowing the observable.
In this scenario mutual information, or Shannon leakage, assumes an attacker that can ask binary questions on the secret's value after observing the system while min-entropy leakage \cite{DBLP:conf/fossacs/Smith09} considers an attacker that has only one attempt to guess the secret's value.

Mutual information has been employed in many other applications including Bayesian networks \cite{Jensen:1996:IBN:546868}, telecommunications \cite{Gallager:1968:ITR:578869}, pattern recognition \cite{escolano:2009:ITC:1629682}, machine learning \cite{MacKay:2002:ITI:971143}, quantum physics \cite{Wilde:2013:QIT:2505455}, and biology \cite{citeulike:9970407}.
In this work we focus on mutual information and its application to the above security scenario.

\subsection{Computing Mutual Information in Probabilistic Systems}
\label{subsec:compare:approaches}
In this section we present two previous approaches to computing mutual information in probabilistic systems in the context of quantitative information flow.
Then we compare the two approaches to discuss their advantages and disadvantages.

In the rest of the paper
a \emph{probabilistic system} $\Sys$ is defined as a finite set of \emph{execution traces} such that each trace $\tr$ records the values of all variables in $\Sys$ and is associated with a probability $\TP_{\!\Sys}[\tr]$.
Note that $\Sys$ does not have non-deterministic transitions.
For the sake of generality we do not assume any specific constraints at this moment.

The main computational difficulty in calculating the mutual information $I(X;Y)$ between input $X$ and output $Y$ lies in the computation of the joint probability distribution $P_{XY}$ of $X$ and $Y$, 
especially when the system consists of a large number of execution traces and when the distribution $P_{XY}$ is represented as a large data structure.
In previous work this computation has been performed either by the \emph{precise} approach using program analysis techniques or by the \emph{statistical} approach using random sampling and statistics.

\subsubsection{Precise Analysis}

Precise analysis consists of analyzing all the execution traces of a system and determining for each trace $\tr$, the input $x$, output $y$, and probability $\TP_{\!\Sys}[\tr]$ by concretely or symbolically executing the system. 
The precise analysis approach in this paper follows the depth-first trace exploration technique presented by Biondi et al.~\cite{DBLP:conf/spin/BiondiLQ15}.

To obtain the exact joint probability $P_{XY}[x,y]$ for each $x\in\X$ and $y\in\Y$ in a system $\Sys$, we sum the probabilities of all execution traces of $\Sys$ that have input $x$ and output $y$, i.e.,
\begin{align*}
P_{XY}[x,y] =
\sum \Bigl\{\, \TP_{\!\Sys}[\tr] ~~\Big|~~ \mbox{$\tr\in\Sys$ has input $x$ and output $y$} \,\Bigr\}
\end{align*}
where $\TP_{\!\Sys}$ is the probability distribution over the set of all traces in $\Sys$.
This means the computation time depends on the number of traces in the system.
If the system has a very large number of traces, it is intractable for the analyst to precisely compute the joint distribution and consequently the mutual information.

In~\cite{tcs:YasuokaT14} the calculation of mutual information is shown to be computationally expensive. This computational difficulty comes from the fact that entropy-based properties are hyperproperties~\cite{ClarksonS10jcs} that are defined using all execution traces of the system and therefore cannot be verified on each single trace.
For example, when we investigate the information leakage in a system, it is insufficient to check the leakage separately for each component of the system, because the attacker may derive sensitive information by combining the outputs of different components.
More generally, the computation of entropy-based properties (such as the amount of leaked information) is not compositional, in the sense that an entropy-based property of a system is not the (weighted) sum of those of the components.

For this reason, it is inherently difficult to na\"ively combine analyses of different components of a system to compute entropy-based properties.
In fact, previous studies on the compositional approach in quantitative information flow analysis have faced certain difficulties in obtaining useful bounds on information leakage~\cite{DBLP:conf/csfw/BartheK11,DBLP:journals/iandc/EspinozaS13,DBLP:journals/corr/KawamotoG15,Kawamoto:17:LMCS}.

\subsubsection{Statistical Analysis}
Due to the complexity of precise analysis, some previous studies have focused on computing approximate values of entropy-based measures.
One of the common approaches is \emph{statistical analysis} based on Monte Carlo methods, in which approximate values are computed from repeated random sampling and their accuracy is evaluated using statistics.
Previous work on quantitative information flow has used statistical analysis to estimate mutual information~\cite{DBLP:conf/tacas/ChatzikokolakisCG10,mod1,brill1}, channel capacity~\cite{DBLP:conf/tacas/ChatzikokolakisCG10,DBLP:conf/isw/BorealeP14} and min-entropy leakage~\cite{DBLP:conf/esorics/ChothiaKN14,ChothiaKawamoto2014}.

In the statistical estimation of mutual information between two random variables $X$ and $Y$ in a probabilistic system, the analyst executes the system many times and collects the execution traces, each of which has a pair of values $(x,y)\in\X\times\Y$ corresponding to the input and output of the trace.
This set of execution traces is used to estimate the empirical joint distribution $\ph_{XY}$ of $X$ and $Y$ and then to estimate the mutual information $\hat{I}(X; Y)$.

Note that the empirical distribution $\ph_{XY}$ is different from the true distribution $P_{XY}$ and thus the estimated mutual information $\hat{I}(X; Y)$ is different from the true value $I(X; Y)$.
In fact, it is known that entropy-based measures such as mutual information and min-entropy leakage have some bias and variance that depends on the number of collected traces, the matrix size and other factors.
However, results on statistics allow us to correct the bias of the estimate and to compute the variance (and the 95\% confidence interval).
This way we can guarantee the quality of the estimation, which differentiates the statistical approach from the testing approach.

\subsubsection{Comparing the Two Analysis Methods}
The cost of the statistical analysis is proportional to the size $\#\X \times \#\Y$ of the joint distribution matrix (strictly speaking, to the number of non-zero elements in the matrix).
Therefore, this method is significantly more 
%
\begin{table}
\centering
\begin{tabular}{l l  l  }
 &{\bf Precise analysis} & {\bf Statistical analysis}
\\[0.8ex]\bhline{0.3mm}\\[-0.8ex]
{\bf Type} & White box & Black/gray box
\\[0.8ex]\cline{1-3}\\[-0.8ex]
{\bf Analyzes} & Source code & Implementation
\\[0.8ex]\cline{1-3}\\[-0.8ex]
{\bf Produces} & Exact value & Estimate \& accuracy evaluation
\\[0.8ex]\cline{1-3}\\[-0.8ex]
{\bf Reduces costs by} & Random sampling & Knowledge of code \& abstraction
\\[0.8ex]\cline{1-3}\\[-0.8ex]
{\bf Impractical for} & Large number of traces & Large channel matrices
\\[0.8ex]\cline{1-3}
\end{tabular}
\caption{Comparison of the precise and statistical analysis methods.}\label{tab:comparison}
\end{table}
%
efficient than precise analysis if the matrix is relatively small and the number of all traces is very large (for instance because the system's internal variables have a large range).

On the other hand, if the matrix is very large, the number of executions needs to be very large to obtain a reliable and small confidence interval.
In particular, for a small sample size, statistical analysis does not detect rare events, i.e., traces with a low probability that affect the result.
Therefore the precise analysis is significantly more efficient than statistical analysis if the number of all traces is relatively small and the matrix is relatively large (for instance because the system's internal variables have a small range).

The main differences between precise analysis and statistical analysis are summarized in Table~\ref{tab:comparison}.

\section{Overview of the Hybrid Statistical Estimation Method}
\label{sec:overview:hybrid}

In this section we overview a new method for estimating the mutual information between two random variables $X$ (over the inputs $\X$) and $Y$ (over the outputs $\Y$) in a system.
The method, we call \emph{hybrid statistical estimation}, integrates both precise and statistical analyses to overcome the limitations on those previous approaches (explained in Section~\ref{subsec:compare:approaches}).

In our hybrid analysis method, we first decompose a given probabilistic system $\Sys$ into mutually disjoint \emph{components}, which we will define below, and then apply different types of analysis (with possibly different parameters) on different components of the system.
More specifically, for each component, our hybrid method chooses the faster analysis between the precise and statistical analyses.
Hence the hybrid analysis of the whole system is faster than the precise analysis alone and than the statistical analysis alone, while it gives more accurate estimates than the statistical analysis alone, as shown experimentally in Section~\ref{sec:evaluation}.

To introduce the notion of components we recall that in Section~\ref{subsec:compare:approaches} a probabilistic system $\Sys$ is defined as the set of all execution traces such that each trace $\tr$ is associated with probability $\TP[\tr]$ \footnote{Note that this work considers only probabilistic systems without non-deterministic transitions.}.
Formally, a \emph{decomposition} $\alpha$ of $\Sys$ is defined 
a collection of  mutually disjoint non-empty subsets of $\Sys$:
$\emptyset \not\in \alpha$,
$\Sys = \bigcup_{S_i \in\alpha} S_i$, and 
for any $S_i, S_j\in \alpha$,\, $S_i \neq S_j$ implies $S_i \cap S_j \neq \emptyset$.
Then each element of $\alpha$ is called a \emph{component}. 
In this sense, components are a partition of the execution traces of $\Sys$.
When $\Sys$ is executed, only one of $\Sys$'s components is executed, since the components are mutually disjoint.
Hence $S_i$ is chosen to be executed with the probability $\TP[S_i]$.

In decomposing a system we roughly investigate the characteristics of each component's behaviour to choose a faster analysis method for each component. Note that information about a component like its number of traces and the size of its joint sub-distribution matrix can be estimated heuristically before computing the matrix itself. This will be explained in Section~\ref{sec:functioning}; before that section this information is assumed to be available. The choice of the analysis method is as follows:

\begin{itemize}
\item If a component's behaviour is deterministic, we perform a precise analysis on it.
\item If a component's behaviour is described as a joint sub-distribution matrix over \emph{small}\footnote{Relatively to the number of all execution traces of the component.} subsets of $\X$ and $\Y$, then we perform a statistical analysis on the component.
\item If a component's behaviour is described as a matrix over \emph{large}\footnotemark[3] subsets of $\X$ and $\Y$, then we perform a precise analysis on the component.
\item By combining the analysis results on all components, we compute the estimated value of mutual information and its variance (and confidence interval).
See Section~\ref{sec:compositional} for details.
\item By incorporating information from \emph{qualitative} information flow analysis, the analyst may obtain partial knowledge on components and be able to reduce the sample sizes.
See Section~\ref{sec:estimate-known-system} for details.
\end{itemize}
See Section~\ref{sec:functioning} for the details on how to decompose a system.

One of the main advantages of hybrid statistical estimation is that we guarantee the quality of the outcome by removing its bias and providing its variance (and confidence interval) even though different kinds of analysis with different parameters (such as sample sizes) are combined together.

Another advantage is the compositionality in estimating bias and variance.
Since the sampling of execution traces is performed independently for each component, we obtain that the bias and variance of mutual information can be computed in a compositional way, i.e., the bias/variance for the entire system is the sum of those for the components.
This compositionality enables us to find optimal sample sizes for the different components that maximize the accuracy of the estimation (i.e., minimize the variance) given a fixed total sample size for the entire system.
On the other hand, the computation of mutual information itself is not compositional~\cite{Kawamoto:17:LMCS}: it requires calculating the \emph{full} joint probability distribution of the system by summing the joint sub-distributions of all components of the system.

Finally, note that these results can be applied to the estimation of Shannon entropy (Section~\ref{sec:other-measures}) and conditional Shannon entropy (Section~\ref{subsec:Conditional-entropy}) as special cases.
The overview of all results is summarized in Table~\ref{table:all-results}.

\begin{table}
\begin{small}
 \begin{tabular}{|l|l|l|c|c|}
  & 
    & {\bf Bias correction}
    & {\bf Variance computation}
    & {\bf Adaptive sampling}
 \\[0.8ex] \bhline{0.3mm}\\[-0.8ex]
 ~{\bf No knowledge}
 	& Mutual information
 	& Theorem~\ref{lem:MI:expectation:general}
	& Theorem~\ref{lem:MI:variance:general}
	& Theorem~\ref{thm:adaptive-sampling}
 \\[0.8ex] \cline{2-5} ~{\bf on the system} \\[-0.8ex] 
 	& Shannon entropy
 	& Proposition~\ref{lem:entropy:expectation}
	& Proposition~\ref{lem:entropy:variance}
	& Proposition~\ref{prop:adaptive-sampling:shannon}
 \\[0.8ex] \cline{1-5}\\[-0.8ex] 
 ~{\bf Knowledge}
    & Mutual information
 	& Proposition~\ref{lem:MI:expectation:KnownPrior}
    & Proposition~\ref{lem:MI:variance:KnownPrior}
    & Proposition~\ref{prop:adaptive-sampling-known-prior}
 \\[0.8ex] \cline{2-5} ~{\bf on the prior} \\[-0.8ex]
 	& Conditional entropy
    & Proposition~\ref{lem:cond-entropy:expectation}
	& Proposition~\ref{lem:cond-entropy:variance}
	& ---
 \\[0.8ex] \cline{1-5}\\[-0.8ex] 
 ~{\bf Abstraction-then-sampling}
 	& Mutual information
 	& Theorem~\ref{thm:MI:expectation:symbolic}
	& Theorem~\ref{thm:MI:variance:symbolic}
	& Theorem~\ref{thm:adaptive-sampling-ATS}
 \\[0.8ex] \cline{1-5}
 \end{tabular}
 \caption{Our results on the hybrid method.}
 \label{table:all-results}
\end{small}
\end{table}

\section{Hybrid Method for Statistical Estimation of Mutual Information}
\label{sec:compositional}

In this section we present a method for estimating the mutual information  between two random variables $X$ (over the inputs $\X$) and $Y$ (over the outputs $\Y$) in a system, and for evaluating the precision and accuracy of the estimation.

We consider a probabilistic system $\Sys$ that consists of $(m+k)$ components $S_1$, $S_2$, $\ldots\,$, $S_m$ and $T_1$, $T_2$, $\ldots\,$, $T_k$ each executed with probabilities $\theta_1$, $\theta_2$, $\ldots\,$, $\theta_m$ and $\xi_1$, $\xi_2$, $\ldots\,$, $\xi_k$, i.e., when $\Sys$ is executed,\, $S_i$ is executed with the probability $\theta_i$ and $T_j$ with the probability $\xi_j$.
Let $\I = \{ 1, 2, \ldots , m \}$ and $\J = \{ 1, 2, \ldots , k \}$, one of which can be empty.
Then the probabilities of all components sum up to $1$, i.e.,
$\sum_{i\in\I} \theta_i +
\sum_{j\in\J} \xi_j = 1$.
We assume that the analyst is able to compute these probabilities by precise analysis. 
In the example in Fig.~\ref{fig:joint_prog} the probabilities of the components $S_1$, $S_2$, and $T$ are explicitly given as $\theta_1=0.25$, $\theta_2=0.45$, and $\xi_1=0.3$. 
However, in general they would be computed by analyzing the behavior of the system before the system executes the three components. More details about how to obtain this in practice are provided when discussing implementation in Section~\ref{sec:functioning}.

Once the system is decomposed into components, each component is analyzed either by precise analysis or by statistical analysis.
We assume that the analyst can run the component $S_i$ for each $i\in\I$ to record a certain number of $S_i$'s execution traces, and precisely analyze the components $T_j$ for $j\in\J$ to record a certain symbolic representation of $T_j$'s all execution traces, e.g., by static analysis of the source code (or of a specification that the code is known to satisfy). 
In the example in Fig.~\ref{fig:joint_prog}, this means that the components $S_1$ and $S_2$ will be analyzed statistically producing an approximation of their joint distributions, while the component $T$ will be analyzed precisely obtaining its exact joint distribution. The two estimates and one precise joint distributions will be composed to obtain a joint distribution estimate for the whole system, as illustrated in Fig.~\ref{fig:joint}.

In the rest of this section we present a method for computing the joint probability distribution $\ph_{XY}$ (Section~\ref{sec:overview:hybrid}), for estimating the mutual information $\hat{I}(X; Y)$ (Section~\ref{subsec:estimate-MI-bias}), and for evaluating the accuracy of the estimation (Section~\ref{subsec:eval-accuracy}).
Then we show the application of our hybrid method to Shannon entropy estimation (Section~\ref{sec:other-measures}).


In the estimation of mutual information between the two random variables $X$ and $Y$ in the system $\Sys$, we need to estimate the joint probability distribution $P_{XY}$ of $X$ and $Y$.

In our approach this is obtained by combining the joint \emph{sub-probability distributions} of $X$ and $Y$ for all the components $S_i$'s and $T_j$'s.
More specifically, let $R_i$ and $Q_j$ be the joint sub-distributions of $X$ and $Y$ for the components $S_i$'s and $T_j$'s respectively.
Then the joint (full) distribution $P_{XY}$ for the whole system $\Sys$ is defined by:
\begin{dmath*}
P_{XY}[x, y] \eqdef
\sum_{i\in\I} R_i[x,y] +
\sum_{j\in\J} Q_j[x,y]
\end{dmath*}
for $x\in\X$ and $y\in\Y$.
Note that for each $i\in\I$ and $j\in\J$, the sums of all probabilities in the sub-distribution $R_i$ and in $Q_j$ respectively equal the probabilities $\theta_i$ (of executing $S_i$) and $\xi_j$ (of executing $T_j$).

To estimate the joint distribution $P_{XY}$ the analyst computes 
\begin{itemize}
\item for each $j\in\J$, the \emph{exact} sub-distribution $Q_j$ for the component $T_j$ by precise analysis on $T_j$, and
\item for each $i\in\I$, the \emph{empirical} sub-distribution $\hat{R}_i$ for $S_i$ from a set of traces obtained by executing $S_i$ a certain number $n_i$ of times.
\end{itemize}

More specifically, the empirical sub-distribution $\Rh_i$ is constructed as follows.
When the component $S_i$ is executed $n_i$ times, let $\Kixy$ be the number of traces that have input $x\in\X$ and output $y\in\Y$.
Then $n_i = \sum_{x\in\X,y\in\Y} \Kixy$.
From these numbers $\Kixy$ of traces we compute the empirical joint (full) distribution $\Dhi$ of $X$ and $Y$ by:
\begin{dmath*}
\Dhi[x,y] \eqdef \frac{\Kixy}{n_i}\;.
\end{dmath*}
Since $S_i$ is executed with probability $\theta_i$,
the sub-distribution $\Rh_i$ is given by
$\Rh_i[x,y] \eqdef \theta_i \Dhi[x,y] = \frac{\theta_i \Kixy}{n_i}$.

Then the analyst sums up these sub-distributions to obtain the joint  distribution $\ph_{XY}$ for the whole system $\Sys$:
\begin{align*}
\ph_{XY}[x,y] \eqdef
\sum_{i\in\I} \Rh_i[x,y] + \sum_{j\in\J} Q_j[x,y] =
\sum_{i\in\I} \frac{\theta_i \Kixy}{n_i} + \sum_{j\in\J} Q_j[x,y]
\texttt{.}
\end{align*}
Note that $R_i$ and $Q_j$ may have different matrix sizes and cover different parts of the joint distribution matrix $\ph_{XY}$, so they may have to be appropriately padded with zeroes for the summation.

\subsection{Estimation of Mutual Information and Correction of its Bias}
\label{subsec:estimate-MI-bias}
In this section we present our new method for estimating mutual information and for correcting its bias.
For each component $S_i$ let $\Di$ be the joint (full) distribution of $X$ and $Y$ obtained by normalizing $R_i$:
$\Di[x,y] = \frac{R_i[x,y]}{\theta_i}$.
Let
$\Dxi[x] = \textstyle\sum_{y\in\Y} \Di[x,y]$,
$\Dyi[y] = \sum_{x\in\X} \Di[x,y]$, and
$\D = \{ (x, y)\in\X\times\Y : P_{XY}[x,y] \neq 0 \}$.

Using the estimated joint distribution $\ph_{XY}$ we can compute the mutual information estimate $\hat{I}(X; Y)$.
Note that the mutual information for the whole system is smaller than (or equals) the weighted sum of those for the components, because of its convexity w.r.t. the channel matrix.
Therefore it cannot be computed compositionally from those of the components,
i.e., it is necessary to compute the joint distribution matrix $\hat{P}_{XY}$ for the whole system.

Since $\hat{I}(X; Y)$ is obtained from a limited number of traces, it has bias, i.e., its expected value $\Expect{\hat{I}(X; Y)}$ is different from the true value $I(X; Y)$.
The bias $\Expect{\hat{I}(X; Y)} - I(X; Y)$ in the estimation is quantified as follows.
%
\begin{restatable}[Mean of estimated mutual information]{theorem}{resMeanMI}\label{lem:MI:expectation:general}
The expected value $\Expect{\hat{I}(X; Y)}$ of the estimated mutual information is given by:
\begin{dmath*}
\displaystyle
\Expect{\hat{I}(X; Y)} = 
I(X; Y)+
\sum_{i\in\I} \frac{\theta_i^2}{2n_i}
\Bigl(
\sum_{\,(x,y)\in\D\hspace{-1ex}}\hspace{-0.2ex}
\varphi_{ixy}
- \sum_{\,x\in\X^+\hspace{-1ex}}\hspace{-0.2ex}
\varphi_{ix}
- \sum_{\,y\in\Y^+\hspace{-1ex}}\hspace{-0.2ex}
\varphi_{iy}
\Bigr)
+ \O(n_i^{-2})
\end{dmath*}
where
$\varphi_{ixy} \allowbreak= {\textstyle \frac{\Di[x,y] - \Di[x,y]^2}{P_{XY}[x, y]}}$,
$\varphi_{ix} \allowbreak= {\textstyle \frac{\Dxi[x] - \Dxi[x]^2}{P_X[x]}}$ and
$\varphi_{iy} \allowbreak= {\textstyle \frac{\Dyi[y] - \Dyi[y]^2}{P_Y[y]}}$.
\end{restatable}

\begin{proof}[Proof sketch.]
Here we present only the basic idea. Appendices~\ref{subsec:proof:mean:ATS} and~\ref{subsec:proof:mean:standard} present a proof of this theorem by showing a more general claim, i.e., Theorem~\ref{thm:MI:expectation:symbolic} in Section~\ref{subsec:known-subsystem}.

By properties of mutual information and Shannon entropy, we have:
\begin{align*}
\Expect{\hat{I}(X; Y)} - I(X; Y)
&=
\Expect{\hat{H}(X) + \hat{H}(Y) - \hat{H}(X, Y)} -
\Bigl( H(X) + H(Y) - H(X, Y) \Bigr)
\\&= 
\Bigl( \Expect{\hat{H}(X)} - H(X) \Bigr) + 
\Bigl( \Expect{\hat{H}(Y)} - H(Y) \Bigr) - 
\Bigl( \Expect{\hat{H}(X, Y)} - H(X, Y) \Bigr)
\texttt{.}
\end{align*}
Hence it is sufficient to calculate the bias in $\hat{H}(X)$, $\hat{H}(Y)$, and $\hat{H}(X, Y)$, respectively.

We calculate the bias in $\hat{H}(X, Y)$ as follows.
Let $\fxy(\Kxy{1}, \Kxy{2}, \ldots, \Kxy{m})$ be the $m$-ary function defined by:
\begin{align*}
\fxy(\Kxy{1}, \Kxy{2}, \ldots, \Kxy{m}) =
\biggl(
\sum_{i\in\I} \frac{\theta_i \Kixy}{n_i} +
\sum_{j\in\J} Q_j[x,y] \biggr)
\log\biggl(
\sum_{i\in\I} \frac{\theta_i \Kixy}{n_i} +
\sum_{j\in\J} Q_j[x,y] \biggr),
\end{align*}
which equals $\ph_{XY}[x,y]\log\ph_{XY}[x,y]$.
Let $\BKxy = (\Kxy{1}, \Kxy{2}, \ldots, \Kxy{m})$.
Then the empirical joint entropy is:
\begin{align*}
\hat{H}(X, Y) =
-\hspace{-1.0ex}\sum_{\,(x,y)\in\D\hspace{-1ex}} \ph_{XY}[x,y] \log \ph_{XY}[x,y]
=
-\hspace{-1.0ex}\sum_{\,(x,y)\in\D\hspace{-1ex}} \fxy(\BKxy).
\end{align*}
Let $\Kioxy = \Expect{\Kixy}$ for each $i\in\I$ and $\BKoxy = \Expect{\BKxy}$.
By the Taylor expansion of $\fxy(\BKxy)$ (w.r.t. the multiple dependent variables $\BKxy$) at $\BKoxy$, we have:
\begin{align*}
\hspace{-3.2ex}
\fxy(\BKxy) = 
\fxy(\BKoxy) +\!
\sum_{i\in\I}\!
{\textstyle \frac{\partial\fxy(\BKoxy)}{\partial \Kixy}}
(\Kixy - \Kioxy) +\!
\frac{1}{2}\!
\sum_{i,j\in\I}\!
{\textstyle \frac{\partial^2\fxy(\BKoxy)}{\partial \Kixy \partial \Kjxy}}
(\Kixy - \Kioxy) (\Kjxy - \Kjoxy)
\!+\!
\sum_{i\in\I}\!\O(\Kixy^3)
{.}
\end{align*}
We use the following properties:
\begin{itemize}
\item $\Expect{\Kixy - \Kioxy} = 0$, which is immediate from $\Kioxy = \Expect{\Kixy}$.
\item $\Expect{ (\Kixy - \Kioxy) (\Kjxy - \Kjoxy) } = 0$ if $i \neq j$, because $\Kixy$ and $\Kjxy$ are independent.
\item $\Expect{ (\Kixy - \Kioxy)^2 } = \Variance{\Kixy} = n_i D_i[x,y] \bigl( 1 - D_i[x,y] \bigr)$.
\end{itemize}
Then 
\begin{align*}
\Expect{\hat{H}(X, Y)} &=
-\hspace{-1.5ex}
\sum_{(x, y)\in\D}\hspace{-0.5ex}
\Expect{\,\fxy(\BKxy)\,}
\\ &=
-\hspace{-1.5ex}
\sum_{(x, y)\in\D}\hspace{-0.5ex}
\Bigl(
\fxy(\BKoxy) +
\frac{1}{2}
\sum_{i\in\I}
{\textstyle \frac{\partial^2\fxy(\BKoxy)}{\partial \Kixy \partial \Kixy}}\,
\Expect{ (\Kixy - \Kioxy)^2 }
+ \O(\Kixy^3)
\Bigr)
\\ &=\!
-\hspace{-1.5ex}
\sum_{(x, y)\in\D}\hspace{-0.5ex}
\Bigl(
\fxy(\BKoxy) +
\frac{1}{2}
\sum_{i\in\I}
{\textstyle \frac{\theta_i^2}{n_i^2 P_{XY}[x, y]}}\,
n_i D_i[x,y] \bigl( 1 - D_i[x,y] \bigr)
+\O(n_i^{-2})
\Bigr)
\\ &=\!
H(X, Y) -
\sum_{i\in\I}
\frac{\theta_i^2}{2n_i}\!
\sum_{(x, y)\in\D}\hspace{-0.5ex}
\varphi_{ixy}
+\O(n_i^{-2})
\texttt{,}
\end{align*}
where the derivation of the equalities is detailed in Appendix~\ref{sec:app:proofs}.
Hence the bias in estimating $H(X,Y)$ is given by:
\begin{align*}
\Expect{\hat{H}(X, Y)} - H(X, Y) =
- \sum_{i\in\I}
\frac{\theta_i^2}{2n_i}\!
\sum_{(x, y)\in\D}\hspace{-0.5ex}
\varphi_{ixy}
+\O(n_i^{-2})
\texttt{.}
\end{align*}
Analogously, we can calculate the bias in $\hat{H}(X)$ and $\hat{H}(Y)$ to derive the theorem.
See Appendices~\ref{subsec:proof:mean:ATS} and~\ref{subsec:proof:mean:standard} for the details.
\end{proof}

Since the higher-order terms in the formula are negligible when the sample sizes $n_i$ are large enough, we use the following as the \emph{point estimate} of the mutual information:
\begin{dmath*}
\pe = 
\hat{I}(X; Y) -
\displaystyle
\sum_{i\in\I}
\frac{\theta_i^2}{2n_i}\hspace{-0.3ex}
\Bigl(
\sum_{\,(x,y)\in\D\hspace{-1.5ex}}\hspace{-1ex}
\widehat{\varphi}_{ixy}
-\hspace{-0.5ex}
\sum_{\,x\in\X^+\hspace{-1ex}}\hspace{-0.5ex}
\widehat{\varphi}_{ix}
-\hspace{-0.5ex}
\sum_{\,y\in\Y^+\hspace{-1ex}}\hspace{-0.5ex}
\widehat{\varphi}_{iy}
{\Bigr)}\!
\end{dmath*}
where
$\widehat{\varphi}_{ixy}$, $\widehat{\varphi}_{ix}$ and $\widehat{\varphi}_{iy}$
are respectively empirical values of $\varphi_{ixy}$, $\varphi_{ix}$ and $\varphi_{iy}$ 
that are computed from traces; i.e.,
$\widehat{\varphi}_{ixy} = {\textstyle \frac{\Dhi[x,y] - \Dhi[x,y]^2}{\ph_{XY}[x, y]}}$,
$\widehat{\varphi}_{ix} = {\textstyle \frac{\Dhi[x] - \Dhi[x]^2}{\ph_{XY}[x]}}$,
and
$\widehat{\varphi}_{iy} = {\textstyle \frac{\Dhi[y] - \Dhi[y]^2}{\ph_{XY}[y]}}$.
Then the bias is closer to $0$ when the sample sizes $n_i$ are larger.

\subsection{Evaluation of the Accuracy of Estimation}
\label{subsec:eval-accuracy}
In this section we present how to evaluate the accuracy of mutual information estimation.
The quality of the estimation depends on the sample sizes $n_i$ and other factors,
and can be evaluated using the variance of the estimate $\hat{I}(X; Y)$. 

\begin{restatable}[Variance of estimated mutual information]{theorem}{resVarianceMI}\label{lem:MI:variance:general}
The variance $\Variance{\hat{I}(X; Y)}$ of the estimated mutual information is given by:
\begin{align*}
\displaystyle\hspace{-2.3ex}
\Variance{\hat{I}(X; Y)} =
\sum_{i\in\I} \frac{\theta_i^2}{n_i}
\Biggl(\hspace{-2ex}
\sum_{~~~~(x, y)\in\D}\hspace{-2.2ex} \Di[x,y]
\Bigl({\textstyle 1 + \log\frac{P_X[x]P_Y[y]}{P_{XY}[x,y]}} \Bigr)^{\!2}
\!-
\biggl(\hspace{-2ex}\sum_{~~~~(x, y)\in\D}\hspace{-2.2ex} \Di[x,y]	
\Bigl({\textstyle 1 + \log\frac{P_X[x]P_Y[y]}{P_{XY}[x,y]}} \Bigr) \biggr)^{\!2}
\Biggr)\!
+ O(n_i^{-2})
{.}
\end{align*}
\end{restatable}

\begin{proof}[Proof sketch.]
The variance is calculated using the following:
\begin{align*}
\Variance{\hat{I}(X; Y)}
&= \displaystyle
\Variance{ \hat{H}(X) + \hat{H}(Y) - \hat{H}(X, Y) }
\\[-0.2ex]
&= \displaystyle
\Variance{\hat{H}(X)} + \Variance{\hat{H}(Y)} + \Variance{\hat{H}(X, Y)} 
\\[-0.2ex]
&~~
+ 2 \Cov{\hat{H}(X), \hat{H}(Y)}
- 2 \Cov{\hat{H}(X), \hat{H}(X, Y)}
- 2 \Cov{\hat{H}(Y), \hat{H}(X, Y)}
{.}
\end{align*}
The calculation of these variances and covariances and the whole proof are shown in Appendices~\ref{subsec:proof:var:ATS} and~\ref{subsec:proof:var:standard}.
(We will present a proof of this theorem by showing a more general claim, i.e., Theorem~\ref{thm:MI:variance:symbolic} in Section~\ref{subsec:known-subsystem}).
\end{proof}

The confidence interval of the estimate of mutual information is also useful to show how accurate the estimate is.
A smaller confidence interval corresponds to a more reliable estimate.
To compute the confidence interval approximately, we assume the sampling distribution of the estimate $\hat{I}(X; Y)$ as a normal distribution.
\footnote{In fact, Brillinge~\cite{brill1} shows that the sampling distribution of mutual information values is approximately normal for large sample size, and Chatzikokolakis et al.~\cite{DBLP:conf/tacas/ChatzikokolakisCG10} employ this fact to approximately compute the confidence interval.
We also empirically verified that the sampling distribution is closer to the normal distribution when the $n_i$'s are larger enough, and the evaluation of the obtained confidence interval will be demonstrated by experiments in Section~\ref{subsec:eval:quality}.}
Then the confidence interval is calculated using the variance $v$ obtained by Theorem~\ref{lem:MI:variance:general} as follows.
Given a significance level $\alpha$, we denote by $z_{\alpha\!/\!2}$ the \emph{z-score} for the $100(1-\frac{\alpha}{2})$ percentile point.
Then \emph{the $(1-\alpha)$ confidence interval} of the estimate is given by:
\begin{dmath*}
[\, \max(0, \pe - z_{\alpha\!/\!2} \sqrt{v}),~ \pe + z_{\alpha\!/\!2} \sqrt{v} \,]\;.
\end{dmath*}
For example, we use the z-score $z_{0.0025} = 1.96$ to compute the 95\% confidence interval.
To ignore the higher order terms the sample size $\sum_{i\in\I} n_i$ needs to be at least $4\!\cdot\!\#\X \!\cdot\!\#\Y$.

By Theorems~\ref{lem:MI:expectation:general} and~\ref{lem:MI:variance:general}, the bias and variance for the whole system can be computed compositionally from those for the components, unlike the mutual information itself.
This allows us to adaptively optimize the sample sizes for the components
as we will see in Section~\ref{sec:adaptive}.

\subsection{Application to Estimation of Shannon Entropy}
\label{sec:other-measures}

Hybrid statistical estimation can also be used to estimate the Shannon entropy $H(X)$ of a random variable $X$ in a probabilistic system.
Although the results for Shannon entropy are straightforward from those for mutual information, we present the formulas here for completeness.
For each $i\in\I$ let $\Dxi$ be the sub-distribution of $X$ for the component $S_i$.
Then the mean and variance of the estimate are obtained in the same way as in the Sections~\ref{subsec:estimate-MI-bias} and~\ref{subsec:eval-accuracy}.
%
\begin{restatable}[Mean of estimated Shannon entropy]{proposition}{resMeanShannonEntropy}\label{lem:entropy:expectation}
The expected value $\Expect{\hat{H}(X)}$ of the estimated Shannon entropy is given by:
\begin{align*}
\displaystyle
\Expect{\hat{H}(X)} =
H(X) -
\sum_{i\in\I} \frac{\theta_i^2}{2n_i}
\sum_{x\in\X^+}\hspace{-1ex}\frac{\Dxi[x] \left( 1 - \Dxi[x] \right)}{P_X[x]}
+ \O(n_i^{-2})
\texttt{.}
\end{align*}
\end{restatable}
%
See Appendix~\ref{subsec:proof:mean:standard} for the proof.
From this we obtain the bias of the Shannon entropy estimates.

\begin{restatable}[Variance of estimated Shannon entropy]{proposition}{resVarianceShannonEntropy}\label{lem:entropy:variance}
The variance $\Variance{\hat{H}(X)}$ of the estimated Shannon entropy is given by:
\begin{align*}
\Variance{\hat{H}(X)} =
\sum_{i\in\I} \frac{\theta_i^2}{n_i}
\Biggl(
\hspace{-1.2ex}\sum_{\hspace{1.5ex}x\in\X^+}\hspace{-1.8ex} \Dxi[x]
\Bigl(1 + \log P_X[x] \Bigr)^{\!2}
\!-\!
\biggl(\hspace{-1.2ex}\sum_{\hspace{1.5ex}x\in\X^+}\hspace{-1.8ex} \Dxi[x]
\Bigl(1 + \log P_X[x] \Bigr) \biggr)^{\!2}
\Biggr)\!
+ \O(n_i^{-2})
\texttt{.}
\end{align*}
\end{restatable}
See Appendix~\ref{subsec:proof:var:standard} for the proof.
From this we obtain the confidence interval of the Shannon entropy estimates.

\section{Estimation Using Prior Knowledge about Systems}
\label{sec:estimate-known-system}

In this section we show how to use prior knowledge about systems to improve the accuracy of the estimation, i.e., to make the variance (and the confidence interval size) smaller and reduce the required sample sizes.

\subsection{Approximate Estimation Using Knowledge of Prior Distributions}
\label{subsec:previous-known-prior}
Our hybrid statistical estimation method integrates both precise and statistical analysis, and it can be seen as a generalization and extension of previous work~\cite{DBLP:conf/tacas/ChatzikokolakisCG10,mod1,brill1}. 

Due to an incorrect computation of the bias, the state-of-the-art statistical analysis tool \LeakWatch{}~\cite{DBLP:conf/esorics/ChothiaKN14,leakwatch:www} does not correctly estimate mutual information. We explain this problem in Section~\ref{subsec:prior-soa} and show how to fix it in Section~\ref{subsec:better-known-prior}. 
We extend this result to the estimation of conditional entropy in Section~\ref{subsec:Conditional-entropy}.

\subsubsection{State of the Art}\label{subsec:prior-soa}
For example, Chatzikokolakis et.al.~\cite{DBLP:conf/tacas/ChatzikokolakisCG10} present a method for estimating mutual information between two random variables $X$ (over secret input values $\X$) and $Y$ (over observable output values $\Y$) when the analyst knows the (prior) distribution $P_X$ of $X$.
In the estimation they collect execution traces by running a system for each secret value $x\in\X$.
Thanks to the precise knowledge of $P_X$, they have more precise and accurate estimates than the other previous work~\cite{mod1,brill1} that also estimates $P_X$ from execution traces.

Estimation using the precise knowledge of $P_X$ is an instance of our result if a system is decomposed into the component $S_x$ for each secret $x\in\X=\I$.
If we assume all joint probabilities are non-zero,
the following approximate result in~\cite{DBLP:conf/tacas/ChatzikokolakisCG10} follows from Theorem~\ref{lem:MI:expectation:general}. 

\begin{corollary}\label{lem:MI:expectation:Chothia}
The expected value $\Expect{\hat{I}(X; Y)}$ of the estimated mutual information is given by:
\begin{dmath*}
\Expect{\hat{I}(X; Y)} = I(X; Y) + {\textstyle\frac{(\#\X - 1)(\#\Y - 1)}{2n}} + \O(n^{-2}),
\end{dmath*}
where $\#\X$ and $\#\Y$ denote the numbers of possible secrets and observables respectively.
\end{corollary}


Using this result the bias $\Expect{\hat{I}(X; Y)} - I(X; Y)$ is calculated as $\frac{(\#\X - 1)(\#\Y - 1)}{2n}$ in~\cite{DBLP:conf/tacas/ChatzikokolakisCG10}, which depends only on the size of the joint distribution matrix.
However, the bias can be strongly influenced by probability values close or equivalent to zero in the distribution; therefore their approximate results can be correct only when all joint probabilities are non-zero and large enough, which is a strong restriction in practice.
We show in Section~\ref{sec:application-QIF} that the tool \LeakWatch{}~\cite{DBLP:conf/esorics/ChothiaKN14} uses Corollary~\ref{lem:MI:expectation:Chothia}, and consequently miscalculates bias and gives an estimate far from the true value in the presence of very small probability values.

\subsubsection{Our Estimation Using Knowledge of Prior Distributions}
\label{subsec:better-known-prior}

To overcome these issues we present more general results in the case in which the analyst knows the prior distribution $P_X$.
We assume that a system $\Sys$ is decomposed into the disjoint component $S_{ix}$ for each index $i\in\I$ and input $x\in\X$, and that each $S_{ix}$ is executed with probability $\theta_{ix}$ in the system $\Sys$.
Let $\Theta = \{ \theta_{ix} : i\in\I,\, x\in\X \}$.

\paragraph{Estimation of Mutual Information}
\label{subsec:better-known-prior-MI}

In the estimation of mutual information we separately execute each component $S_{ix}$ multiple times to collect execution traces.
Unlike the previous work the analyst may change the number of executions $n_{i} P_X[x]$ to $n_i \lambda_i[x]$ where $\lambda_i[x]$ is an \emph{importance prior} that the analyst chooses to determine how the sample size $n_i$ is allocated for each component $S_{ix}$.
An adaptive way of choosing the importance priors will be presented in Section~\ref{sec:adaptive}.
Let $\Lambda = \{ \lambda_{i} : i\in\I \}$.

Given the number $\Kixy$ of $S_{ix}$'s traces with output $y$, we define the conditional distribution $\Di$ of output given input:
$
\Di[y|x]\!\eqdef\!\frac{\Kixy}{n_i \lambda_i[x]}.
$
Let $\Mixy\!=\!\frac{\theta_{ix}^2}{\lambda_i[x]} \Di[y|x] \left( 1\!-\!\Di[y|x] \right)$.
Then we can calculate the mean and variance of the mutual information $\hat{I}_{\Theta,\Lambda}(X; Y)$ using $\hat{\Di}$, $\Theta$, $\Lambda$ as follows.

\begin{restatable}[Mean of mutual information estimated using the knowledge of the prior]{proposition}{resMeanMIKnownPrior}\label{lem:MI:expectation:KnownPrior}
~~
The expected value $\Expect{\hat{I}_{\Theta,\Lambda}(X; Y)}$ of the estimated mutual information is given by:
\begin{align*}
\displaystyle\hspace{-1ex}
\Expect{\hat{I}_{\Theta,\Lambda}(X; Y)} =
I(X; Y)+
\sum_{i\in\I}
\frac{1}{2n_i}
\sum_{y\in\Y^+}\hspace{-0.2ex}
\Bigl(
\sum_{x\in\D_y}\hspace{-0.7ex}
{\textstyle \frac{\Mixy}{P_{XY}[x, y]}}
-
{\textstyle \frac{\sum_{x\in\D_y}\Mixy }{P_{Y}[y]} }
\!\Bigr)\!
+ \O(n_i^{-2})
\texttt{.}
\end{align*}
\end{restatable}

\begin{restatable}[Variance of mutual information estimated using the knowledge of the prior]{proposition}{resVarianceMIKnownPrior}\label{lem:MI:variance:KnownPrior}
~~~
The variance\, $\Variance{\hat{I}_{\Theta,\Lambda}(X; Y)}$ of the estimated mutual information is given by:
\begin{align*}
\hspace{-1ex}
\Variance{\hat{I}_{\Theta,\Lambda}(X; Y)} =
\sum_{i\in\I}\sum_{x\in\X^+}\hspace{-0.3ex}
\frac{ \theta_{ix}^2}{n_i \lambda_i[x]}
\!\Biggl(
\hspace{-1.7ex}
\sum_{~~~~y\in\D_x}\hspace{-1.8ex}
{\textstyle D_i[y|x]
\Bigl(\log\frac{P_Y[y]}{P_{XY}[x,y]}\!\Bigr)^{\!2}}
\!-
\biggl(
\hspace{-1.7ex}
\sum_{~~~~y\in\D_x}\hspace{-1.8ex} D_i[y|x]
\Bigl({\textstyle \log\frac{P_Y[y]}{P_{XY}[x,y]}} \Bigr) \biggr)^{\!2}
\Biggr)
\!+ \O(n_i^{-2})
\texttt{.}
\end{align*}
\end{restatable}
See Appendix~\ref{subsec:proof:est:prior} for the proofs.

\subsubsection{Estimation of Conditional Entropy}
\label{subsec:Conditional-entropy}

The new method can also estimate the conditional Shannon entropy $H(X|Y)$ of a random variable $X$ given a random variable $Y$ in a system.
In the context of quantitative security, $H(X|Y)$ represents the uncertainty of a secret $X$ after observing an output $Y$ of the system.
The mean and variance of the conditional entropy are obtained from those of the mutual information in the case where the analyst knows the prior.
\begin{proposition}[Mean of estimated conditional entropy]\label{lem:cond-entropy:expectation}
The expected value $\Expect{\hat{H}_{\Theta,\Lambda}(X|Y)}$ of the estimated conditional Shannon entropy is given by $H(X)-\Expect{\hat{I}_{\Theta,\Lambda}(X; Y)}$ where $\Expect{\hat{I}_{\Theta,\Lambda}(X; Y)}$ is the expected value of the mutual information in the case where the analyst knows the prior (shown in Proposition~\ref{lem:MI:expectation:KnownPrior}).
\end{proposition}
\begin{proof}
By $H(X|Y) = H(X) - I(X; Y)$, we obtain 
$\Expect{\hat{H}_{\Theta,\Lambda}(X|Y)} = H(X) - \Expect{\hat{I}_{\Theta,\Lambda}(X; Y)}$.
\end{proof}

\begin{proposition}[Variance of estimated conditional entropy]\label{lem:cond-entropy:variance}
The variance $\Variance{\hat{H}_{\Theta,\Lambda}(X|Y)}$ of the estimated conditional Shannon entropy coincides with the variance $\Variance{\hat{I}_{\Theta,\Lambda}(X; Y)}$ of the mutual information in the case where the analyst knows the prior (shown in Proposition~\ref{lem:MI:variance:KnownPrior}).
\end{proposition}
\begin{proof}
By $H(X|Y) = H(X) - I(X; Y)$, we obtain 
$\Variance{\hat{H}_{\Theta,\Lambda}(X|Y)} = \Variance{\hat{I}_{\Theta,\Lambda}(X; Y)}$.
\end{proof}

\subsection{Abstraction-Then-Sampling Using Partial Knowledge of Components}
\label{subsec:known-subsystem}
\begin{wrapfigure}[16]{t}{0.43\linewidth}
  \centering
  \includegraphics[width=0.42\textwidth]{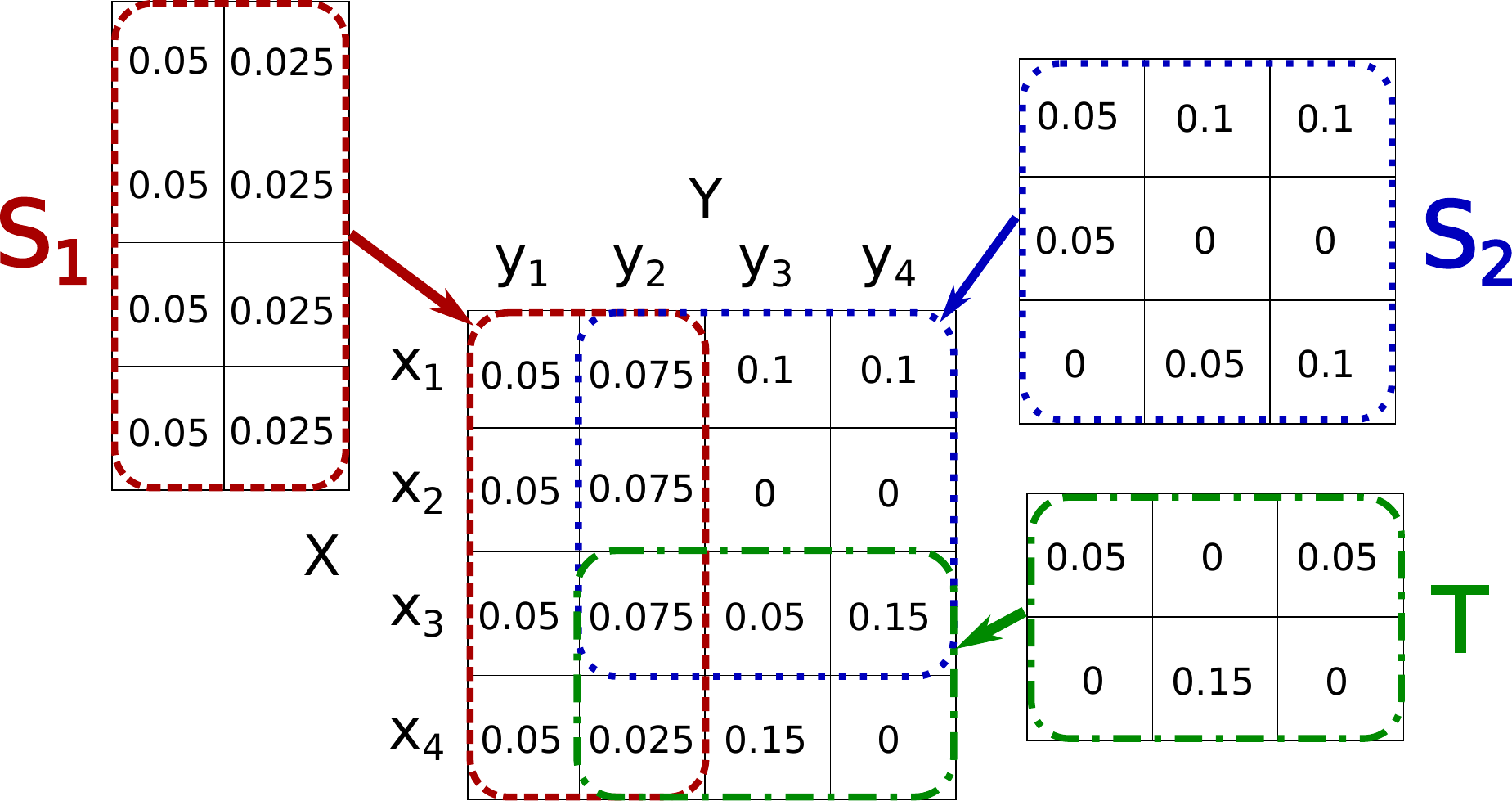}\vspace{10pt}
  \caption{Joint distribution composed of 3 components. All the rows of component $S_1$ are identical, hence abstraction-then-sampling can be used on it.}
  \label{fig:jointalt}
\end{wrapfigure}

In this section we extend the hybrid statistical estimation method to consider the case in which the analyst knows that the output of some of the components does not depend on the secret input (for instance by static code analysis).
Such prior knowledge may help us abstract components into simpler ones and thus reduce the sample size for the statistical analysis.

We illustrate the basic idea of this ``abstraction-then-sampling'' technique as follows.
Let us consider an analyst who knows two pairs $(x,y)$ and $(x',y')$ of inputs and outputs have the same probability in a component $S_i$:
$\Di[x,y] = \Di[x',y']$.
Then, when we construct the empirical distribution $\Dhi$ from a set of traces, we can count the number $K_{i\{(x,y),(x',y')\}}$ of traces having either $(x,y)$ or $(x',y')$, and divide it by two:
$\Kixy = K_{ix'y'} = \frac{K_{i\{(x,y),(x',y')\}}}{2}$.
Then to achieve a certain accuracy, the sample size required for the estimation using the prior knowledge on the equality $\Kixy = K_{ix'y'}$ is smaller than that without using it.

In the following we generalize this idea to deal with similar information that the analyst may possess about the components.
Let us consider a (probabilistic) system in which for some components, observing the output provides no information on the input. Assume that the analyst is aware of this by \emph{qualitative} information analysis (for verifying non-interference).
Then such a component $S_i$ has a sub-channel matrix where all non-zero rows have an identical conditional distribution of outputs given inputs \cite{cover2006eit}.
Consequently, when we estimate the $\#\X_i\times\#\Y_i$ matrix of $S_i$ it suffices to estimate one of the rows, hence the number of executions is proportional to $\#\Y_i$ instead of $\#\X_i\times\#\Y_i$.

The abstraction-then-sampling approach can be simply explained by referring to the joint distribution matrix in Fig.~\ref{fig:jointalt}. 
Note that each row of the sub-distribution matrix for component $S_1$ is identical, even though the rows of the joint matrix are not, and assume that the analyst knows this by analyzing the code of the program and finding out that for component $S_1$ the output is independent from the input.
Then the analyst would know that it is unnecessary to execute the component separately for each possible input value in $S_1$: it is sufficient to execute the component only for one value of the input, and to apply the results to each row in the sub-distribution matrix for component $S_1$. 
This allows the analyst to obtain more precise results and a smaller variance (and confidence interval) on the estimation given a fixed total sample size $n_i$ for the component.

Note that even when some components leak no information, computing the mutual information for the whole system requires constructing the matrix of the system, hence the matrices of all components.

Let $\II$ be the set of indexes of components that have channel matrices whose non-zero rows consist of the same conditional distribution.
For each $i\in\II$, we define $\pi_i[x]$ as the probability of having an input $x$ in the component $S_i$.
To estimate the mutual information for the whole system, we apply the abstraction-then-sampling technique to the components $\II$ and the standard sampling technique (shown in Section~\ref{sec:compositional}) to the components $\I\setminus\II$.

Then the mean and variance of the mutual information are as follows.
The following results show that the bias and confidence interval are narrower than when not using the prior knowledge of components.

\begin{restatable}[Mean of mutual information estimated using the abstraction-then-sampling]{theorem}{resMeanMISymbolic}\label{thm:MI:expectation:symbolic}
The expected value $\Expect{\hat{I}_{\II}(X; Y)}$ of the estimated mutual information is given by:
\begin{align*}
\displaystyle&\hspace{-2ex}
\Expect{\hat{I}_{\II}(X; Y)} = 
I(X; Y)+\hspace{-1ex}
\sum_{i\in\I\setminus\II}\hspace{-0.5ex}
\frac{\theta_i^2}{2n_i}\hspace{-0.2ex}
\biggl(\!
\sum_{\,(x,y)\in\D\hspace{-3.2ex}}\hspace{-0.2ex}
\varphi_{ixy}
-\hspace{-0.5ex}
\sum_{x\in\X^+\hspace{-2.5ex}}\hspace{-0.2ex}
\varphi_{ix}
-\hspace{-0.5ex}
\sum_{y\in\Y^+\hspace{-2.5ex}}\hspace{-0.2ex}
\varphi_{iy}\hspace{-0.5ex}
{\biggr)}\hspace{-0.5ex}
+\hspace{-0.5ex}
\displaystyle
\sum_{i\in\II}\hspace{-0.5ex}
\frac{\theta_i^2}{2n_i}\hspace{-0.2ex}
\biggl(\!
\sum_{\,(x,y)\in\D\hspace{-3.2ex}}\hspace{-0.2ex}
\psi_{ixy}
-\hspace{-0.5ex}
\sum_{\,y\in\Y^+\hspace{-1.5ex}}\hspace{-0.2ex}
\varphi_{iy}\hspace{-0.5ex}
\biggr)\!
+ \O(n_i^{-2})
\end{align*}
where $\psi_{ixy} \eqdef \frac{\Di[x,y]\pi_i[x] - \Di[x,y]^2}{P_{XY}[x, y]}$.
\end{restatable}

See Appendix~\ref{subsec:proof:mean:ATS} for the proof.

%
\begin{restatable}[Variance of mutual information estimated using the abstraction-then-sampling]{theorem}{resVarianceMISymbolic}\label{thm:MI:variance:symbolic}~~
The variance \allowbreak $\Variance{\hat{I}_{\II}(X; Y)}$ of the estimated mutual information is given by:
\begin{align*}
\Variance{\hat{I}_{\II}(X; Y)} =
&\hspace{-1ex}
\sum_{~~i\in\I\setminus\II}\hspace{-0.5ex} 
\frac{\theta_i^2}{n_i}
\Biggl(
\hspace{-2ex}
\sum_{~~~~(x, y)\in\D}\hspace{-2ex} \Di[x,y]
\Bigl({\textstyle 1 + \log\!\frac{P_X[x]P_Y[y]}{P_{XY}[x,y]}} \Bigr)^{\!2}
\!-\!
\biggl(\hspace{-2ex}\sum_{~~~~(x, y)\in\D}\hspace{-2ex} \Di[x,y]	
\Bigl({\textstyle 1 + \log\!\frac{P_X[x]P_Y[y]}{P_{XY}[x,y]}} \Bigr) \biggr)^{\!2}
\Biggr)\!
\\
&\hspace{0.2ex}
+\hspace{-0.2ex}
\sum_{i\in\II} \frac{\theta_i^2}{n_i}
\biggl(\hspace{-0.5ex}
\sum_{~~y\in\Y^+}\hspace{-0.7ex}
\Dyi[y] \gamma_{ixy}^{\!2}
-
\Bigl(\hspace{-0.5ex}
\sum_{~~y\in\Y^+}\hspace{-0.7ex}
\Dyi[y] \gamma_{ixy}
\Bigr)^{\!2}
\biggr)\!
+\, \O(n_i^{-2})\;
\end{align*}
where 
$\gamma_{ixy}\eqdef\displaystyle
 \log P_Y[y] - \sum_{x\in\X}\pi_i[x]\log P_{XY}[x,y]$.
\end{restatable}

See Appendix~\ref{subsec:proof:var:ATS} for the proof.

\section{Adaptive Optimization of Sample Sizes}
\label{sec:adaptive}

In this section we present a method for deciding the sample size $n_i$ of each component $S_i$ to estimate mutual information with an optimal accuracy when using the hybrid estimation technique in Section~\ref{sec:compositional} and its variants using prior information on the system in Section~\ref{sec:estimate-known-system}. The proof for all results in this section can be found in Appendix~\ref{subsec:proof:adaptive}.

\paragraph{Mutual Information}

To decide the sample sizes we take into account the trade-off between accuracy and cost of the statistical analysis:
The computational cost increases proportionally to the sample size $n_i$ (i.e., the number of $S_i$'s execution traces), while a larger sample size $n_i$ provides a smaller variance hence a more accurate estimate.

More specifically, given a budget of a total sample size $n$ for the whole system, we obtain an optimal accuracy of the estimate by adjusting each component's sample size $n_i$
\footnote{This idea resembles the \emph{importance sampling} in statistical model checking in that the sample size is adjusted to make the estimate more accurate.}
 (under the constraint $n = \sum_{i\in\I} n_i$).
To compute the optimal sample sizes, we first run each component to collect a small number (compared to $n$, for instance dozens) of execution traces.
Then we calculate certain intermediate values in computing the variance and determine sample sizes for further executions.
Formally, let $v_i$ be the following intermediate value of the variance for the component $S_i$:
\begin{align*}
v_i = \displaystyle
\theta_i^2
\Biggl(
\hspace{-2.2ex} 
\sum_{~~~~(x, y)\in\D}\hspace{-2.7ex} \hat{\Di}[x,y]
{\textstyle \Bigl(1 + \log\!\frac{\widehat{P}_X[x]\widehat{P}_Y[y]}{\widehat{P}_{XY}[x,y]} \Bigr)^{\!2}}
\!-\!
\biggl(
\hspace{-2.2ex} 
\sum_{~~~~(x, y)\in\D}\hspace{-2.7ex} \hat{\Di}[x,y]
{\textstyle \Bigl(1 + \log\!\frac{\widehat{P}_X[x]\widehat{P}_Y[y]}{\widehat{P}_{XY}[x,y]} \Bigr) \biggr)^{\!2}}
\Biggr)
{.}
\end{align*}
Then we find $n_i$'s that minimize the variance $v = \sum_{i\in\I} \frac{v_i}{n_i}$ of the estimate by using the following theorem.

\begin{restatable}[Optimal sample sizes]{theorem}{resOptimalSS}\label{thm:adaptive-sampling}
Given the total sample size $n$ and the above intermediate variance $v_i$ of the component $S_i$ for each $i\in\I$,\,
the variance of the mutual information estimate is minimized if, for all $i\in\I$,\, the sample size $n_i$ for $S_i$ is given by:
$n_i = \frac{\sqrt{v_i}n}{\sum_{j = 1}^{m} \sqrt{v_j}}$.
\end{restatable}

\paragraph{Shannon Entropy}
Analogously to Theorem~\ref{thm:adaptive-sampling} we can adaptively optimize the sample sizes in the estimation of Shannon entropy in Section~\ref{sec:other-measures}.
To compute the optimal sample sizes we define $v'_i$ by:
\begin{align*}
v'_i =
\theta_i^2
\Biggl(
\hspace{-1.2ex}\sum_{\hspace{1.5ex}x\in\X^+}\hspace{-1.8ex} \Dxi[x]
\Bigl(1 + \log P_X[x] \Bigr)^{\!2}
\!-\!
\biggl(\hspace{-1.2ex}\sum_{\hspace{1.5ex}x\in\X^+}\hspace{-1.8ex} \Dxi[x]
\Bigl(1 + \log P_X[x] \Bigr) \biggr)^{\!2}
\Biggr)\!
+ \O(n_i^{-2})
\texttt{.}
\end{align*}
Then we can compute the optimal sample sizes by using the following proposition. 

\begin{restatable}[Optimal sample sizes for Shannon entropy estimation]{proposition}{resOptimalSSforShannon}\label{prop:adaptive-sampling:shannon}
Given the total sample size $n$ and the above intermediate variance $v'_i$ of the component $S_i$ for each $i\in\I$,\,
the variance of the Shannon entropy estimate is minimized if, for all $i\in\I$,\, the sample size $n_i$ for $S_i$ satisfies
$n_i = \frac{\sqrt{v'_i}n}{\sum_{j = 1}^{m} \sqrt{v'_j}}$.
\end{restatable}

\paragraph{Knowledge of the Prior}
Analogously to Theorem~\ref{thm:adaptive-sampling}, the sample sizes $n_i$ and the importance priors $\lambda_i$ can be adaptively optimized 
when the prior distribution of the input is known to the analyst (as presented in Section~\ref{subsec:previous-known-prior}).

\begin{restatable}[Optimal sample sizes when knowing the prior]{proposition}{resSampleSizesKnownPrior}\label{prop:adaptive-sampling-known-prior}
For each $i\in\I$ and $x\in\X$, let $v_{ix} $ be the following intermediate variance of the component $S_{ix}$.
\begin{align*}
v_{ix} = \displaystyle
\theta_{ix}^2
\Biggl(
\sum_{y\in\D_x}\hspace{-0.2ex} \hat{\Di}[y|x]
\biggl({\textstyle \log\frac{\widehat{P}_Y[y]}{\widehat{P}_{XY}[x,y]}} \biggr)^{\!2}
\!-
\biggl(\sum_{y\in\D_x}\hspace{-0.2ex} \hat{\Di}[y|x]
\Bigl({\textstyle \log\frac{\widehat{P}_Y[y]}{\widehat{P}_{XY}[x,y]}} \Bigr) \biggr)^{\!2}
\Biggr).
\end{align*}
Given the total sample size $n$, the variance of the estimated mutual information is minimized if, for all $i\in\I$ and $x\in\X$,\, the sample size $n_i$ and the importance prior $\lambda_i$ satisfy:
$
n_i\lambda_i[x] = \frac{\sqrt{v_{ix}}n}{\sum_{j = 1}^{m} \sqrt{v_{jx}}}\;
\texttt{.}
$
\end{restatable}

\paragraph{Abstraction-then-sampling}
Finally, the sample sizes can be optimized for the abstraction-then-sampling approach in Section~\ref{subsec:known-subsystem} by using the following theorem.

\begin{restatable}[Optimal sample sizes using the abstraction-then-sampling]{theorem}{resAdaptiveMISymbolic}\label{thm:adaptive-sampling-ATS}Let $v_{i}^{\star}$ be the following intermediate variance of the component $S_i$:
\begin{align*}
\hspace{-3.3ex}
v_{i}^{\star} =
\begin{cases}
\displaystyle
~\theta_i^2
\Biggl(
\hspace{-2ex}
\sum_{~~~~(x, y)\in\D}\hspace{-2ex} \hat{\Di}[x,y]
\Bigl({\textstyle 1 + \log\!\frac{\ph_X[x]\ph_Y[y]}{\ph_{XY}[x,y]}} \Bigr)^{\!2}
\!-\!
\biggl(\hspace{-2ex}\sum_{~~~~(x, y)\in\D}\hspace{-2ex} \hat{\Di}[x,y]	
\Bigl({\textstyle 1 + \log\!\frac{\ph_X[x]\ph_Y[y]}{\ph_{XY}[x,y]}} \Bigr) \biggr)^{\!2}
\Biggr)\!
&
\mbox{if $i\in\I\setminus\II$}
\\[3ex]
\displaystyle
~\theta_i^2
\biggl(\hspace{-0.5ex}
\sum_{~~y\in\Y^+}\hspace{-0.7ex}
\hat{\Dyi}[y]
\hat{\gamma}_{ixy}^{\!2}
-
\Bigl(\hspace{-0.5ex}
\sum_{~~y\in\Y^+}\hspace{-0.7ex}
\hat{\Dyi}[y] \hat{\gamma}_{ixy}
\Bigr)^{\!2}
\biggr)
&
\mbox{if $i\in\II$}
\end{cases}
\end{align*}
Given the total sample size $n$, the variance of the estimated mutual information is minimized if, for all $i\in\I$ and $x\in\X$,\, the sample size $n_i$ is given by:
$n_i = \frac{\sqrt{v_{i}^{\star}}n}{\sum_{j = 1}^{m} \sqrt{v_{j}^{\star}}}
{.}
$
\end{restatable}

\section{Implementation in the \toolname  Tool}\label{sec:functioning}
We describe how \toolname estimates the Shannon leakage of a given program, i.e., the mutual information between secret and output, implementing the hybrid statistical estimation procedure described above. The tool determines which components of the program to analyze with precise analysis and which with statistical analysis, and inserts appropriate annotations in the code. 
The components are analyzed with the chosen technique and the results are composed into a joint probability distribution of the secret and observable variables. 
Finally, the mutual information and its confidence interval are computed from the joint distribution.


The \toolname tool, including user documentation and source code is freely available at~\url{https://project.inria.fr/hyleak/}.
Multiple examples and the scripts to generate the results are also provided. 

\toolname is very simple to use. The code of the system to analyze is written in a file e.g., \texttt{system.hyleak}. We invoke the tool with the command:

\smallskip
\hspace{-3.5ex}\texttt{{\small ./hyleak system.hyleak}}
\smallskip
 
The tool generates various \texttt{\!.pp} text files with 
analysis information and the control flow graph of the program. 
Finally, it outputs the prior and posterior Shannon entropy estimates for the secret, the estimated Shannon leakage of the 
program before and after bias correction, and its confidence interval.
\toolname can also print the channel matrix and 
%
%
additional information; the full list of arguments is printed by \texttt{./hyleak -h}.

\subsection{Illustrating Example: Random Walk}\label{subsec:implementation-example}

Consider the following random walk problem (modeled in Fig.~\ref{fig:randomwalkcode}). 

The secret is the initial location of an agent, encoded by a single natural number representing an approximate distance from a given point, e.g., in meters. 
Then the agent takes a fixed number of steps. At each step the distance of the agent increases or decreases by 10 meters with the same probability. After this fixed number of random walk steps, the final location of the agent is revealed, and the attacker uses it to guess the initial location of the agent. 

This problem is too complicated to analyze by precise analysis, because the analysis needs to explore every possible combination of random paths, amounting to an exponential number in the random walk steps. It is also intractable to analyze with a fully statistical approach, since there are hundreds of possible secret values and the program has to be simulated many times for each of them to sufficiently observe the agent's behavior.

As shown in Section~\ref{sec:evaluation}, \toolname's hybrid approach computes the leakage significantly faster than the fully precise analysis and more accurately than the fully statistical analysis.

\begin{figure}
\begin{small}
\begin{algorithm}[H]
const $MAX$:=14;\\
secret int32 $sec$ := [201,800]; \\[-0.1ex]
observable int32 $obs$ := 0; \\[-0.1ex]
public int32 $loc$ := 0; \\[-0.1ex]
public int32 $seed$ := 0; \\[-0.1ex]
public int32 $ran$ := 0; \\[-0.1ex]
\uIf{ $sec \le$ 250\tcp{sec:[201,800]; TOT_OBS = 1 TOT_INT = 1; SEC DEP} }{ 
  $loc$ := 200; \\[-0.1ex]
}\uElseIf{ $sec \le$ 350 }{
  $loc$ := 300; \\[-0.1ex]
}\uElseIf{ $sec \le$ 450 }{
  $loc$ := 400; \\[-0.1ex]
}\uElseIf{ $sec \le$ 550 }{
  $loc$ := 500; \\[-0.1ex]
  }\uElseIf{ $sec \le$ 650 }{
  $loc$ := 600; \\[-0.1ex]
  }\uElseIf{ $sec \le$ 750 }{
  $loc$ := 700; \\[-0.1ex]
} \Else{
  $loc$ := 800; \\[-0.1ex]
}
{\color{red}simulate-abs;}\tcp{ loc: [200,800]; sec: [201,800]; TOT_OBS = 1; TOT_INT = 601}
\For{ $time$ in [0,$MAX$]  }{
  $ran$ := \random(0,9); \\[-0.1ex]
  \eIf{ $ran \le$ 5 }{
    $loc$ := $loc$ + 10; \\[-0.1ex]
  }{
    $loc$ := $loc$ - 10; \\[-0.1ex]
  }
\tcp{loc: [50,950]; ran: [0,9], TOT_OBS = 1; TOT_INT = 9010}}  
$obs$ := $loc$; \tcp{obs: [50,950]; TOT_OBS = 901; TOT_INT = 9010} 
return;
\end{algorithm}
\end{small}
\caption{Source code for the Random Walk illustrative example explained in Section~\ref{subsec:implementation-example}. The comments show the estimates for the value ranges of some variables, as computed by HyLeak following Step 2 of Section~\ref{subsec:implement-architecture}, where \texttt{TOT_OBS} represent the estimate of the possible combinations of values of all observable variables and \texttt{TOT_INT} the estimate of the possible combinations of values of  all internal variables. The red \texttt{simulate-abs} statement in Line 23 shows where the statement will be automatically added to implement the division into components, as explained in Step 2 of Section~\ref{subsec:implement-architecture} and in more details in Section~\ref{subsec:implement-components}.\label{fig:randomwalkcode}}
\end{figure}

\subsection{Architecture}\label{subsec:implement-architecture}

The \toolname tool implementation consists of the following 4 steps.
Steps 1 and 2 are implemented with different ANTLR parsers~\cite{ANTLR3:2007}.
The implementation of Step 3 inherits some code from the QUAIL tool~\cite{quail:www,DBLP:conf/cav/BiondiLTW13,DBLP:conf/spin/BiondiLQ15}
to employ QUAIL's optimization techniques for precise analysis, i.e., parallel analysis of execution traces and compact Markovian state representation.

\subsubsection*{Step 1: Preprocessing}
\paragraph*{Step 1a. Lexing, parsing and syntax checking.}
\toolname starts by lexical analysis, macro substitution and syntax analysis.
In macro substitution the constants defined in the input program are replaced with their declared values, and simple operations are resolved immediately. In the example in Fig.~\ref{fig:randomwalkcode}, this replaces the value of constant \textit{MAX} on Line 24 with its declared value from Line 1.
The tool checks whether the input program correctly satisfies the language syntax. 
In case of syntax errors, an error message describing the problem is produced and execution is terminated. 

\paragraph*{Step 1b. Loop unrolling and array expansion.} 
\texttt{for} loops ranging over fixed intervals are unrolled to optimize the computation of variable ranges 
and thus program decomposition in Step 2. In the example in Fig.~\ref{fig:randomwalkcode}, the \texttt{for} loop in Line 24 gets replaced by a fixed number of repetitions of its code with increasing values of the variable \textit{time}.
Similarly, arrays are replaced with multiple variables indexed by their position number in the array.
Note that these techniques are used only to optimize program decomposition and not required to compute the leakage in programs with arbitrary loops.

\subsubsection*{Step 2: Program Decomposition and Internal Code Generation}

If a \texttt{simulate} or \texttt{simulate-abs} statement is present in the code, the program decomposition step is skipped and such statements are used to determine program decomposition. 

The code may be decomposed heuristically only at conditional branching
to efficiently compute the probability of each component being executed. 
Moreover, each component must be a terminal in the control flow graph, hence no component is executed afterwards.
This is because the estimation method requires that the channel matrix for the system is the weighted sum of those for its components, and that the weight of a component is the probability of executing it, as explained in Sections~\ref{sec:overview:hybrid}.

The analysis method and its parameters for each component $S_i$ are decided by estimating the computational cost of analyzing $S_i$.
Let $\Z_i$ be the set of all \emph{internal randomness} (i.e., the non-secret variables whose values are assigned according to probability distributions) in $S_i$.
Then the cost of the statistical analysis is proportional to $S_i$'s sub-channel matrix size $\#\X_i\times\#\Y_i$, while the cost of the precise analysis is proportional to the number of all traces in $S_i$ (in the worst case proportional to $\#\X_i\times\#\Z_i$).
Hence the cost estimation is reduced to counting $\#\Y_i$ and $\#\Z_i$.

To obtain this, for each variable and each code line, an estimation of the number of possible values of the variable at the specific code line is computed. This is used to evaluate at each point in the input program whether it would be more expensive to use precise or statistical analysis. These estimations are shown as comments for different lines of the source code in Fig.~\ref{fig:randomwalkcode}. 
To reduce the computational cost of the estimation of variable ranges, we apply ad-hoc heuristics to obtain approximate estimates.

After determining the decomposition of the program, \toolname automatically adds \texttt{simulate} and/or \texttt{simulate-abs} statements in the code to signal which parts of the input program should be analyzed with standard random sampling (Section~\ref{sec:compositional}) and with abstraction-then-sampling (Section~\ref{subsec:known-subsystem}) respectively.
For instance, since no annotations originally exist in the example source code in Fig.~\ref{fig:randomwalkcode}, \toolname adds the \texttt{simulate-abs} statement (written in red) on Line 23.
The procedure for decomposition is shown in Fig.~\ref{proc:decompose} and is illustrated in Section~\ref{subsec:implement-components} using the Random Walk example of Fig.~\ref{fig:randomwalkcode}.
While the decomposition procedure is automated, it is a heuristic that does not guarantee to produce an optimal decomposition.
Hence for usability the choice of analysis can be controlled by user's annotations on the code. 

At the end, 
the input program is translated into a simplified internal language. Conditional statements and loops (\texttt{if}, \texttt{for}, and \texttt{while}) are rewritten into \texttt{if-goto} statements. 

\subsubsection*{Step 3: Program Analysis}
\begin{wrapfigure}[47]{r}{0.51\textwidth}
\centering
\vspace{-0.5cm}
\includegraphics[width=0.52\textwidth]{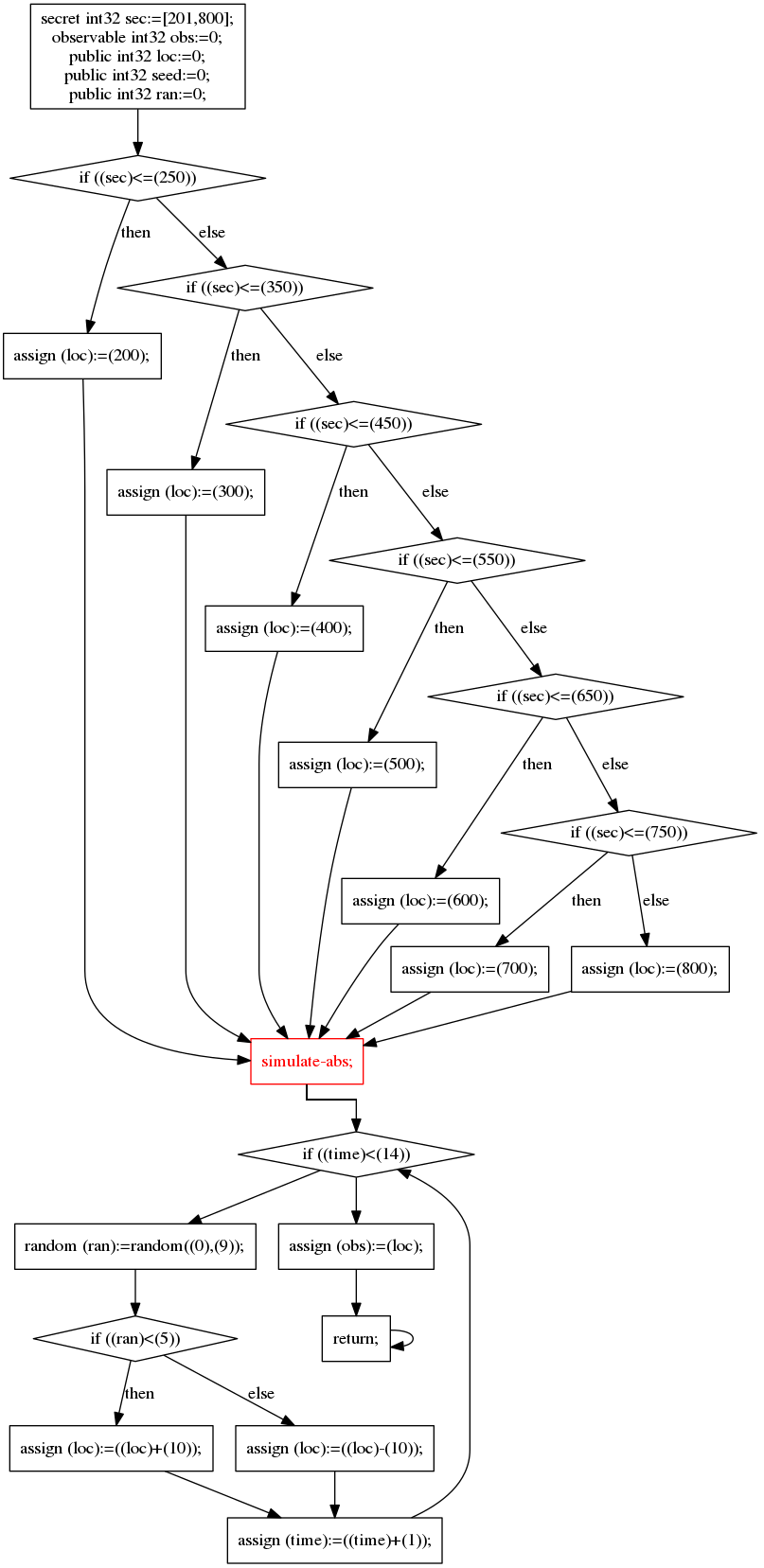}\vspace{1ex}
     \caption{Control flow graph for the input code of Fig.~\ref{fig:randomwalkcode}. The red node corresponds to the \texttt{simulate-abs} statement on Line 23 of Fig.~\ref{fig:randomwalkcode}. Each arrow entering the red node comes from a symbolic path with different secret and internal values, inducing a different component.}
     \label{fig:randomwalkcfg}
\end{wrapfigure}
In this step the tool analyzes the executions of the program using the two approaches.

\paragraph*{Step 3a. Precise analysis.} 
The tool performs a depth-first symbolic execution of all possible execution traces of the input program, until it finds a \texttt{return}, \texttt{simulate}, or \texttt{simulate-abs} statement.
When reaching a \texttt{return} statement the tool recognizes the symbolic path as terminated and stores its secret and output values.
In the cases of \texttt{simulate} and \texttt{simulate-abs} statements it halts the symbolic path, saves the resulting program state, and schedules it for standard random sampling (Section~\ref{sec:compositional}) or for abstraction-then-sampling  (Section~\ref{subsec:known-subsystem}), respectively, starting from the saved program state.
In the example in Fig.~\ref{fig:randomwalkcode}, the tool analyzes the code precisely and generates one symbolic path for each of the possible \texttt{if-elseif-else-end} statements from Line 7 to Line 22, so 7 symbolic paths in total 
(as shown in the control flow graph of the code in Fig.~\ref{fig:randomwalkcfg}).
Then each of the 7 symbolic paths meets the \texttt{simulate-abs} statement in Line 23, so they get removed from precise analysis and scheduled for abstraction-then-sampling.

\paragraph*{Step 3b. Statistical analysis.} 
The tool performs all the statistical analyses and abstraction-then-sampling analyses, using the saved program states from Step 3a as starting point of each component to analyze statistically. 
The sample size for each simulation is automatically decided by using heuristics to have better accuracy with less sample size, as explained in Sections~\ref{sec:compositional} and~\ref{sec:estimate-known-system}.
The results of each analysis is stored as an appropriate joint probability sub-distribution between secret and observable values.
In the example in Fig.~\ref{fig:randomwalkcode}, each of the 7 symbolic paths scheduled for abstraction-then-sampling gets analyzed with the technique. 
For each of the symbolic paths \toolname choses a value of the secret and samples only that one, then applies the results for all secret values of the component. 
\toolname recomputes the assignment of samples to the components for each 10\% of the total samples, following the optimal sample sizes computed for abstraction-then-sampling components in Theorem~\ref{thm:adaptive-sampling-ATS}.

\begin{figure}
\begin{tcolorbox}
\begin{enumerate}
\item Build the control flow graph of the system. 
\item Mark all possible components based on each conditional branching. 
Each possible component must be a terminal as explained in Section~\ref{sec:compositional}.
\item For each possible component $S_i$, check whether it is deterministic or not (by syntactically checking an occurrence of a probabilistic assignment or a probabilistic function call). If it is, mark the component for precise analysis, since deterministic systems necessarily have a smaller number of traces.
\item For each possible component $S_i$, check whether $S_i$'s output variables are independent of its input variables inside $S_i$ (by \emph{qualitative} information flow). If so, mark that the abstraction-then-sampling technique in Section~\ref{subsec:known-subsystem} is to be used on the component, meaning that component $S_i$ will be sampled on a single secret input value and the results will be applied to all secret values.
\item For each $S_i$, estimate an approximate range size $\#Z_i$ of its internal variables and $\#Y_i$ of its observable variables. 
\item Looking from the leaves to the root of the control flow graph, 
estimate the cost of statistical and precise analyses,
decide the decomposition into components, and
mark each component for the cheaper analysis between the two. 
(For example, use a heuristics that marks each component $S_i$ for precise analysis if $\#Z_i\leq \#X_i$ and for statistical analysis otherwise.)

\item Join together adjacent components if they are marked for precise analysis, or if they are marked for statistical analysis and have the same input and output ranges.
\item For each component, perform precise analysis or statistical analysis as marked.
\end{enumerate}
\end{tcolorbox} 
\caption{Procedure for the decomposition of a system into components given its source code described in Section~\ref{sec:overview:hybrid}. The actual implementation here uses the control flow graph of the system to guide the procedure.}
\label{proc:decompose}
\end{figure}

\subsubsection*{Step 4: Leakage Estimation}
In this step the tool aggregates all the data collected by the precise and statistical analyses (performed in Steps 3) and estimates the Shannon leakage of the input program, together with evaluation of the estimation.
This is explained in detail together with the program decomposition in Section~\ref{subsec:implement-components}.

\subsection{On the Division into Components of the Random Walk Benchmark}\label{subsec:implement-components}

In this section we briefly discuss how the Random Walk example in Fig.~\ref{fig:randomwalkcode} can be divided into components using the procedure in Fig.~\ref{proc:decompose}. The procedure in Fig.~\ref{proc:decompose} shows in more detail how to implement the procedure described in Section~\ref{sec:overview:hybrid}. In the implementation, the construction of a the control flow graph of the system is used to guide the division into components.

The control flow graph generated by \toolname is shown in Fig.~\ref{fig:randomwalkcfg}. Note that \toolname has added a \texttt{simulate-abs} statement to the code, visible in the control flow graph. 

The control flow graph in Fig.~\ref{fig:randomwalkcfg} helps understanding how \toolname has implicitly divided the program into components. 
Note that the \texttt{simulate-abs} node marked in red has 7 in-edges. Each of these edges corresponds to a different symbolic path (i.e., a set of execution traces following the same edge), with different possible values for the secret variable \texttt{sec} and the internal variable \texttt{loc}. 
\toolname has determined heuristically that at Line 23 the number of possible values of the observable variables (\texttt{TOT_OBS = 1}) is smaller than the number of possible values of internal variables (\texttt{TOT_INT = 601}), hence statistical simulation will be more efficient than precise analysis on these components. 
Also, in the code after line 23, \toolname has determined by syntactically analyzing the code that the values of the observable variables do not depend on the secret, hence each row of the sub-channel matrix for each of these components is identical, 
much like component $S_1$ in Fig.~\ref{fig:jointalt}. 

Hence, the abstraction-then-sampling technique of Section~\ref{subsec:known-subsystem} can be applied, meaning that the behavior of each component will be simulated only for a single value of the secret in the set of possible secret values of the component, and the results will be applied to each row of the channel matrix.

Now that the analysis has gathered all the necessary information, \toolname computes the leakage of the system under analysis. More specifically, it constructs an (approximate) joint posterior distribution of the secret and observable values of the input program from all the collected data produced by Step 3, as explained in Section~\ref{sec:overview:hybrid}.
Then the tool estimates the Shannon leakage value from the joint distribution, including bias correction (see Sections~\ref{sec:compositional} and~\ref{sec:estimate-known-system}).
Finally, a 95\% confidence interval for the estimated leakage value is computed to roughly evaluate the quality of the analysis.

In the example in Fig.~\ref{fig:randomwalkcode}, \toolname{} outputs the prior Shannon entropy 8.9658, the posterior Shannon entropy 7.0428, the Shannon leakage (after bias correction) 1.9220, and the confidence interval [1.9214, 1.9226].
\section{Evaluation}\label{sec:evaluation}

We evaluate experimentally the effectiveness of our hybrid method compared to the state of the art.
We first discuss the cost and quality of the estimation, then test the hybrid method against fully precise\slash fully statistical analyses on Shannon leakage benchmarks.

\subsection{On the Tradeoff between the Cost and Quality of Estimation}
\label{subsec:eval:quality}
In the hybrid statistical estimation, the estimate takes different values probabilistically, 
because it is computed 
from a set of traces that are generated by executing a probabilistic system.
Fig.~\ref{fig:MI-case1} shows the sampling distribution of the mutual information estimate of the joint distribution in Fig.~\ref{fig:joint} in Section~\ref{sec:intro}.
The graph shows the frequency 
(on the $y$ axis) of the mutual information estimates (on the $x$ axis) when performing the estimation 1000 times.
In each estimation we perform precise analysis on the component $T$ and statistical analysis on $S_1$ and $S_2$ (with a sample size of 5000).
The graph is obtained from 1000 samples each of which is generated by combining precise analysis on a component and statistical analysis on 2 components (using 5000 randomly generated traces).
As shown in Fig.~\ref{fig:MI-case1} the estimate after the correction of bias by Theorem~\ref{lem:MI:expectation:general} is closer to the true value.
The estimate is roughly between the lower and upper bounds of the 95\% confidence interval calculated using Theorem~\ref{lem:MI:variance:general}.

The interval size depends on the sample size in statistical analysis as shown in Fig.~\ref{fig:MI-sample-size}.
In Fig.~\ref{fig:MI-sample-size} we illustrated the relationships between the size of the confidence interval and the sample size in the statistical analysis.
We used an example with the randomly generated $10 \times 10$ channel matrix presented in Fig.~\ref{fig:bigmatrix} and the uniform prior.
The graph shows the frequency (on the $y$ axis) of the corrected mutual information estimates (on the $x$ axis) that are obtained by estimating the mutual information value 1000, 5000 and 10000 times.
When the sample size is $k$ times larger then the confidence 
interval is $\sqrt{k}$ times narrower.
\begin{figure*}
 \begin{minipage}{0.43\hsize}
  \begin{center}
  \mbox{\raisebox{-10pt}{\includegraphics[width=1.01\textwidth]{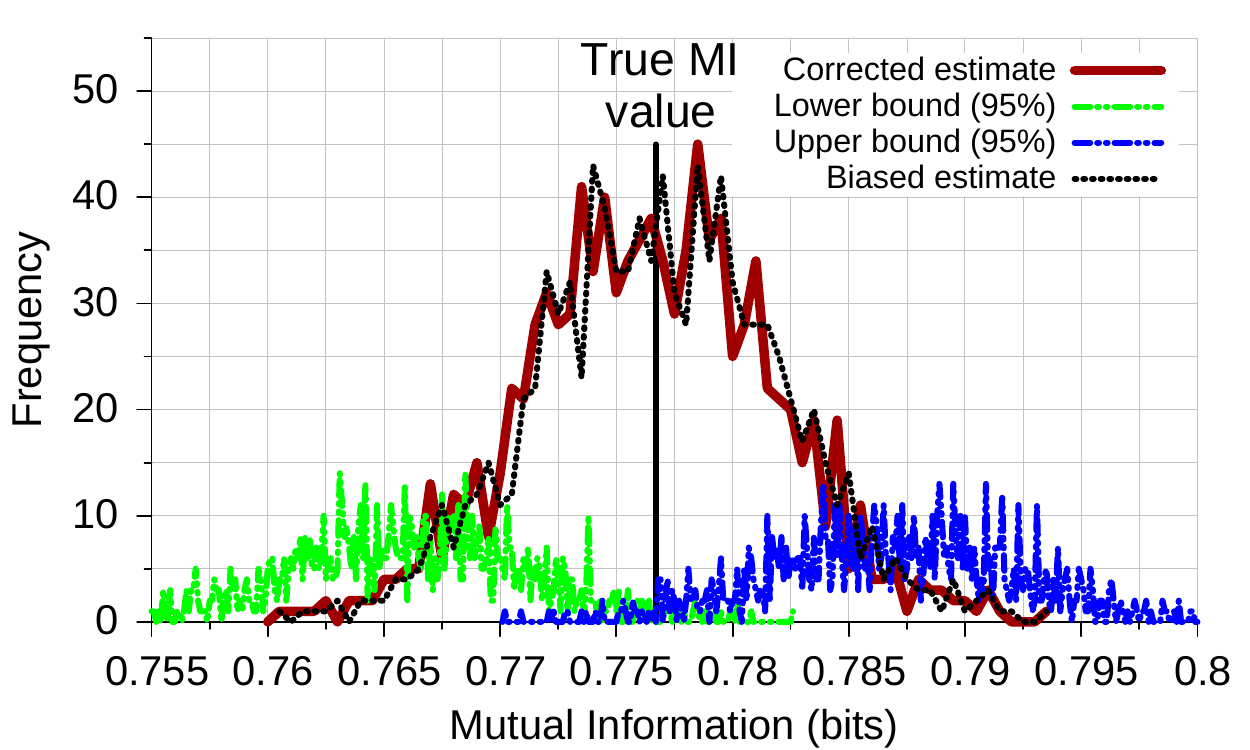}}}
  \end{center}
  \caption{Distribution of mutual information estimate and its confidence interval.}
  \label{fig:MI-case1}
 \end{minipage}~~~~~
 \begin{minipage}{0.49\hsize}
  \begin{center}
  %
\centering
\begin{scriptsize}
\tabcolsep = 0.55mm
\renewcommand{\arraystretch}{1.15}
\begin{tabular}{c c c c c c c c c c c c c}
 & & \multicolumn{10}{c}{\bf Observable} \\
&  & {\bf 0} & {\bf 1} & {\bf 2} & {\bf 3} & {\bf 4} & {\bf 5} & {\bf 6} & {\bf 7} & {\bf 8} & {\bf 9} \\
\midrule
\centering \multirow{13}{*}{\rotatebox{90}{\mbox{\bf Secret\hspace{0.25cm}}}~~~}
& {\bf 0} & 0.2046  & 0.1102  & 0.0315  & 0.0529  & 0.1899  & 0.0064  & 0.0791  & 0.1367  & 0.0386  & 0.1501 \\
\cmidrule{2-12}
& {\bf 1} & 0.0852  & 0.0539  & 0.1342  & 0.0567  & 0.1014  & 0.1254  & 0.0554  & 0.1115  & 0.0919  & 0.1844 \\
\cmidrule{2-12}
& {\bf 2} & 0.1702  & 0.0542  & 0.0735  & 0.0914  & 0.0639  & 0.1322  & 0.1119  & 0.0512  & 0.1172  & 0.1343 \\
\cmidrule{2-12}
& {\bf 3} & 0.0271  & 0.1915  & 0.0764  & 0.1099  & 0.0982  & 0.0761  & 0.0843  & 0.1364  & 0.0885  & 0.1116 \\
\cmidrule{2-12}
& {\bf 4}  & 0.0957  & 0.1977  & 0.0266  & 0.0741  & 0.1496  & 0.2177  & 0.0610  & 0.0617  & 0.0841  & 0.0318 \\
\cmidrule{2-12}
& {\bf 5} & 0.0861  & 0.1275  & 0.1565  & 0.1193  & 0.1321  & 0.1716  & 0.0136  & 0.0984  & 0.0183  & 0.0766 \\
\cmidrule{2-12}
& {\bf 6} & 0.0173  & 0.1481  & 0.1371  & 0.1037  & 0.1834  & 0.0271  & 0.1289  & 0.1690  & 0.0036  & 0.0818 \\
\cmidrule{2-12}
& {\bf 7} & 0.0329  & 0.0825  & 0.0333  & 0.1622  & 0.1530  & 0.1378  & 0.0561  & 0.1479  & 0.0212  & 0.1731 \\
\cmidrule{2-12}
& {\bf 8} & 0.1513  & 0.0435  & 0.0527  & 0.2022  & 0.0189  & 0.2159  & 0.0718  & 0.0063  & 0.1307  & 0.1067 \\
\cmidrule{2-12}
& {\bf 9} & 0.0488  & 0.1576  & 0.1871  & 0.1117  & 0.1453  & 0.0349  & 0.0549  & 0.1766  & 0.0271  & 0.056
\end{tabular}
\end{scriptsize}
\caption{Channel matrix for the experiments in Section~\ref{subsec:eval:quality}.}\label{fig:bigmatrix}

  \end{center}
 \end{minipage}
\end{figure*}

\begin{center}
\begin{figure*}
  \centering
  \subfloat[Estimates and sample sizes.]{\includegraphics[width=0.43\textwidth]{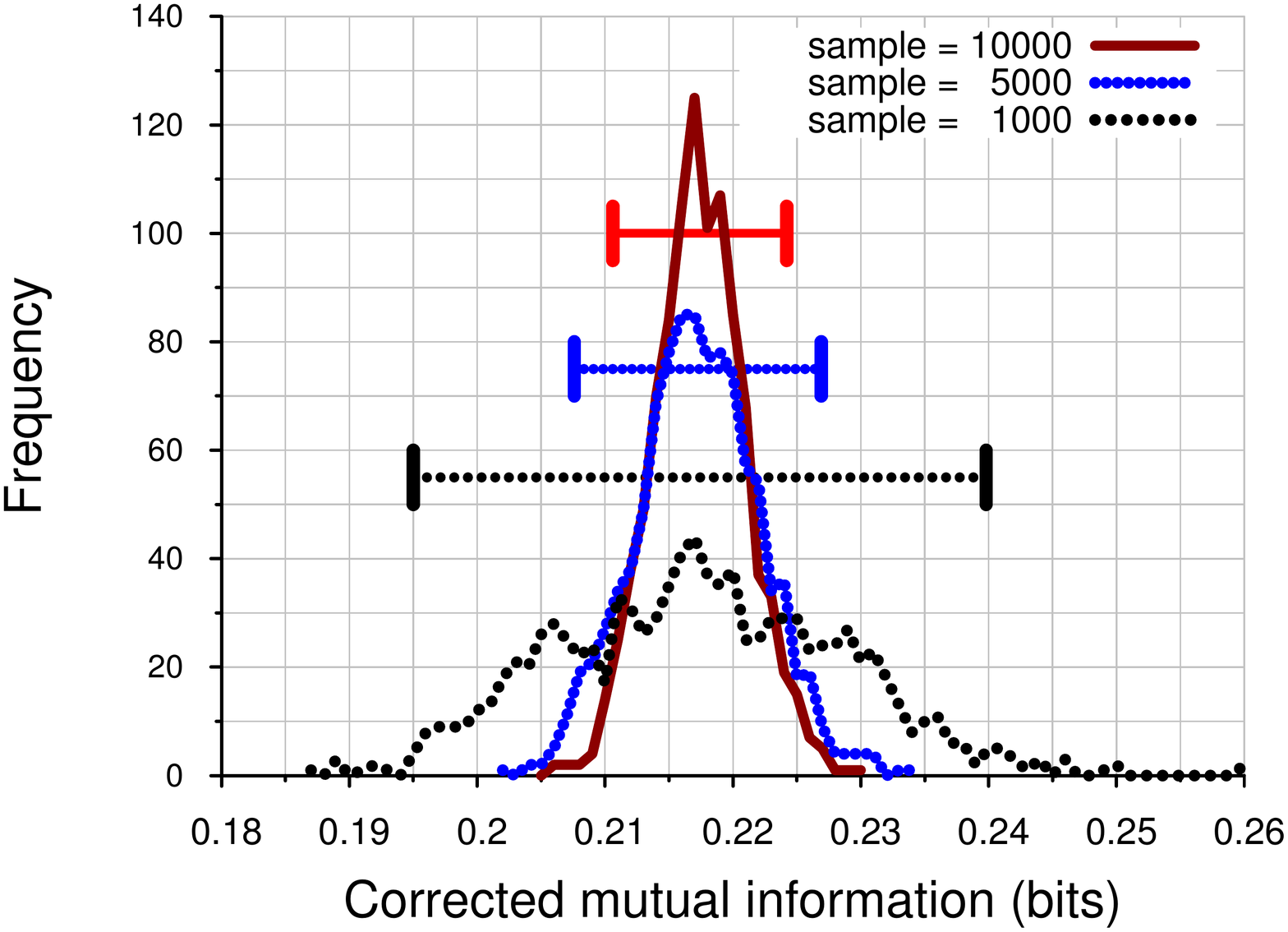}\label{fig:MI-sample-size}}~~~~~~~~
  \subfloat[Estimates and the ratio of precise analysis.]{\includegraphics[width=0.43\textwidth]{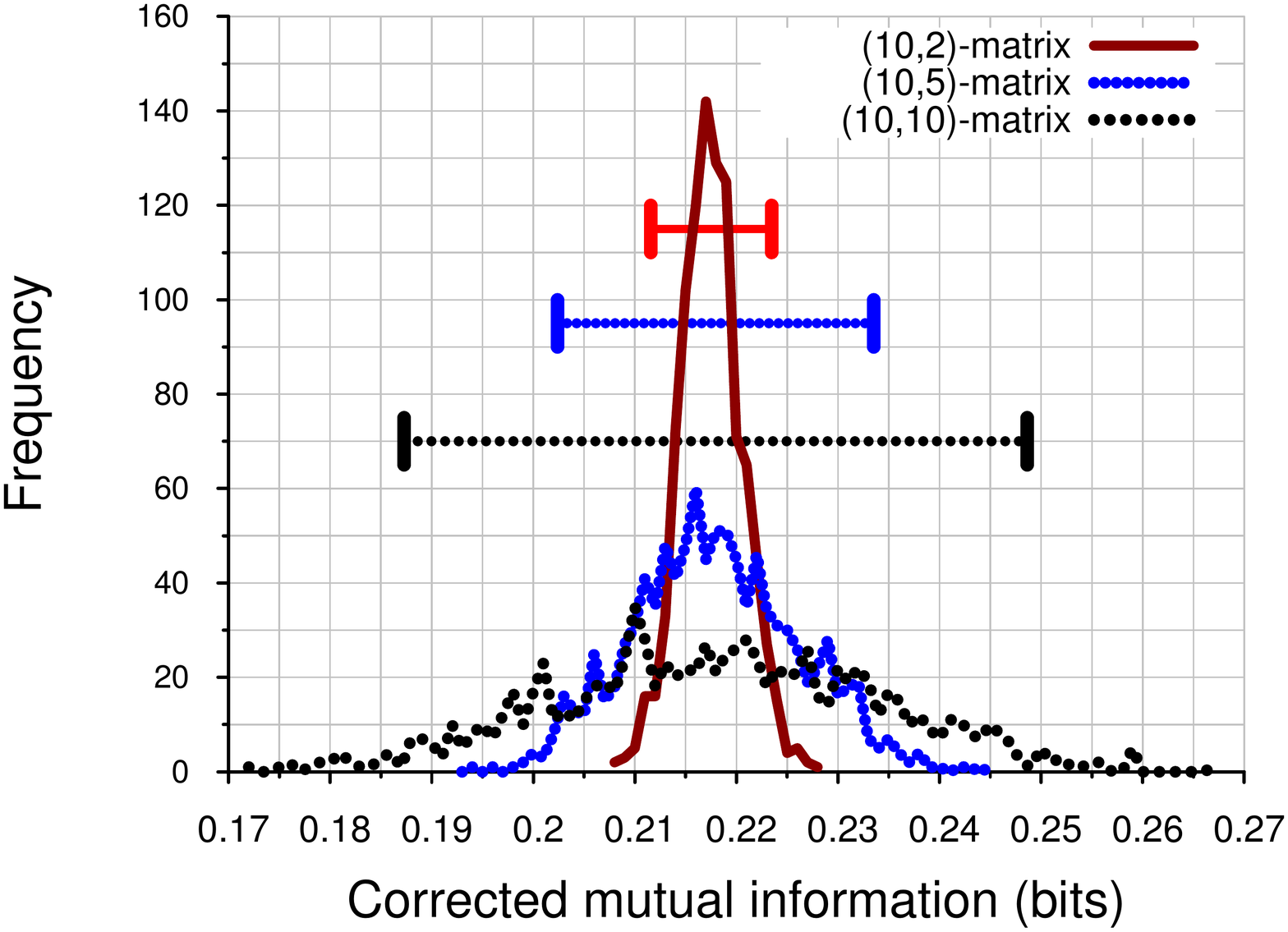}\label{fig:MI-domain-obs}}~~~~~~~~
  \caption{Smaller intervals when increasing the sample size or the ratio of precise analysis.\label{fig:MI-confidence-interval-changes}}
\end{figure*}
\end{center}

The interval size also depends on the amount of precise analysis as shown in Fig.~\ref{fig:MI-domain-obs}.
If we perform precise analysis on larger components, then the sampling distribution becomes more centered (with shorter tails) and the confidence interval becomes narrower.
For instance, in Fig.~\ref{fig:MI-domain-obs} we illustrated the relationships between the size of the confidence interval and the amount of precise analysis.
The graph shows the frequency (on the $y$ axis) of the corrected mutual information estimates (on the $x$ axis) that are obtained by estimating the mutual information value 1000 times when statistical analysis is applied to a $10 \times 2$, $10 \times 5$ and $10 \times 10$ sub-matrix of the full $10 \times 10$ matrix.
Using statistical analysis only on a smaller component ($10 \times 2$ sub-matrix) yields a smaller confidence interval than using it on the whole system ($10 \times 10$ matrix).
More generally, if we perform precise analysis on larger components, then we have a smaller confidence interval.
This means that the hybrid approach produces better estimates than the state of the art in statistical analysis.
Due to the combination with precise analysis, the confidence interval estimated by our approach is smaller than \LeakWatch{}~\cite{DBLP:conf/esorics/ChothiaKN14} for the same sample size.

\subsection{Shannon Leakage Benchmarks}\label{sec:shortexamples}

We compare the performance of our hybrid method with fully precise/statistical analysis on Shannon leakage benchmarks.
Our implementations of precise and statistical analyses are variants of the state-of-the art tools QUAIL{}~\cite{DBLP:conf/cav/BiondiLTW13,quail:www} and \LeakWatch{}~\cite{DBLP:conf/esorics/ChothiaKN14,leakwatch:www} respectively.
All experiments are performed on an Intel i7-3720QM 2.6GHz eight-core machine with 8GB of RAM running Fedora 21.

\subsubsection{Random Walk}\label{subsec:result-randomwalk} 

We analyze the random walk example presented in Section~\ref{subsec:implementation-example}
for different values of number of steps $MAX$. 
We plot the computation times and the errors in leakage values computed  by the three different methods
in the graphs presented in Fig~\ref{fig:results:random-walk}.
These graph show again that the execution time of precise analysis grows exponentially to the number of steps $MAX$, while \toolname and fully randomized analysis do not require much time even for large values of $MAX$. In the fully randomized analysis the error is always much larger than when using \toolname.

\subsubsection{Reservoir Sampling}
\SetAlFnt{\footnotesize}

\begin{minipage}[t]{0.29\textwidth}
\begin{algorithm}[H]
\scriptsize
const N; \tcp{number of elements}
const K; \tcp{selection}
secret array[$N$] of int1 s;\\
observable array[$K$] of int1 r;\\
public int32 j := 0;\\
\lFor{$i$ in $[0,K\text{-}1]$}{
  r[$i$] := s[$i$]
}
\For{$i$ in $[K,N\text{-}1]$}{
  $j$ := \random(0,$i$);\\
  \lIf{$j$<$K$}{r[$j$] := s[$i$]}
}
\end{algorithm}
\captionof{figure}{Reservoir sampling.}\label{alg:reservoir}
\end{minipage}
\hspace{3mm}
\begin{minipage}[t]{0.67\textwidth}
The reservoir sampling problem \cite{Vitter:1985:RSR:3147.3165} consists of selecting $K$ elements randomly from a pool of $N>K$ elements. 
We quantify the information flow of the commonly-used \emph{Algorithm R} \cite{Vitter:1985:RSR:3147.3165}, shown in Fig~\ref{alg:reservoir}, for various values of $N$ and $K=N/2$.
In the algorithm, the first $K$ elements are chosen as the sample, then each other element has a probability to replace one element in the sample.

We plot the computation times and the errors in leakage values computed by the three different methods in the graphs presented in Fig~\ref{fig:results:reservoir}.
It shows that \toolname hybrid approach performs faster than the full simulation approach and gives more precise results.
Compared to the precise analysis, the hybrid approach run faster when increasing the model's complexity.
\end{minipage}

\begin{figure}
  \centering
  \subfloat{\label{fig:random-walk_time}\includegraphics[width=0.41\textwidth]{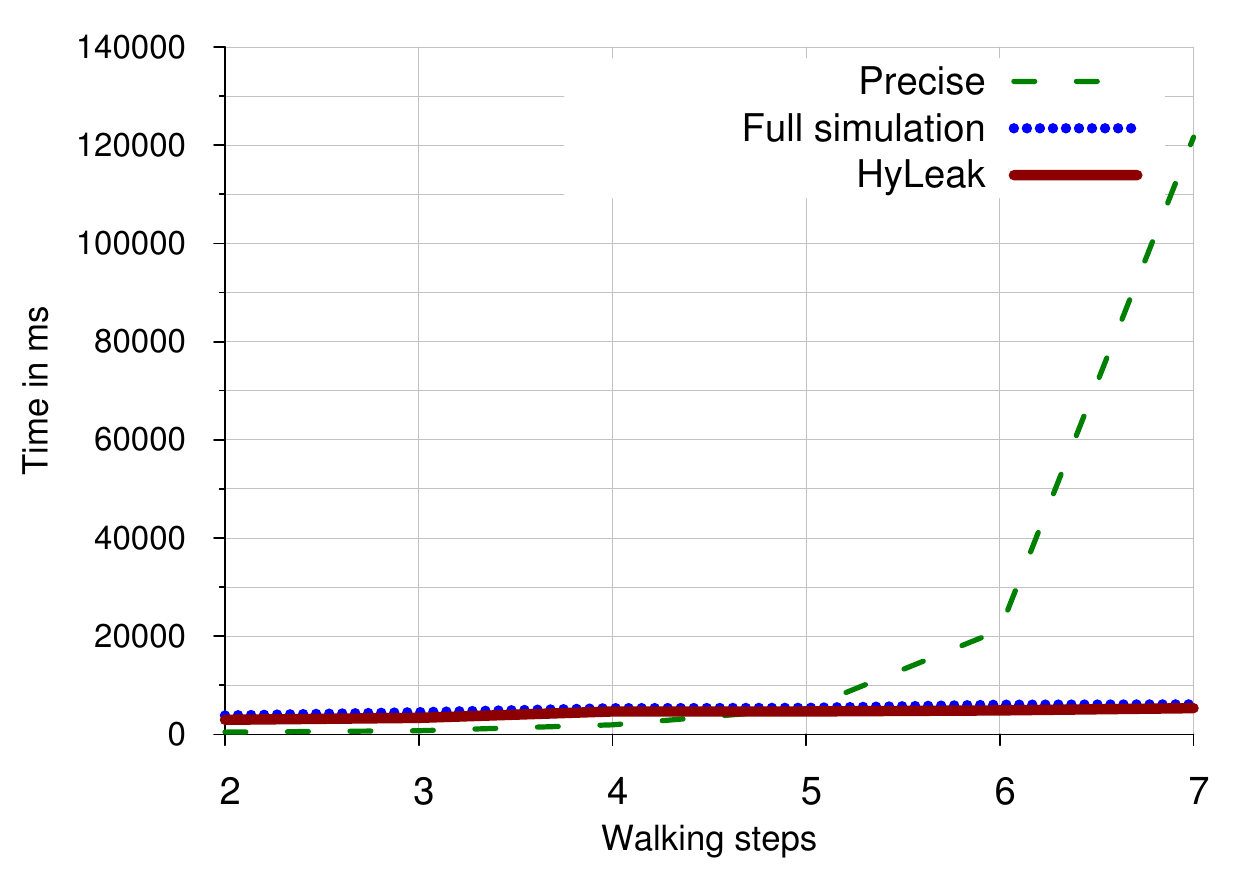}}~~~~~~~~
  \subfloat{\label{fig:random-walk_err}\includegraphics[width=0.41\textwidth]{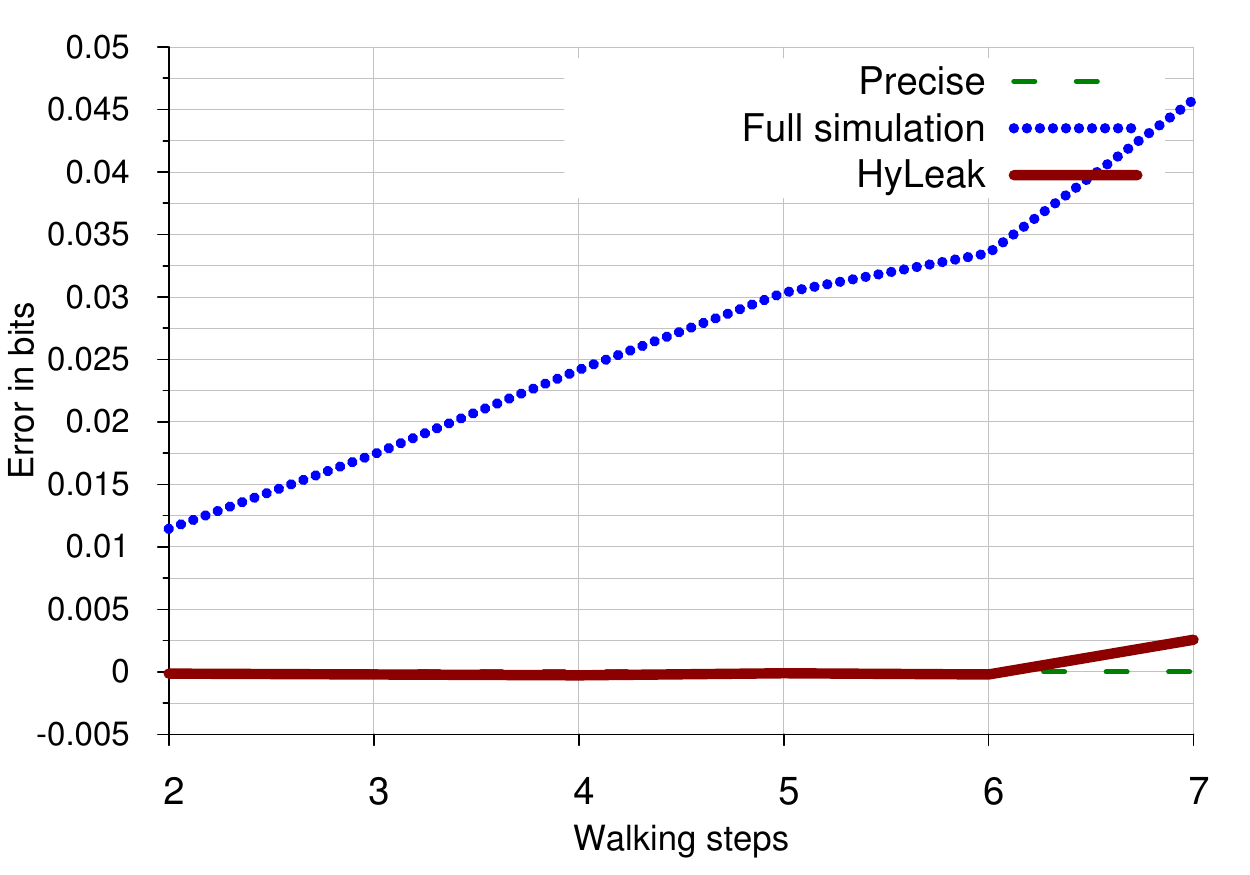}}~~~~~~~~
  \caption{Random walk experimental results.}
  \label{fig:results:random-walk}
\end{figure}

\begin{figure}
  \centering
  \subfloat{\label{fig:reservoir-time}\includegraphics[width=0.41\textwidth]{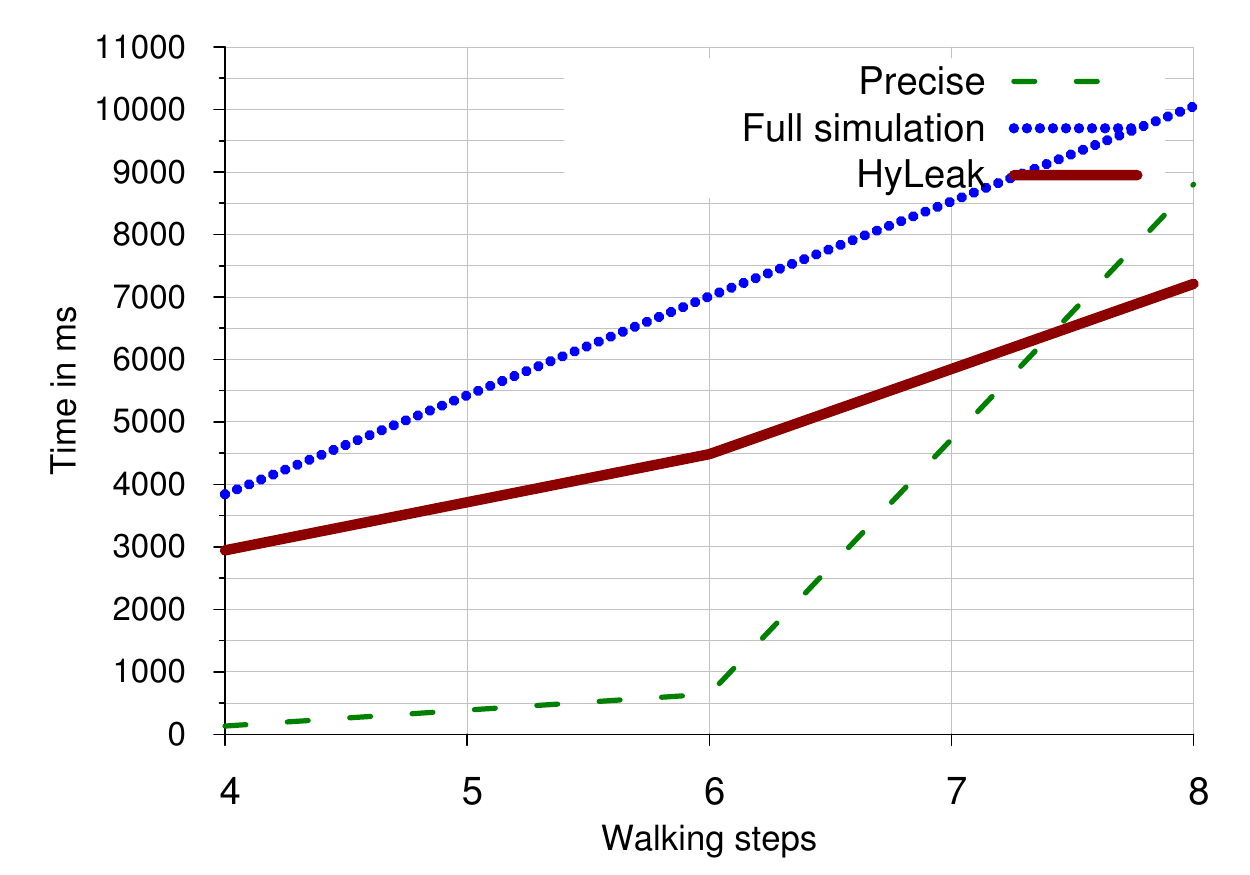}}~~~~~~~~
  \subfloat{\label{fig:reservoir_err}\includegraphics[width=0.41\textwidth]{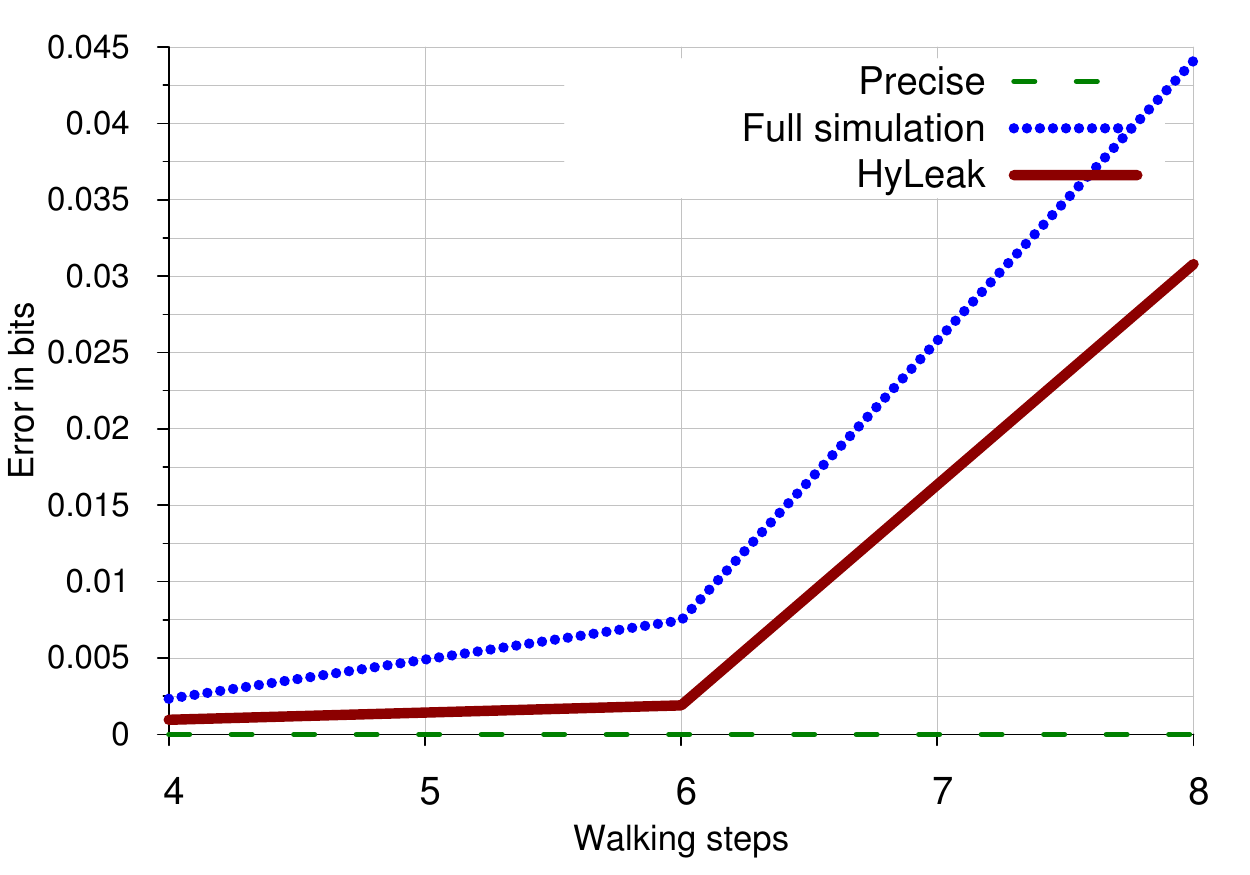}}
  \caption{Reservoir sampling experimental results.}
  \label{fig:results:reservoir}
\end{figure}

\begin{figure}
\begin{algorithm}[H]
\scriptsize
const $N$:=4; \tcp{number of cryptographers at each table}
const $M$:=8; \tcp{total number of cryptographers}
\tcc{these bits represent the coin tosses for the three tables}
public array [$N$] of int1 $coinA$, $coinB$, $coinC$;\\

public int32 $lies$; \tcp{this is for the liar}

\tcc{these bits represent the bits declared by the three cryptographers at each table}
public array [$N$] of int1 $declA$, $declB$, $declC$;\\

\tcc{these are the outputs at each table}
observable int1 $outputA$ := 0;\\
observable int1 $outputB$ := 0;\\
observable int1 $outputC$ := 0;\\

secret int32 $h$ := $[0,M]$; \tcp{the secret has M+1 possible values}

lies := \random($1,M$); 
\lFor{$c$ in $coinA$}{$c$ := \randombit($0.5$)}
\lFor{$c$ in $coinB$}{$c$:=\randombit($0.5$)}
\lFor{$c$ in $coinC$}{$c$:=\randombit($0.5$)}

\For{$i$ in $[0,N-1]$}{
  $declA[i]$:=$coinA[i] \xor coinA[(i+1)\% N]$;\\
  \lIf{$h==i+1$}{$declA[i]$:=~$!declA[i]$}
  \lIf{$lies==i+1$}{$declA[i]$:=~$!declA[i]$}
  $outputA$ := $outputA \xor declA[i]$;\\
  $declB[i]$:=$coinB[i] \xor coinB[(i+1)\%N]$;\\
  \lIf{$h==i+3$}{$declB[i]$:=~$!declB[i]$;}
  \lIf{lies==i+3}{$declB[i]$:=~$!declB[i]$;}
  $outputB$ := $outputB \xor declB[i]$;\\
  $declC[i]$ := $coinC[i] \xor coinC[(i+1)\%N]$;\\
  \lIf{$h==i+5$}{$declC[i]$:=~$!declC[i]$}
  \lIf{$lies==i+5$}{$declC[i]$:=~$!declC[i]$}
  $outputC$ := $outputC \xor declC[i]$;
}
\end{algorithm}
\caption{Multiple lying cryptographers.}\label{alg:multiple_lying_crypto}
\end{figure}

\subsubsection{Multiple Lying Cryptographers Protocol}\label{sec:application-QIF}
The lying cryptographers protocol is a variant of the dining cryptographer multiparty computation protocol \cite{Chaum88thedining} in which a randomly-chosen cryptographer declares the opposite of what they would normally declare, i.e. they lie if they are not the payer, and do not lie if they are the payer.
We consider three simultaneous lying cryptographers implementation in which 8 cryptographers run the protocol on three separate overlapping tables $A$, $B$ and $C$ with 4 cryptographers each. Table $A$ hosts cryptographers $1$ to $4$, Table $B$ hosts cryptographers $3$ to $6$, and Table $C$ hosts cryptographers $5$ to $8$. The identity of the payer is the same in all tables.

The division into components is executed following the principles in Section~\ref{subsec:implement-components}. 
The hybrid approach divides the protocol into 8 components, one for each possible liar. Then each component is analyzed statistically.

The results of the experiment are summarized in Table~\ref{tab:benchmark_results}.
Note that this example has some zeroes in the  probabilities in the channel matrix,
which makes the tool \LeakWatch{}~\cite{DBLP:conf/esorics/ChothiaKN14,leakwatch:www} compute an incorrect result due to incorrect bias estimation,
as explained in Section~\ref{subsec:previous-known-prior}. In fact, the leakage value obtained by \LeakWatch{} using Corollary~\ref{lem:MI:expectation:Chothia} is 0.36245, which is far from the correct value of 0.503 in Table~\ref{tab:benchmark_results}. On the other hand, the error of our calculation using Proposition~\ref{lem:MI:expectation:KnownPrior} is 1.87e-4 even with fully statistical analysis.

\begin{figure}
  \centering
  \subfloat{\label{fig:window_time}\includegraphics[width=0.41\textwidth]{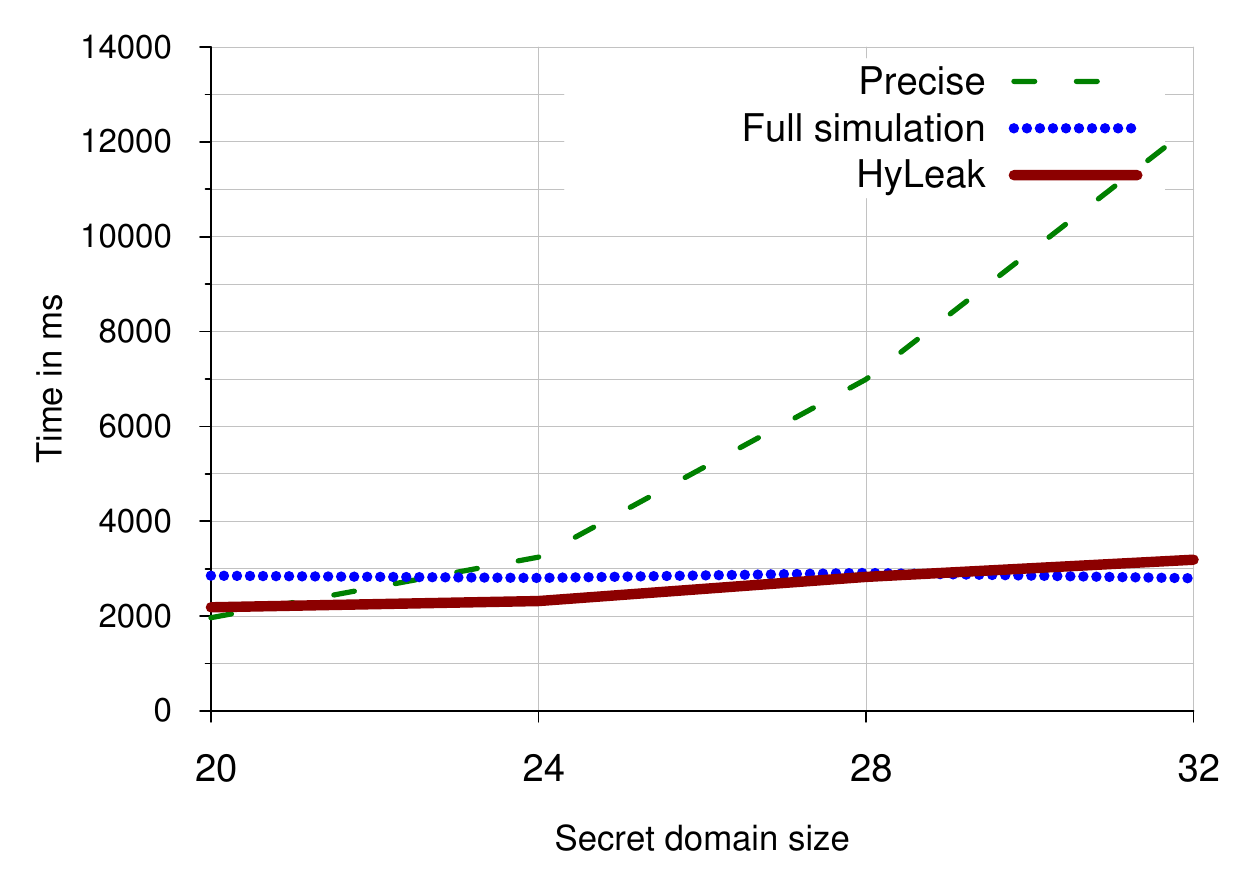}}~~~~~~~~
  \subfloat{\label{fig:window_err}\includegraphics[width=0.41\textwidth]{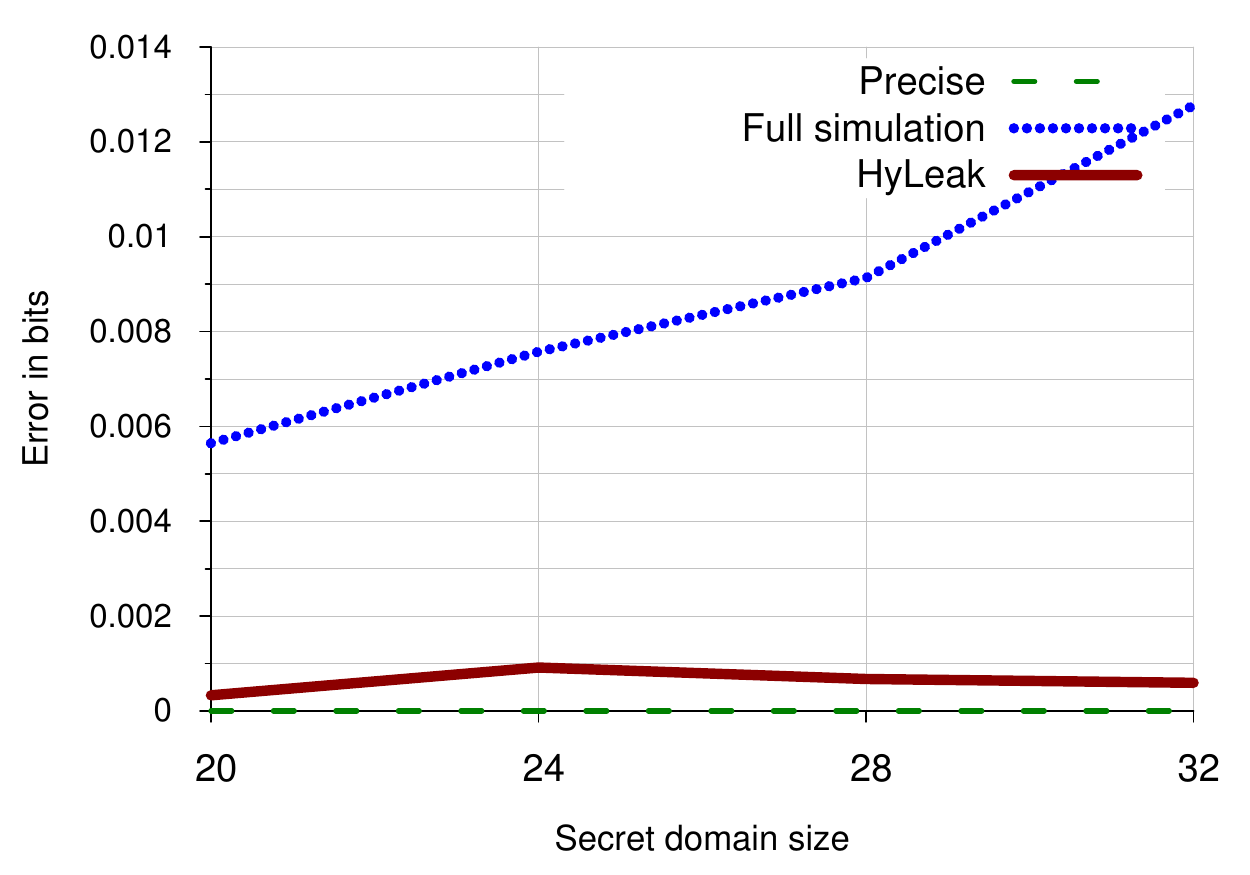}}~~~~~~~~
  \caption{Shifting window experimental results.}
  \label{fig:results:window}
\end{figure}

\subsubsection{Shifting Window}

\begin{minipage}[t]{0.37\textwidth}
\begin{algorithm}[H]
\scriptsize
const N; \tcp{number of elements}
const W; \tcp{window size}
secret int32 $sec$ := [0,$N$-1];\\
observable int32 $obs$;\\
public int32 $minS$, $sizeS$, $sizeS$, $minO$, $sizeO$;\\
$minS$ := \random($0$,$N$-$W$-$1$);\\
\eIf{$sec \geq minS$}{
  $sizeS$ := \random($1$,$W$);\\
  \eIf{$sec \leq minS\text{+}sizeS$}{
    $minO$ := \random($0$,$N$-$W$-$1$);\\
    $sizeO$ := \random($1$,$W$);\\
    $obs$ := \random($minO$,$minO$+$sizeO$);
  }{
    $obs$ := \random($0$,$N$-$1$);
  }
}{
  $obs$ := \random($0$,$N$-$1$);
}
\end{algorithm}
\captionof{figure}{Shifting Window.}\label{alg:window}
\end{minipage}
\hspace{4mm}
\begin{minipage}[t]{0.59\textwidth}
In the Shifting Window example (Fig.~\ref{alg:window}) the secret \texttt{sec} can take $N$ possible values, and an interval (called a ``window'') in the secret domain is randomly selected from $1$ to $W$.
If the value of the secret is inside the window, then another window is randomly chosen in a similar way and the program outputs a random value from this second window.
Otherwise, the program outputs a random value over the secret domain.\\

In Fig.~\ref{fig:results:window} we present the results of experiments on the shifting window when increasing the size of the secret domain.
The execution time of precise analysis grows proportionally to the secret domain size $N$ while \toolname and fully randomized analysis do not require much time for a larger $N$.
In the fully randomized analysis the error from the true value grows rapidly while in using \toolname the error is much smaller.
\end{minipage}

\subsubsection{Probabilistically Terminating Loop}\label{subsec:prob-terminating}

\begin{minipage}[t]{0.37 \textwidth}
\begin{algorithm}[H]
\scriptsize
const $N$; \tcp{number of secrets}
const $BOUND$;\\
secret int32 $sec$ := $[0,N]$;\\
observable int32 $obs$;\\
public int32 $time$ := 0;\\
public int2 $terminate$ := 0;\\
public int32 $rand$;\\
\While{$terminate \neq 1$}{
  $rand$ := \random($1$,$N$);\\
  \lIf{$sec \leq rand$}{$terminate$ := $1$}
  $time$ := $time$+$1$;
}
\eIf{$time < BOUND$}{
  $obs$ := $time$;
}{
  $obs$ := $BOUND$;
}
\end{algorithm}
\captionof{figure}{Probabilistically Terminating Loop.}\label{fig:prob-terminating}
\end{minipage}
\hspace{4mm}
\begin{minipage}[t]{0.59\textwidth}
The tool \toolname can analyze programs that terminate only probabilistically.
For instance, the program shown in Fig.~\ref{fig:prob-terminating} has a loop that terminates depending on the randomly generated value of the variable \texttt{rand}.
No previous work has presented an automatic measurement of information leakage by probabilistic termination, as precise analysis cannot handle non-terminating programs, which typically causes non-termination of the analysis of the program.
On the other hand, the stochastic simulation of this program supported in \toolname terminates after some  number of iterations in practice although it may take long for some program executions to terminate.

We analyze this model with the hybrid and full simulation approaches only, as the precise analysis does not terminate.
The results are given in Table~\ref{tab:benchmark_results}. It shows that also with this problem the hybrid approach performs faster than the full simulation approach.
\end{minipage}

\subsubsection{Smart Grid Privacy}\label{subsec:smart-grid}

A smart grid is an energy network where users (like households) may consume or produce energy.
In Fig.~\ref{fig:smart-grid} we describe a simple model of a smart grid using the \toolname language. This example is taken from~\cite{DBLP:conf/spin/BiondiLQ15}.
The users periodically negotiate with a central aggregator in charge
of balancing the total consumption among several users. In practice each user
declares to the aggregator its consumption plan. The aggregator sums up
the consumptions of the users and checks if it falls within admitted bounds.
If not it answers to the users that the consumption is too low or too high by a certain amount, such that they adapt their demand.
This model raises some privacy issues as some attacker can try to guess the consumption of a user, 
and for instance infer whether or not this particular user is at home.

In Fig.~\ref{fig:results:smartgrid} we present the experiment results of this smart grid example for different numbers of users.
\toolname takes less time than both fully precise analysis and fully randomized analysis (as shown in the left figure).
Moreover, it is closer to the true value than fully randomized analysis
especially when the number of users is larger (as shown in the right figure).

\begin{figure}
\begin{algorithm}[H]
\scriptsize
const $N$:=9;  \tcp{the total number of users}
const $S$;  \tcp{the number of users we care about}
const $C$:=3;  \tcp{the possible consumptions level}
const $M$:=0;  \tcp{the consumption level of the attacker}
const $LOWT$:=2; \tcp{the lower threshold}
const $HIGHT$:=9; \tcp{the upper threshold}
\tcc{the observable is the order given by the control system}
observable int32 $order$;\\
observable int1 $ordersign$;\\
\tcc{the secret is the consumption of each user we care about}
secret array [$S$] of int32 $secretconsumption$ := $[0,C\text{-}1]$;\\
\tcc{the other consumptions are just private}
private array [$N$-($S$+$1$)] of int32 $privateconsumption$ := [$0$,$C$-$1$];\\
public int32 $total$ := $M$; \tcp{this is the projected consumption}
\tcc{count the secret consumptions}
\For{ $i$ in [$0$,$S$-$1$]}{
  \For{ $j$ in [$0$,$C$-$1$]}{
    \lIf{$secretconsumption[i]==j$}{$total$ := $total+j$}
  }
}
\tcc{count the private consumptions}
\For{ $i$ in [$0$,$N$-$S$-$1$]}{
  \For{ $j$ in [$0$,$C$-$1$]}{
    \lIf{$privateconsumption[i]==j$}{$total$ := $total+j$}
  }
}
\uIf{$total<LOWT$}{
  $order$ := $LOWT - total$;\\
  $ordersign$ := $0$;\\
}\uElseIf{$total > HIGHT$}{
  $order$ := $total - HIGHT$;\\
  $ordersign$ := $1$;\\
}\Else{
  $order$ := $0$;\\
  $ordersign$ := $0$;\\
}
\end{algorithm}
\caption{Smart Grid Example.}
\label{fig:smart-grid}
\end{figure}

\begin{figure}
  \centering
  \subfloat{\label{fig:smartgrid-time}\includegraphics[width=0.41\textwidth]{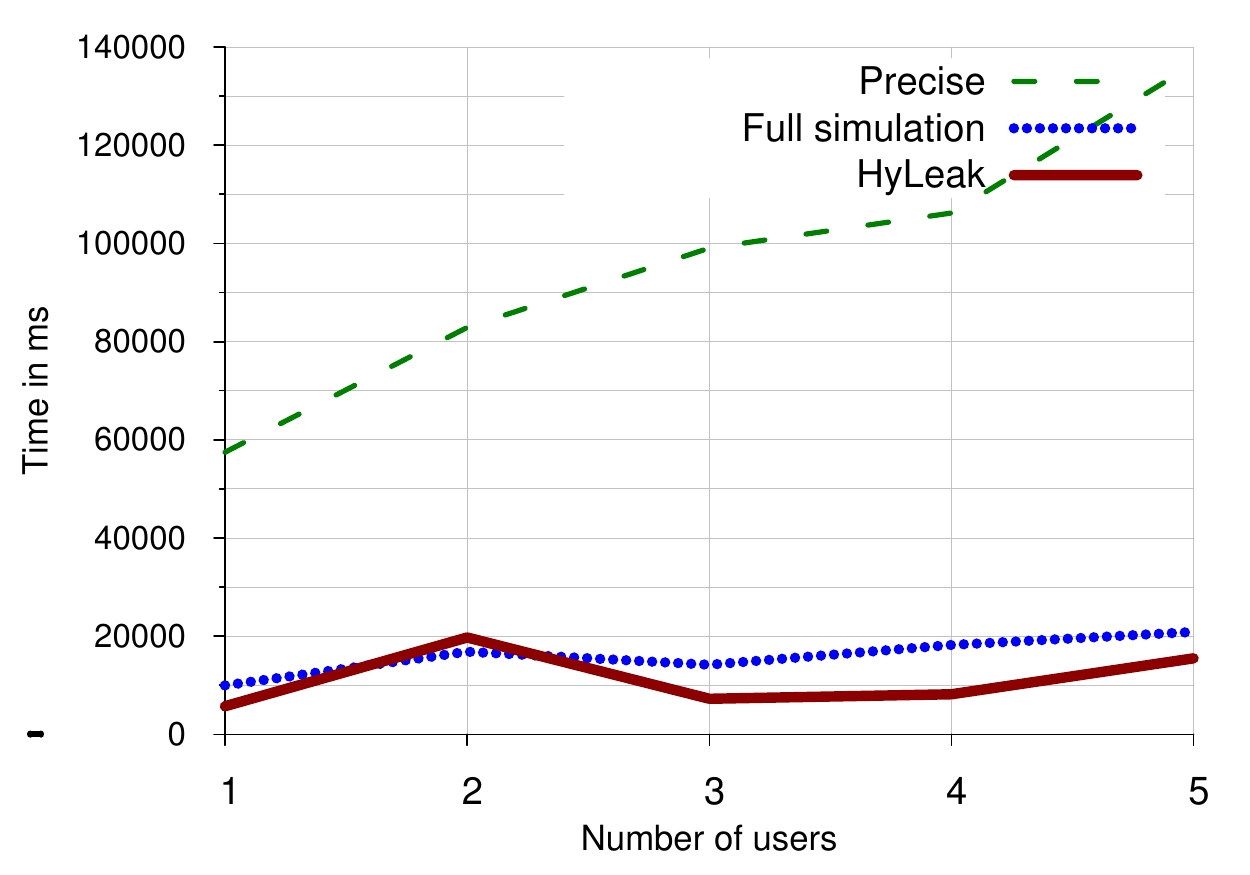}}~~~~~~~~
  \subfloat{\label{fig:smartgrid_err}\includegraphics[width=0.41\textwidth]{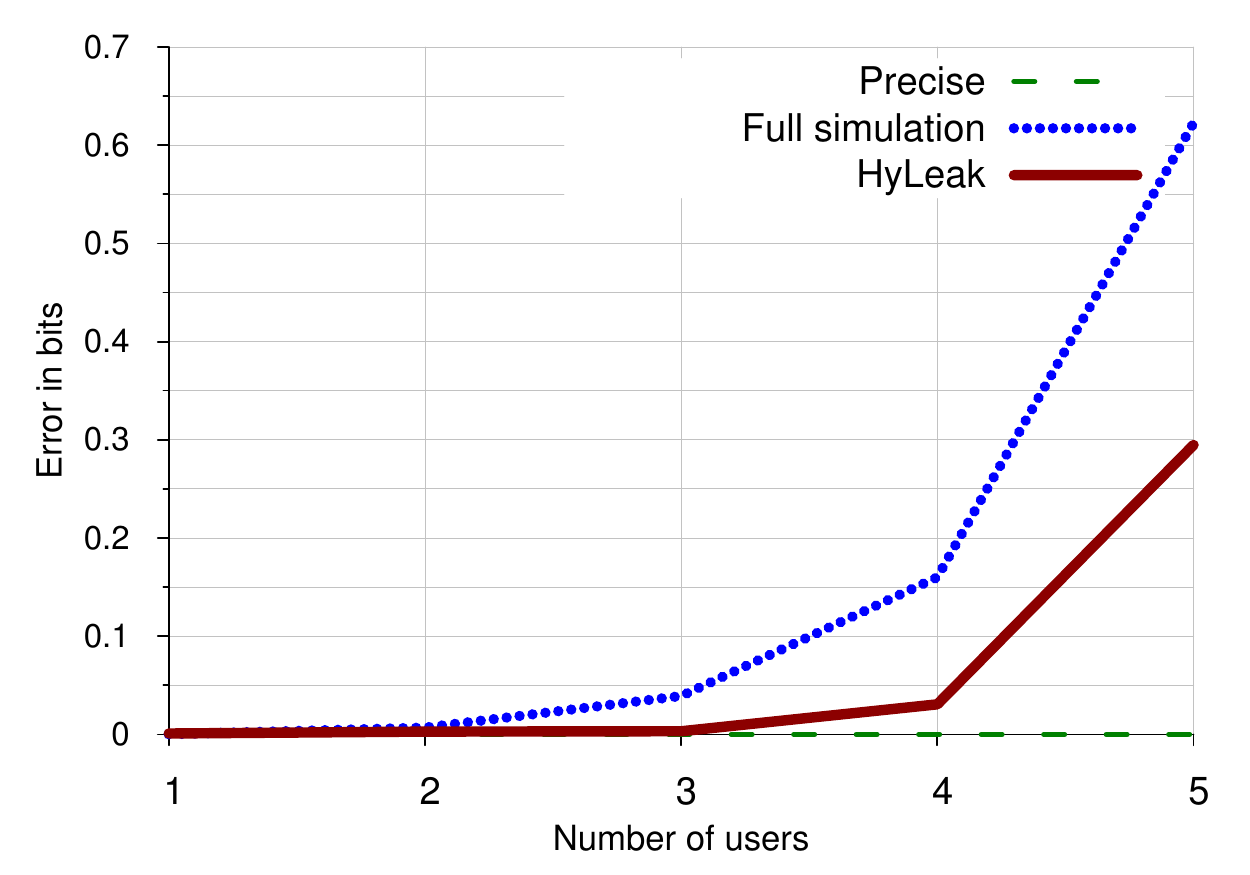}}~~~~~~~~
  \caption{Smart grid experimental results.}
  \label{fig:results:smartgrid}
\end{figure}

\subsubsection{Benchmarks results}
In Table~\ref{tab:benchmark_results} we show the results of all the benchmarks using fully precise, fully statistical and hybrid analyses, for a sample size of 50000 executions.
Timeout is set at 10 minutes. 

\begin{table*}
\centering
\begin{tabular}{c c c c c c c c c c c}
\cmidrule[0.1em]{3-11}
 & & \multicolumn{3}{c}{\bf Precise} & \multicolumn{3}{c}{\bf Statistical} &  \multicolumn{3}{c}{\bf Hybrid} \\[-0.3ex]
\cmidrule{3-11}
 & & {\bf Time(s)} & {\bf Leakage} & {\bf Error} & {\bf Time(s)} & {\bf Leakage} & {\bf Error} & {\bf Time(s)} & {\bf Leakage} & {\bf Error} \\[-0.3ex]
\midrule
\parbox[t]{6mm}{\multirow{6}{*}{\rotatebox[origin=c]{90}{\parbox{1cm}{\centering Random\\ walk}}}} & {\bf N=2} & 0.467 &  2.17 &  0 &  3.85 &  2.19 &  1.15e-2 & 2.97 &  2.17 &  1.37e-4 \\[-0.3ex]
 & {\bf N=3} & 0.748 &  2.17 &  0 &  4.51 &  2.19 &  1.74e-2 & 3.35 &  2.17 &  2.05e-4  \\[-0.3ex]
 & {\bf N=4} & 1.93  &  2.17 &  0 &  5.34 &  2.2  &  2.42e-2 & 4.66 &  2.17 &  2.74e-4  \\[-0.3ex]
 & {\bf N=5} & 5.25  &  2.14 &  0 &  5.4  &  2.17 &  3.03e-2 & 4.66 &  2.14 &  1.14e-4  \\[-0.3ex]
 & {\bf N=6} & 21.2  &  2.11 &  0 &  6.05 &  2.14 &  3.35e-2 & 4.87 &  2.11 &  2.03e-4  \\[-0.3ex]
 & {\bf N=7} & 122   &  2.07 &  0 &  6.14 &  2.12 &  4.58e-2 & 5.34 &  2.07 &  2.57e-4  \\[-0.3ex]
\midrule
\parbox[t]{2mm}{\multirow{5}{*}{\rotatebox[origin=c]{90}{Reservoir}}} & {\bf N=4} & 0.134 &  0.732 &  0 &  3.84 & 0.734 & 2.33e-3 & 2.95 &  0.731 & 9.55e-4 \\[-0.3ex]
 & {\bf N=6}  & 0.645   &  0.918 &  0 &  7    & 0.926 & 7.47e-3 & 4.48 & 0.917 & 1.89e-3  \\[-0.3ex]
 & {\bf N=8}  & 8.8     &  1.1   &  0 &  10.1 & 1.14  & 4.41e-2 & 7.21 & 1.13  & 3.08e-2  \\[-0.3ex]
 & {\bf N=10} & timeout &  n/a   &  n/a &  15.7 & 1.62  &  n/a      & 11.1 & 1.61  &  n/a \\[-0.3ex]
 & {\bf N=12} & timeout &  n/a   &  n/a &  27.7 & 3.02  &  n/a      & 20.7 & 3.01  &  n/a  \\[-0.3ex]
\midrule
\parbox[t]{6mm}{\multirow{3}{*}{\rotatebox[origin=c]{90}{\parbox{1cm}{\centering Lying\\ crypto.}}}} \\ 
 &  & 397 &  0.503 &  0 &  106 &  0.503 &  1.87e-4 & 78.8 &  0.503 &  1.37e-6 \\[-0.3ex]
\\
\midrule
\parbox[t]{6mm}{\multirow{4}{*}{\rotatebox[origin=c]{90}{\parbox{1cm}{\centering Shifting\\ window}}}} & {\bf N=20} & 1.97 &  1.51e-2 &  0 &  2.85 &  2.08e-2 & 5.65e-3 & 2.19 &  1.48e-2 & 3.31e-4 \\[-0.3ex]
 & {\bf N=24} & 3.24 &  1.46e-2 &  0 &  2.81 &  2.21e-2 & 7.58e-3 & 2.32 & 1.55e-2 & 9.16e-4  \\[-0.3ex]
 & {\bf N=28} & 6.99 &  1.42e-2 &  0 &  2.91 &  2.33e-2 & 9.13e-3 & 2.83 & 1.35e-2 & 6.75e-4  \\[-0.3ex]
 & {\bf N=32} & 12.4 &  1.38e-2 &  0 &  2.8  &  2.66e-2 & 1.28e-2 & 3.19 & 1.33e-2 & 5.93e-4  \\[-0.3ex]
\midrule
\parbox[t]{6mm}{\multirow{5}{*}{\rotatebox[origin=c]{90}{\parbox{1.5cm}{\centering Probabilistic\\ termination}}}} \\[0.1cm]
 & {\bf N=5} & n/a &  n/a &  n/a &  4.14 &  0.424 &  n/a & 2.76 &  0.432 &  n/a \\[-0.3ex]
 & {\bf N=7} & n/a &  n/a &  n/a &  4.13 &  0.455 &  n/a & 3.08 &  0.454 &  n/a  \\[-0.3ex]
 & {\bf N=9} & n/a &  n/a &  n/a &  4.43 &  0.472 &  n/a & 3.71 &  0.473 &  n/a \\[-0.3ex]
\\[0cm]
\midrule
\multirow{5}{*}{\rotatebox[origin=c]{90}{Smart grid}} & {\bf S=1} & 57.5 &  8.49e-2 &  0 &  9.96 &  8.51e-2 & 2.31e-4 & 5.74 & 8.59e-2 & 9.43e-4  \\[-0.3ex]
 & {\bf S=2} & 82.9 &  0.181 &  0 &  16.8 &  0.188 & 6.82e-3 & 19.7 & 0.178 & 2.85e-3  \\[-0.3ex]
 & {\bf S=3} & 99   &  0.293 &  0 &  14.2 &  0.332 & 3.9e-2  & 7.23 & 0.296 & 3.16e-3  \\[-0.3ex]
 & {\bf S=4} & 106  &  0.425 &  0 &  18.2 &  0.585 & 0.16    & 8.16 & 0.455 & 3.06e-2  \\[-0.3ex]
 & {\bf S=5} & 136  &  0.587 &  0 &  20.9 &  1.21  & 0.623   & 15.5 & 0.882 & 0.295    \\[-0.3ex]
\bottomrule
\end{tabular}
\caption{Shannon leakage benchmark results using the three different methods (precise, full simulation and hybrid).
The results contain the time (in seconds) taken for the analysis, the value of the leakage (in bits), and the error (in bits) compared to the true result, when the true
result has been computed with the precise analysis. The result n/a means either that the experiment cannot be performed on this example, which is the case for the precise analysis of the probabilistic
terminating loop, or that the error cannot be computed because the precise analysis was not successful.
\label{tab:benchmark_results}}
\end{table*}

The results in Table~\ref{tab:benchmark_results} show the superiority of our hybrid approach compared to the state of the art.
The hybrid analysis scales better than the precise analysis, since it does not need to analyze every trace of the system.
Compared to fully statistical analysis, our hybrid analysis exploits precise analysis on components of the system where statistical estimation would be more expensive than precise analysis.
This allows the hybrid analysis to focus the statistical estimation on components of the system where it converges faster, thus obtaining a smaller confidence interval in a shorter time.

\section{Conclusions and Future Work}\label{sec:conclusion}
We have proposed a hybrid statistical estimation method for estimating mutual information by combining precise and statistical analysis, and for compositionally computing the bias and accuracy of the estimate. This naturally extends to the computation of Shannon entropy and conditional Shannon entropy, generalizing previous approaches on computing  mutual information. The method automatically decomposes a system into components and determines which type of analysis is better for each component according to the components' properties. 

We have also introduced an algorithm to adaptively find the optimal sample sizes for different components in the statistical analysis to minimize their variance and produce a more accurate estimate given a sample size. Moreover, we have presented how to reduce sample sizes by using prior knowledge about systems, including the abstraction-then-sampling technique with qualitative analysis. 
We have shown how to leverage this information on the system to reduce the computation time and error of the estimation.

We have provided an implementation in the freely available tool \toolname{}. 
We have shown both theoretical and experimental results to demonstrate that the proposed approach and implementation outperform the state of the art.

Future work includes developing theory and tools that extend our hybrid  method to the analysis of other properties and integrate further symbolic abstraction techniques into our estimation method.
In particular, we are planning to apply our approach to reasoning about uncertainty propagation in probabilistic programs (e.g.,~\cite{Bouissou:16:TACAS}) by replacing concentration inequalities with our hybrid statistical estimation techniques.
Another possible application of the hybrid analysis is to compute information leakage among adaptive agents in the game-theoretic framework~\cite{Alvim:17:GameSec,Alvim:18:Entropy}, in which each agent probabilistically chooses a strategy that is modeled as a component of a channel.

\section*{Acknowledgment}
This work was supported by JSPS KAKENHI Grant Number JP17K12667, by JSPS and Inria under the Japan-France AYAME Program, by the MSR-Inria Joint Research Center, by the Sensation European grant, and by r\'egion Bretagne.

\bibliographystyle{alpha}
\bibliography{references}

\appendix
\section{Proofs of Technical Results}\label{sec:app:proofs}
In this section we present the detailed proofs of our results.

Hereafter we denote by $Q$ the joint sub-distribution obtained by summing $Q_j$'s:
\begin{dmath*}
Q[x,y] \eqdef \sum_{j\in\J} Q_j[x,y]\;~.
\end{dmath*}
We write $\Qxy$ to denote $Q[x, y]$ for abbreviation.
Then $\Qxy$ is the probability that the execution of the system $\Sys$ yields one of $T_j$'s and has input $x$ and output $y$.

\subsection{Proofs for the Mean Estimation Using the Abstraction-Then-Sampling}
\label{subsec:proof:mean:ATS}

In this section we present the proof for Theorem~\ref{thm:MI:expectation:symbolic} in Section~\ref{subsec:known-subsystem}, i.e., the result on mean estimation using the abstraction-then-sampling.
To show the theorem we present and prove Propositions~\ref{prop:joint:expectation:symbolic}, \ref{prop:Y:expectation:symbolic}, and~\ref{prop:X:expectation:symbolic} below.

First, recall that $\D$ is defined as the set of pairs consisting of inputs and outputs that appear with non-zero probabilities in the execution of the whole system $\Sys$: 
\begin{align*}
\D = \left\{ (x, y) \in\X\times\Y : P_{XY}[x,y] > 0 \right\}
\texttt{.}
\end{align*}
Recall also that $\I = \set{1, 2, \ldots, m}$.
Let $\II = \set{1, 2, \ldots, m'}$ for $m' \le m$.

\begin{proposition}[Mean of joint entropy estimated using the abstraction-then-sampling]\label{prop:joint:expectation:symbolic}
~~
The expected value $\Expect{\hat{H}_{\II}(X; Y)}$ of the estimated joint entropy is given by:
\begin{align*}
\displaystyle&\hspace{-2ex}
\Expect{\hat{H}_{\II}(X, Y)} = 
H(X, Y)-
\!\sum_{i\in\I\setminus\II}\hspace{-0.2ex}
\frac{\theta_i^2}{2n_i}
\Bigl(
\sum_{\,(x,y)\in\D\hspace{-1.2ex}}
\varphi_{ixy}
{\Bigr)}
-
\displaystyle
\sum_{i\in\II}
\frac{\theta_i^2}{2n_i}
\Bigl(
\sum_{\,(x,y)\in\D\hspace{-1.2ex}}
\psi_{ixy}
\Bigr)
+ \O(n_i^{-2})
\end{align*}
where 
$\varphi_{ixy} \allowbreak= {\textstyle \frac{\Di[x,y] - \Di[x,y]^2}{P_{XY}[x, y]}}$
~and~
$\psi_{ixy} \eqdef \frac{\Di[x,y]\pi_i[x] - \Di[x,y]^2}{P_{XY}[x, y]}$.
\end{proposition}

\begin{proof}
We use notations that we have introduced in the previous proofs.
For each $i\in\I$, let $\X_i$ be the set of the elements of $\X$ that appear with non-zero probabilities in the component $S_i$.

As explained in Section~\ref{subsec:known-subsystem} we apply the standard sampling technique (shown in Section~\ref{sec:compositional}) to the components $S_i$ with $i\in\I\setminus\II$, and the abstraction-then-sampling technique to the components $S_i$ with $i\in\II$.
We briefly recall the two techniques below.

Using the standard sampling technique we compute the empirical sub-distribution $\Rh_i$ for $S_i$ with $i\in\I\setminus\II$ as follows.
The analyst first runs $S_i$ a certain number $n_i$ of times to obtain the set of execution traces.
Let $\Kixy$ be the number of traces that have input $x\in\X$ and output $y\in\Y$.
Then $n_i = \sum_{x\in\X,y\in\Y} \Kixy$.
From these numbers $\Kixy$ of traces we compute the empirical joint (full) distribution $\Dhi$ of $X$ and $Y$ by:
\begin{dmath*}
\Dhi[x,y] \eqdef \frac{\Kixy}{n_i}\;.
\end{dmath*}
Since $S_i$ is executed with probability $\theta_i$,
the sub-distribution $\Rh_i$ is given by
$\Rh_i[x,y] \eqdef \theta_i \Dhi[x,y] = \frac{\theta_i \Kixy}{n_i}$.

On the other hand, we use the abstraction-then-sampling sampling technique to compute the empirical sub-distribution $\Rh_i$ for $S_i$ with $i\in\II$ as follows.
Recall that for each $i\in\II$, $\pi_i[x]$ is the probability of having an input $x$ in the component $S_i$.
For each $i\in\II$ all the non-zero rows of $S_i$'s channel matrix are the same conditional distribution; i.e., for each $x,x'\in\X_i$ and $y\in\Y_i$,\, $\frac{P_{XY}[x,y]}{\pi_i[x]} = \frac{P_{XY}[x',y]}{\pi_i[x']}$ when $\pi_i[x] \neq 0$ and $\pi_i[x'] \neq 0$.
Therefore it is sufficient to estimate only one of the rows.
We execute the component $S_i$ with an identical input $x\in\X$ $n_i$ times to record the traces.
Let $\Kiy$ be the number of traces of the component $S_i$ that outputs $y$.
Then we define the empirical joint (full) distribution $\Dhi$ of $X$ and $Y$ as:
\begin{dmath*}
\Dhi[x,y] \eqdef \frac{\pi_i[x]\Kiy}{n_i}
{.}
\end{dmath*}
%
Since $S_i$ is executed with probability $\theta_i$, the sub-distribution $\Rh_i$ is given by:
${\displaystyle \Rh_i[x,y] \eqdef \theta_i \Dhi[x,y] = \frac{\theta_i \pi_i[x]\Kiy}{n_i}}$
{.}

Now the empirical joint probability distribution $\ph_{XY}$ is computed from the above empirical sub-distributions $\Rh_i$ (obtained either by standard sampling or by abstraction-then-sampling) and the exact sub-distributions $Q_j$ (obtained by precise analysis):
\begin{align} \label{eq:pxy}
\ph_{XY}[x,y] = \Qxy + \sum_{i\in\I\setminus\II}\!\frac{\theta_i \Kixy}{n_i} + \sum_{i\in\II} \frac{\theta_i \pi_i[x] \Kiy}{n_i}.
\end{align}

Let $\BKxy = (\Ky{1}, \Ky{2}, \ldots, \Ky{m'}, \Kxy{(m'+1)}, \Kxy{(m'+2)}, \ldots, \Kxy{m})$,
and $\fxy(\BKxy)$ be the $m$-ary function:
\begin{align*}
\fxy(\BKxy) =
\!\Bigl( \Qxy + \hspace{-0.5ex}\sum_{i\in\I\setminus\II}\hspace{-0.5ex} {\textstyle\frac{\theta_i \Kixy}{n_i}} + \sum_{i\in\II} {\textstyle\frac{\theta_i \pi_i[x]\Kiy}{n_i}} {\Bigr)} 
\log\Bigl( \Qxy + \hspace{-0.5ex}\sum_{i\in\I\setminus\II}\hspace{-0.5ex} {\textstyle\frac{\theta_i \Kixy}{n_i}} + \sum_{i\in\II} {\textstyle\frac{\theta_i \pi_i[x]\Kiy}{n_i}} \Bigr)
\end{align*}
which equals $\ph_{XY}[x,y]\log\ph_{XY}[x,y]$.
Then the empirical joint entropy is:
\begin{align*}
\hat{H}_\II(X, Y) =
-\hspace{-1.0ex}\sum_{(x, y)\in\D} \ph_{XY}[x,y] \log \ph_{XY}[x,y]
=
-\hspace{-1.0ex}\sum_{(x, y)\in\D} \fxy(\BKxy).
\end{align*}
Let $\Kioxy = \Expect{\Kixy}$ for each $i\in\I$ and $\BKoxy = \Expect{\BKxy}$.
Then $\Kioxy = n_i \Di[x,y] = \frac{n_i R_i[x,y]}{\theta_i}$, and
$\Kioy = \frac{n_i\Di[x,y]}{\pi_i[x]} = \frac{n_i R_i[x,y]}{\theta_i\pi_i[x]}$.
By the Taylor expansion of $\fxy(\BKxy)$ (w.r.t. the multiple dependent variables $\BKxy$) at $\BKoxy$, we have:
\begin{align*}
\hspace{-2.5ex}
\fxy(\BKxy) &=
\fxy(\BKoxy)
+ \hspace{-1ex}\sum_{i\in\I\setminus\II}\hspace{-1ex}
{\textstyle \frac{\partial\fxy(\BKoxy)}{\partial \Kixy}}
(\Kixy - \Kioxy)
+ \hspace{-0ex}\sum_{i\in\II}
{\textstyle \frac{\partial\fxy(\BKoxy)}{\partial \Kiy}}
\pi_i[x]
(\Kiy - \Kioy)
\\[0.5ex]&\displaystyle
+ \hspace{-1ex}\sum_{\substack{i,j \in\I\setminus\II}}\hspace{-0.5ex} 
\frac{1}{2}
{\textstyle \frac{\partial^2\fxy(\BKoxy)}{\partial \Kixy \partial \Kjxy}}
(\Kixy - \Kioxy)(\Kjxy - \Kjoxy)
+ \hspace{-1ex}\sum_{\substack{i\in\I\setminus\II \\ j\in\II}}\hspace{-0.8ex} 
\frac{1}{2}
{\textstyle \frac{\partial^2\fxy(\BKoxy)}{\partial \Kixy \partial \Kjy}}
\pi_j[x] (\Kixy - \Kioxy)(\Kjy - \Kjoy)
\\[0.5ex]&\displaystyle
+ \hspace{-1ex}\sum_{\substack{i,j \in\II}}\hspace{-0.5ex} 
\frac{1}{2}
{\textstyle \frac{\partial^2\fxy(\BKoxy)}{\partial \Kiy \partial \Kjy}}
\pi_i[x]\pi_j[x] (\Kiy - \Kioy)(\Kjy - \Kjoy)
+ \O(\BKxy^3)
\vspace{-2ex}
\end{align*}

To compute the expected value $\Expect{\hat{H}_\II(X, Y)}$ of the estimated joint entropy, it should be noted that:
\begin{itemize}
\item $\Expect{\Kixy - \Kioxy} = 0$, which is immediate from $\Kioxy = \Expect{\Kixy}$.
\item $\Expect{\Kiy - \Kioy} = 0$, which is immediate from $\Kioy = \Expect{\Kiy}$.
\item If $i \neq j$ then $\Expect{ (\Kixy - \Kioxy) (\Kjxy - \Kjoxy) } = 0$, because $\Kixy$ and $\Kjxy$ are independent.
\item If $i \neq j$ then $\Expect{ (\Kixy - \Kioxy) (\Kjy - \Kjoy) } = 0$, because $\Kixy$ and $\Kjy$ are independent.
\item If $i \neq j$ then $\Expect{ (\Kiy - \Kioy) (\Kjy - \Kjoy) } = 0$, because $\Kiy$ and $\Kjy$ are independent.
\item For each $i\in\I\setminus\II$, $(\Kixy \colon (x,y) \in \D)$ follows the multinomial distribution with the sample size $n_i$ and the probabilities $D_i[x,y]$ for $(x,y) \in \D$, therefore
\begin{align*}
\Expect{ (\Kixy - \Kioxy)^2 } =
\Variance{\Kixy} =
n_i D_i[x,y] \bigl( 1 - D_i[x,y] \bigr)
{.}
\end{align*}
\item For each $i\in\II$, $(\Kiy \colon y \in \Y^+)$ follows the multinomial distribution with the sample size $n_i$ and the probabilities 
$\frac{D_i[x,y]}{\pi_i[x]}$ 
for $(x,y) \in \D$, therefore
\begin{align*}
\Expect{ (\Kiy - \Kioy)^2 } =
\Variance{\Kiy} = 
n_i \frac{D_i[x,y]}{\pi_i[x]} \Bigl( 1 - \frac{D_i[x,y]}{\pi_i[x]} \Bigr).
\end{align*}
\end{itemize}

Hence the expected value of $\fxy(\BKxy)$ is given by:
\begin{align*}
\Expect{\fxy(\BKxy)} = \fxy(\BKoxy)
+ \hspace{-1.5ex}
\sum_{\substack{i \in\I\setminus\II}}
\hspace{-1.2ex}
{\textstyle \frac{1}{2}}
{\textstyle \frac{\partial^2\fxy(\BKoxy)}{\partial \Kixy^2}}
\Expect{ (\Kixy \!- \Kioxy)^2 }
+ \hspace{-0.5ex}
\sum_{\substack{i \in\II}}
\hspace{-0.5ex}
{\textstyle \frac{\pi_i[x]^2 }{2}}
{\textstyle \frac{\partial^2\fxy(\BKoxy)}{\partial \Kiy^2}}
\Expect{ (\Kiy \!- \Kioy)^2 }
\!+\!\O(\BKxy^3)
{.}
\end{align*}

Therefore the expected value of $\hat{H}_\II(X, Y)$ is given by:
\begin{align*}
\Expect{\hat{H}_\II(X, Y)}
=\,& H(X, Y) -
\displaystyle
\hspace{-1ex}
\sum_{(x, y)\in\D}
\!\biggl(
\sum_{i\in\I\setminus\II}\hspace{-0.5ex}
{\textstyle \frac{1}{2}}
{\textstyle \frac{\theta_i^2}{n_i^2 P_{XY}[x,y]}}
n_i {\textstyle D_i[x,y] \left( 1 - D_i[x,y] \right)}
\\
& \hspace{17ex} + 
\!\sum_{i\in\II}
{\textstyle \frac{\pi_i[x]^2}{2}}
{\textstyle \frac{\theta_i^2}{n_i^2 P_{XY}[x,y]}}
n_i {\textstyle \frac{D_i[x,y]}{\pi_i[x]} \left( 1 - \frac{D_i[x,y]}{\pi_i[x]} \right)}
+\O(n_i^{-2})
\biggr)
\\[0.5ex]
=\,& H(X, Y) -
\displaystyle
\hspace{-1.0ex}
\sum_{i\in\I\setminus\II}
\hspace{-0.8ex}
{\textstyle\frac{\theta_i^2}{2n_i}}
\hspace{-1.0ex}
\sum_{(x, y)\in\D}\hspace{-1.5ex}
{\textstyle\frac{\Di[x,y] \left( 1 - \Di[x,y] \right)}{P_{XY}[x,y]}}
- \sum_{i\in\II} 
\hspace{-0.2ex}
{\textstyle\frac{\theta_i^2}{2n_i}}
\hspace{-1.0ex}
\sum_{(x, y)\in\D}\hspace{-1.5ex}
{\textstyle\frac{\Di[x,y] \left( \pi_i[x] - \Di[x,y] \right)}{P_{XY}[x,y]}}
+\sum_{i\in\I}\O(n_i^{-2})
{.}
\end{align*}
\end{proof}

\begin{proposition}[Mean of marginal output entropy estimated using the abstraction-then-sampling]\label{prop:Y:expectation:symbolic}
The expected value $\Expect{\hat{H}_{\II}(Y)}$ of the empirical output entropy is given by:
\begin{align*}
\displaystyle&\hspace{-2ex}
\Expect{\hat{H}_{\II}(Y)} = 
H(Y)-
\sum_{i\in\I}
\frac{\theta_i^2}{2n_i}
\Bigl(
\sum_{\,y\in\Y^+\hspace{-1ex}}\hspace{-0.1ex}
\varphi_{iy}
{\Bigr)}
+ \O(n_i^{-2})
\texttt{.}
\end{align*}
\end{proposition}

\begin{proof}
Recall that 
$\D$ is the set of pairs of inputs and outputs with non-zero probabilities,
$\D_x = \{ y \colon (x, y) \in\D \}$
and
$\D_y = \{ x \colon (x, y) \in\D \}$.
For each $i\in\I\setminus\II$ and $y \in\Y$ let $\Liy = \sum_{x\in\D_y} \Kixy$.
Recall the empirical joint distribution $\ph_{XY}$ in Equation~(\ref{eq:pxy}) in the proof of Proposition~\ref{prop:joint:expectation:symbolic}.

Now the empirical marginal distribution $\ph_{Y}$ on outputs is given by:
\begin{align*}
\ph_{Y}[y] = 
\sum_{x\in\D_y}
\ph_{XY}[x,y]
&= 
\sum_{x\in\D_y}
\Bigl(
\Qxy + \hspace{-0.5ex}
\sum_{i\in\I\setminus\II}\!\frac{\theta_i \Kixy}{n_i} + 
\sum_{i\in\II} \frac{\theta_i \pi_i[x] \Kiy}{n_i}
\Bigr)
\\ &=
\sum_{x\in\D_y}\hspace{-0.5ex}\Qxy +
\sum_{i\in\I\setminus\II}\!\frac{\theta_i \Liy}{n_i} +
\sum_{i\in\II} \frac{\theta_i \Kiy}{n_i}
{.}
\end{align*}

Let $\BKy = (\Ly{1}, \Ly{2}, \ldots, \Ly{m'}, \Ky{m'+1}, \Ky{m'+2}, \ldots, \Ky{m})$,
and $\fy(\BKy)$ be the following $m$-ary function:
\begin{align*}
\fy(\BKy) =
{\Bigl(} \sum_{x\in\D_y}\Qxy + \hspace{-0.5ex}
\sum_{i\in\I\setminus\II}\hspace{-0.5ex} {\textstyle \frac{\theta_i \Liy}{n_i}} + \hspace{-0.2ex}
\sum_{i\in\II} {\textstyle \frac{\theta_i \Kiy}{n_i}} {\Bigr)}
\log{\Bigl(} \sum_{x\in\D_y}\Qxy + \hspace{-0.5ex}
\sum_{i\in\I\setminus\II}\hspace{-0.5ex} {\textstyle \frac{\theta_i \Liy}{n_i}} + \hspace{-0.2ex}
\sum_{i\in\II} {\textstyle \frac{\theta_i \Kiy}{n_i}} \Bigr),
\end{align*}
which equals $\ph_{Y}[y]\log\ph_{Y}[y]$.

Let $\Y^+$ be the set of outputs with non-zero probabilities.
Then the empirical marginal entropy is:
\begin{align*}
\hat{H}_\II(Y) =
-\hspace{-1.0ex}\sum_{y\in\Y^+} \ph_{Y}[y] \log \ph_{Y}[y]
=
-\hspace{-1.0ex}\sum_{y\in\Y^+} \fy(\BKy).
\end{align*}

Let $\Lioy = \Expect{\Liy}$ for each $i\in\I\setminus\II$, and $\BKoy = \Expect{\BKy}$.
Then $\Lioy = \sum_{x\in\D_y} \Kioxy$.
By the Taylor expansion of $\fy(\BKy)$ (w.r.t. the multiple dependent variables $\BKy$) at $\BKoy$, we have:

\begin{align*}
\hspace{-2.5ex}
\fy(\BKy) &=
\fy(\BKoy)
+ \hspace{-1ex}\sum_{i\in\I\setminus\II}\hspace{-1ex}
{\textstyle \frac{\partial\fy(\BKoy)}{\partial \Liy}}
(\Liy - \Lioy)
+ \hspace{-0ex}\sum_{i\in\II}
{\textstyle \frac{\partial\fy(\BKoy)}{\partial \Kiy}}
(\Kiy - \Kioy)
\\[0.5ex]&\displaystyle
+ \hspace{-1ex}\sum_{\substack{i,j \in\I\setminus\II}}\hspace{-0.5ex} 
\frac{1}{2}
{\textstyle \frac{\partial^2\fy(\BKoy)}{\partial \Liy \partial \Ljy}}
(\Liy - \Lioy)(\Ljy - \Ljoy)
+ \hspace{-1ex}\sum_{\substack{i\in\I\setminus\II \\ j\in\II}}\hspace{-0.8ex} 
\frac{1}{2}
{\textstyle \frac{\partial^2\fy(\BKoy)}{\partial \Liy \partial \Kjy}}
(\Liy - \Lioy)(\Kjy - \Kjoy)
\\[0.5ex]&\displaystyle
+ \hspace{-1ex}\sum_{\substack{i,j \in\II}}\hspace{-0.5ex} 
\frac{1}{2}
{\textstyle \frac{\partial^2\fy(\BKoy)}{\partial \Kiy \partial \Kjy}}
(\Kiy - \Kioy)(\Kjy - \Kjoy)
+ \O(\BKy^3)
\vspace{-2ex}
\end{align*}

Recall that $\Dyi[y] = \sum_{x\in\X} \Di[x,y]$.
To compute the expected value $\Expect{\hat{H}_\II(Y)}$ of the estimated marginal entropy, it should be noted that:
\begin{itemize}
\item
$\Expect{\Liy - \Lioy} = 0$, which is immediate from $\Lioy = \Expect{\Liy}$.
\item
$\Expect{\Kiy - \Kioy} = 0$, which is immediate from $\Kioy = \Expect{\Kiy}$.
\item
If $i \neq j$ then $\Expect{ (\Liy - \Lioy)(\Ljy - \Ljoy) } = 0$,
because $\Liy$ and $\Ljy$ are independent.
\item
$\Expect{ (\Liy - \Lioy)(\Kjy - \Kjoy) } = 0$,
because $\Liy$ and $\Kjy$ are independent.
\item 
If $i \neq j$ then $\Expect{ (\Kiy - \Kioy)(\Kjy - \Kjoy) } = 0$,
because $\Kiy$ and $\Kjy$ are independent.

\item
For $i\in\I\setminus\II$, $(\Liy \colon y \in \Y^+)$ follows the multinomial distribution with the sample size $n_i$ and the probabilities $\Dyi[y]$ for $y \in \Y^+$, therefore
\begin{align*}
\Expect{ (\Liy - \Lioy)^2 }
= \Variance{ \Liy }
= n_i \Dyi[y] \left( 1 - \Dyi[y] \right)
{.}
\end{align*}

\item
For $i\in\II$, $(\Kiy \colon y \in \Y^+)$ follows the multinomial distribution with the sample size $n_i$ and the probabilities $\Dyi[y]$ for $y \in \Y^+$, therefore
\begin{align*}
\Expect{ (\Kiy - \Kioy)^2 }
= \Variance{ \Kiy }
= n_i \Dyi[y] \left( 1 - \Dyi[y] \right)
{.}
\end{align*}
\end{itemize}

Hence the expected value of $\fy(\BKy)$ is given by:
\begin{align*}
\Expect{\fy(\BKy)} = \fy(\BKoy)
&+ \hspace{-0.5ex}
\sum_{\substack{i \in\I\setminus\II}}\hspace{-0.2ex}
\frac{1}{2}
{\textstyle \frac{\partial^2\fy(\BKoy)}{\partial \Liy^2}}\,
\Expect{ (\Liy - \Lioy)^2 }
+ 
\sum_{\substack{i \in\II}}
\frac{1}{2}
{\textstyle \frac{\partial^2\fy(\BKoy)}{\partial \Kiy^2}}\,
\Expect{ (\Kiy - \Kioy)^2 }
+ \O(\BKy^3)
{.}
\end{align*}
Therefore the expected value of $\hat{H}_\II(Y)$ is given by:
\begin{align*}
\Expect{\hat{H}_\II(Y)} &= 
\displaystyle
H(Y) -
\hspace{-0.5ex}
\sum_{y\in\Y^+} \Bigl(\hspace{-1.0ex}
\sum_{~~~i\in\I\setminus\II}\hspace{-1ex} 
{\textstyle \frac{\theta_i^2}{2n_i^2 P_Y[y]}}
\Expect{ (\Liy - \Lioy)^2 }
+ \hspace{-0.5ex}\sum_{i\in\II} 
{\textstyle \frac{\theta_i^2}{2n_i^2 P_Y[y]}}
\Expect{ (\Kiy - \Kioy)^2 }
+ \O(n_i^{-2})
\Bigr)\!
\\
&=
H(Y) -
\displaystyle
\sum_{i\in\I} \frac{\theta_i^2}{2n_i}
\sum_{y\in\Y^+}\hspace{-0.5ex}\frac{\Dyi[y] \left( 1 - \Dyi[y] \right)}{P_Y[y]}
+ \O(n_i^{-2})
{.}
\end{align*}
\end{proof}

\begin{proposition}[Mean of marginal input entropy estimated using the abstraction-then-sampling]\label{prop:X:expectation:symbolic}
The expected value $\Expect{\hat{H}_{\II}(X)}$ of the empirical input entropy is given by:
\begin{align*}
\displaystyle&\hspace{-2ex}
\Expect{\hat{H}_{\II}(X)} = 
H(X)-\hspace{-0.5ex}
\sum_{i\in\I\setminus\II}\hspace{-0.5ex}
\frac{\theta_i^2}{2n_i}
\Bigl(
\sum_{\,x\in\X^+\hspace{-1ex}}\hspace{-0.1ex}
\varphi_{ix}
{\Bigr)}
+ \O(n_i^{-2})
\texttt{.}
\end{align*}
\end{proposition}

\begin{proof}
For the components $i\in\II$, the prior $\pi_i[x]$ is known to the analyst and used in the abstraction-then-sampling technique.
Hence these components produce no bias in estimating $H(X)$.
For the components $i\in\I\setminus\II$, we derive the bias in a similar way to the proof of Proposition~\ref{prop:Y:expectation:symbolic}.
Hence the theorem follows.
\end{proof}

\resMeanMISymbolic*

\begin{proof}

By Propositions~\ref{prop:joint:expectation:symbolic}, \ref{prop:Y:expectation:symbolic}, and \ref{prop:X:expectation:symbolic}, we obtain the expected value of the estimated mutual information:
\begin{align*}
\Expect{\hat{I}_\II(X; Y)}
&=
\Expect{\hat{H}_\II(X)} + \Expect{\hat{H}_\II(Y)} - \Expect{\hat{H}_\II(X, Y)} \\
&=
I(X; Y) +
\displaystyle
\sum_{i\in\I\setminus\II}
\frac{\theta_i^2}{2n_i}\!%
\Bigl(\!
\sum_{\,(x,y)\in\D\hspace{-2ex}}\hspace{-1ex}
{\textstyle \frac{\Di[x,y] - \Di[x,y]^2}{P_{XY}[x, y]}}
-\hspace{-0.5ex}
\sum_{\,x\in\X^+\hspace{-1.5ex}}\hspace{-1ex}
{\textstyle \frac{\Dxi[x] - \Dxi[x]^2}{P_X[x]}}
-\hspace{-0.5ex}
\sum_{\,y\in\Y^+\hspace{-1.5ex}}\hspace{-1ex}
{\textstyle \frac{\Dyi[y] - \Dyi[y]^2}{P_Y[y]}}
\!\Bigr) \\
&\hspace{10.5ex}
+
\displaystyle
\sum_{i\in\II}
\frac{\theta_i^2}{2n_i}\!%
\Bigl(\!
\sum_{\,x\in\D_x\hspace{-0.5ex}}\hspace{-0.5ex}
{\textstyle \frac{\Di[x,y]\pi_i[x] - \Di[x,y]^2}{P_{XY}[x, y]}}
-\hspace{-0.5ex}
\sum_{\,y\in\Y^+\hspace{-1.5ex}}\hspace{-1ex}
{\textstyle \frac{\Dyi[y] - \Dyi[y]^2}{P_Y[y]}}
\!\Bigr)\!
+ \O(n_i^{-2}).
\end{align*}
Therefore we obtain the theorem.
\end{proof}

\subsection{Proofs for the Mean Estimation Using Only the Standard Sampling}
\label{subsec:proof:mean:standard}

In this section we present the proofs for Theorem~\ref{lem:MI:expectation:general} in Section~\ref{subsec:estimate-MI-bias} and Proposition~\ref{lem:entropy:expectation} in Section~\ref{sec:other-measures}.

\resMeanMI*
\begin{proof}
If $\II = \emptyset$, then $\Expect{\hat{I}(X; Y)} = \Expect{\hat{I}_\II(X; Y)}$.
Hence the claim follows from Theorem~\ref{thm:MI:expectation:symbolic}.
\end{proof}

\resMeanShannonEntropy*
\begin{proof}
If $\II = \emptyset$, then $\Expect{\hat{H}(X)} = \Expect{\hat{H}_\II(X)}$.
Hence the claim follows from Proposition~\ref{prop:X:expectation:symbolic}.
\end{proof}

\subsection{Proof for the Variance Estimation Using the Abstraction-Then-Sampling}
\label{subsec:proof:var:ATS}

In this section we present the proof for Theorem~\ref{thm:MI:variance:symbolic} in Section~\ref{subsec:known-subsystem}, i.e., the result on variance estimation using the abstraction-then-sampling.

To show the proofs we first calculate the covariances between random variables in Lemmas~\ref{lem:covariance:Kxy:Kxy}, \ref{lem:covariance:Kxy:Ly}, \ref{lem:covariance:Ly:Ly}, \ref{lem:covariance:Ly:Ky}, and~\ref{lem:covariance:Ky:Ky} as follows.
Recall that the covariance $\Cov{A,B}$ between two random variables $A$ and $B$ is defined by:
\begin{align*}
\Cov{A,B} \eqdef
\Expect{\bigl( A - \Expect{A} \bigr) \bigl( B - \Expect{B} \bigr)}
{.}
\end{align*}

\begin{lemma}[Covariance between $\Kixy$ and $K_{\!i'x'y'}$]\label{lem:covariance:Kxy:Kxy}
For any $i,i'\in\I\setminus\II$,\,
$\Cov{\Kixy, K_{\!i'x'y'}}$ is given by:
\begin{align*}
\Cov{\Kixy, K_{\!i'x'y'}}
&=
\begin{cases}
0					&\mbox{ if $i \neq i'$} \\
  n_i \Di[x,y] (1 - \Di[x,y])	&\mbox{ if $i = i'$, $x = x'$ and $y = y'$} \\
- n_i \Di[x,y] \Di[x',y']  	&\mbox{ otherwise.}
\end{cases}
\end{align*}
\end{lemma}
\begin{proof}
Let $i,i'\in\I\setminus\II$.
If $i \neq i'$ then $\Kixy$ and $K_{\!i'x'y'}$ are independent, hence their covariance is $0$.
Otherwise, the theorem follows from that fact that for each $i\in\I\setminus\II$, $(\Kixy \colon (x,y) \in \D)$ follows the multinomial distribution with the sample size $n_i$ and the probabilities $D_i[x,y]$ for $(x,y) \in \D$.
\end{proof}

\begin{lemma}[Covariance between $\Kixy$ and $L_{i'\cdot y'}$]\label{lem:covariance:Kxy:Ly}
For any $i,i'\in\I\setminus\II$,\,
$\Cov{\Kixy, L_{i'\cdot y'}}$ is given by:
\begin{align*}
\Cov{\Kixy, L_{i'\cdot y'}}
&=
\begin{cases}
0					&\mbox{ if $i \neq i'$} \\
  n_i \Di[x,y] (1 - \Dyi[y])	&\mbox{ if $i = i'$, $x = x'$ and $y = y'$} \\
- n_i \Di[x,y] \Dyi[y']  	&\mbox{ otherwise.}
\end{cases}
\end{align*}
\end{lemma}

\begin{proof}
Let $i,i'\in\I\setminus\II$.
If $i \neq i'$ then $\Kixy$ and $L_{\!i'\cdot y'}$ are independent, hence their covariance is $0$.
Otherwise, the covariance $\Cov{\Kixy, L_{i\cdot y'}}$ is calculated as:
\begin{align*}
\Cov{\Kixy, L_{i\cdot y'}}
=
Cov\Bigl[ \Kixy,\, \sum_{x'\in\D_{y}} K_{\!ix'y} \Bigr]
=
\sum_{x'\in\D_y} \Cov{ \Kixy, K_{\!ix'y} }
{.}
\end{align*}
Hence, when $y = y'$:
\begin{align*}
\Cov{\Kixy, L_{i\cdot y'}}
&=
\displaystyle
n_i \Di[x,y] (1 - \Di[x,y])
-\hspace{-1ex}
\sum_{x'\in\D_y\setminus\{x\}}\hspace{-1.5ex} n_i \Di[x,y] \Di[x',y] \\
&=
\displaystyle
n_i \Di[x,y]
-\hspace{-1ex}
\sum_{x'\in\D_y} n_i \Di[x,y] \Di[x',y] \\
&=
\displaystyle
n_i \Di[x,y] (1 - \Dyi[y])
{.}
\end{align*}
When $y \neq y'$:
\begin{align*}
\Cov{\Kixy, L_{i\cdot y'}}
&=
\displaystyle
-\sum_{x'\in\D_y}\hspace{-0.5ex} n_i \Di[x,y] \Di[x',y] =
- n_i \Di[x,y] \Dyi[y]
{.}
\end{align*}
\end{proof}

\begin{lemma}[Covariance between $\Liy$ and $L_{i'\cdot y'}$]\label{lem:covariance:Ly:Ly}
For any $i\in\I\setminus\II$, 
$\Cov{\Liy, L_{i\cdot y'}}$ is given by:
\begin{align*}
\Cov{\Liy, L_{i'\cdot y'}}
&=
\begin{cases}
0					&\mbox{ if $i \neq i'$} \\
  n_i \Dyi[y] (1 - \Dyi[y])	&\mbox{ if $i = i'$, $x = x'$ and $y = y'$} \\
- n_i \Dyi[y] \Dyi[y']  	&\mbox{ otherwise.}
\end{cases}
\end{align*}
\end{lemma}

\begin{proof}
Let $i,i'\in\I\setminus\II$.
If $i \neq i'$ then $\Liy$ and $L_{\!i'\cdot y'}$ are independent, hence their covariance is $0$.
Otherwise, the covariance is calculated as:
\begin{align*}
\Cov{\Liy, \Liy} =
Cov\Bigl[ \sum_{x\in\D_y} \Kixy,\, \sum_{x'\in\D_{y}} K_{\!ix'y} \Bigr] =
\sum_{x\in\D_y} \sum_{x'\in\D_y} \Cov{ \Kixy, K_{\!ix'y} }
{.}
\end{align*}
Hence, when $y = y'$:
\begin{align*}
\Cov{\Liy, \Liy}
=
\displaystyle
\sum_{x\in\D_y} \hspace{-0.5ex} \Bigl( n_i \Di[x,y] (1 - \Di[x,y])
-\hspace{-1ex}
\sum_{x'\in\D_y\setminus\{x\}}\hspace{-1.5ex} n_i \Di[x,y] \Di[x',y] \Bigr)
=
\displaystyle
n_i \Dyi[y] (1 - \Dyi[y])
{.}
\end{align*}
When $y \neq y'$:
\begin{align*}
\Cov{ \Liy, L_{i\cdot y'} }
=
\displaystyle
-\hspace{-0.5ex}\sum_{x\in\D_y} \sum_{x'\in\D_{y'}}\hspace{-1ex} n_i \Di[x,y] \Di[x',y']
=
\displaystyle
-n_i \Dyi[y] \Dyi[y']
{.}
\end{align*}
\end{proof}

\begin{lemma}[Covariance between $\Liy$ and $K_{i'\cdot y'}$]\label{lem:covariance:Ly:Ky}
For any $i\in\I\setminus\II$ and any $i'\in\II$,\, 
$\Cov{\Liy, K_{\!i'\cdot y'}} = 0$.
\end{lemma}

\begin{proof}
The claim is immediate from the fact that $\Liy$ and $K_{\!i'\cdot y'}$ are independent.
\end{proof}

\begin{lemma}[Covariance between $\Kiy$ and $K_{i'\cdot y'}$]\label{lem:covariance:Ky:Ky}
For any $i,i'\in\II$, 
$\Cov{\Kiy, K_{i'\cdot y'}}$ is given by:
\begin{align*}
\Cov{\Kiy, K_{i'\cdot y'}}
&=
\begin{cases}
0					&\mbox{ if $i \neq i'$} \\
  n_i \Dyi[y] (1 - \Dyi[y])	&\mbox{ if $i = i'$ and $y = y'$} \\
- n_i \Dyi[y] \Dyi[y']  	&\mbox{ otherwise.}
\end{cases}
\end{align*}
\end{lemma}

\begin{proof}
Let $i,i'\in\II$.
If $i \neq i'$ then $\Kiy$ and $K_{i'\cdot y'}$ are independent, hence their covariance is $0$.
Otherwise, the claim follows from that fact that for each $i\in\II$, $(\Kiy \colon y \in \Y_i^+)$ follows the multinomial distribution with the sample size $n_i$ and the probabilities $\Dyi[y]$ for $y \in \Y_i^+$.
\end{proof}

\resVarianceMISymbolic*

\begin{proof}
We first define $\Bixy$ and $\Biy$ by the following:
\begin{itemize}
\item 
$\Bixy \eqdef
\frac{\partial\fxy(\BKoxy)}{\partial \Kixy} =
\frac{\theta_i}{n_i} \large(1 + \log P_{XY}[x,y] \large)
$.
\item
For each $i\in\I\setminus\II$,
$\Biy \eqdef
\frac{\partial\fy(\BKoy)}{\partial \Liy} =
\frac{\theta_i}{n_i} \large(1 + \log P_{Y}[y] \large)
$.
\item
For each $i\in\II$,
$\Biy \eqdef
\frac{\partial\fy(\BKoy)}{\partial \Kiy} =
\frac{\theta_i}{n_i} \large(1 + \log P_{Y}[y] \large)
$.
\end{itemize}

Then the variance of $\hat{H}_\II(X, Y)$ is obtained from Lemmas~\ref{lem:covariance:Kxy:Kxy} and~\ref{lem:covariance:Ky:Ky} as follows:
\begin{align*}
\Variance{\hat{H}_\II(X, Y)}
&=
\Expect{\hat{H}_\II(X, Y)^2} - \left(\Expect{\hat{H}_\II(X, Y)}\right)^2 \\
&= \displaystyle \hspace{-1ex}
\sum_{i,i'\in\I\setminus\II}
\sum_{(x, y)\in\D} \sum_{(x', y')\in\D} \hspace{-1.5ex}
\Bixy B_{i'x'y'} \Cov{\Kixy, K_{\!i'x'y'}}
\\
&~~+ \displaystyle \hspace{-1ex}
\sum_{i,i'\in\II}
\sum_{(x, y)\in\D} \sum_{(x', y')\in\D}\hspace{-1.5ex}
\pi_i[x] \Bixy \pi_{i'}[x'] B_{i'x'y'} \Cov{\Kiy, K_{\!i'\cdot y'}}
+ \O(n_i^{-2}) \\
&= \displaystyle \hspace{-0.5ex}
\sum_{i\in\I\setminus\II}
\sum_{(x, y)\in\D}\hspace{-1.5ex}
n_i \Bigl(
\Bixy^2 \Di[x,y] (1 - \Di[x,y]) - \hspace{-4ex}\sum_{~~~~~(x',y')\in\D\setminus\set{(x,y)}\hspace{-4.5ex}}\hspace{-4.5ex} \Bixy B_{ix'y'} \Di[x,y] \Di[x',y']
\Bigr)
\\
&~~+
\sum_{i\in\II}\!
\sum_{(x, y)\in\D}\hspace{-1.5ex}
n_i \Bigl(
\pi_i[x]^2 \Bixy^2 \Dyi[y] (1 - \Dyi[y]) - \hspace{-4ex}
\sum_{~~~~~(x',y')\in\D\setminus\set{(x,y)}\hspace{-4.5ex}}\hspace{-4.5ex} \pi_i[x] \pi_i[x'] \Bixy B_{ix'y'} \Dyi[y] \Dyi[y']
\Bigr)
\\
&~~+ \O(n_i^{-2}) \\[1ex]
&= \displaystyle
\sum_{i\in\I\setminus\II}
n_i \hspace{-1ex}\sum_{(x, y)\in\D}\hspace{-1ex} \Di[x,y] \Bixy \Bigl(
\Bixy - \hspace{-1.0ex}\sum_{(x',y')\in\D}\hspace{-2ex}  B_{ix'y'} \Di[x',y']
\Bigr)
\\[0.2ex]
&~~+ \displaystyle
\sum_{i\in\II}
n_i \biggl(\, 
\sum_{y\in\Y^+} \Dyi[y] \Bigl(\sum_{x\in\D_y} \pi_i[x]^2 \Bixy^2 \Bigr) -
\Bigl(\sum_{(x, y)\in\D} \Dyi[y] \pi_i[x] \Bixy \Bigr)^{\!2}
\biggr)
+ \O(n_i^{-2})
{.}
\end{align*}

The variances of $\hat{H}_\II(Y)$ is obtained from Lemmas~\ref{lem:covariance:Ly:Ly} and~\ref{lem:covariance:Ky:Ky} as follows:
\begin{align*}
\Variance{\hat{H}_\II(Y)}
&=
\Expect{ \hat{H}_\II(Y)^2 } - \left(\Expect{\hat{H}_\II(Y)}\right)^2 \\
&= \displaystyle \hspace{-1.5ex}
\sum_{y,y'\in\Y^+}\hspace{-1ex}
\Bigl(
\sum_{i,i'\in\I\setminus\II}\hspace{-1.5ex}
\Biy B_{i'\cdot y'}
\Cov{\Liy, L_{\!i'\cdot y'}}
+\hspace{-1ex}
\sum_{i,i'\in\II}\hspace{-1ex}
\Biy B_{i'\cdot y'}
\Cov{\Kiy, K_{\!i'\cdot y'}}
\Bigr)
+ \O(n_i^{-2})
\\[0.2ex]
&= \displaystyle \hspace{-0.2ex}
\sum_{i\in\I}
n_i \hspace{-0.5ex}\sum_{y\in\Y^+}\hspace{-0.5ex}
\Dyi[y] \Biy
\Bigl(
 \Biy - \hspace{-0.5ex} \sum_{y'\in\Y^+}\hspace{-1.0ex} B_{i\cdot y'}\Dyi[y']
\Bigr)
+ \O(n_i^{-2})
\\[0.2ex]
&= \displaystyle \hspace{-0.2ex}
\sum_{i\in\I\setminus\II}n_i
\hspace{-2.0ex}\sum_{~~~~(x, y)\in\D}\hspace{-2ex} \Di[x,y]
\Biy \Bigl(\Biy - \hspace{-4ex}\sum_{~~~~~(x', y')\in\D}\hspace{-3.5ex} B_{iy'} \Di[x',y'] \Bigr)
\\
&\hspace{1ex}
+\sum_{i\in\II}
n_i \hspace{-0.5ex}\sum_{y\in\Y^+}\hspace{-0.5ex}
\Dyi[y] \Biy
\Bigl(
 \Biy - \hspace{-0.5ex} \sum_{y'\in\Y^+}\hspace{-1.0ex} B_{i\cdot y'}\Dyi[y']
\Bigr)
+ \O(n_i^{-2})
{.}
\end{align*}

Similarly, for $\B_{ix\cdot} = \frac{\theta_i}{n_i} \Bigl(1 + \log P_{X}[x] \Bigr)$,
the variance of $\hat{H}_\II(X)$ is given by:
\begin{align*}
\Variance{\hat{H}_\II(X)}
\displaystyle
&=
\sum_{i\in\I\setminus\II}\hspace{-0.5ex}
n_i \hspace{-0.5ex}\sum_{(x, y)\in\D}\hspace{-0.5ex}
\Di[x,y] B_{ix\cdot}
\Bigl(
 B_{ix\cdot} - \hspace{-0.5ex} \sum_{(x', y')\in\D}\hspace{-1.0ex} B_{i x'\cdot}\Di[x',y']
\Bigr)  + \O(n_i^{-2})
{,}
\end{align*}
which is symmetric to $\Variance{\hat{H}_\II(Y)}$ only w.r.t. $\I\setminus\II$
\footnote{Note that the abstraction-then-sampling relies on partial knowledge on the prior, i.e., 
the analyst knows $\pi_i[x]$ for all $i\in\II$, hence $\Variance{\hat{H}_\II(X)}$ has no term for $\II$.
On the other hand, the standard sampling here does not use knowledge on the prior.}.

The covariance between $\hat{H}_\II(X, Y)$ and $\hat{H}_\II(Y)$ is obtained from Lemmas~\ref{lem:covariance:Kxy:Ly}
 and~\ref{lem:covariance:Ky:Ky} as follows:
\begin{align*}
\Cov{\hat{H}_\II(Y), \hat{H}_\II(X, Y)}
&= \displaystyle
\sum_{i\in\I\setminus\II}
\sum_{(x, y)\in\D} \sum_{y'\in\Y^+}
\Bixy B_{i\cdot y'} \Cov{\Kixy, L_{i\cdot y'}} \\
&~~+ \displaystyle
\sum_{i\in\II}
\sum_{(x, y)\in\D} \sum_{y'\in\Y^+}
\pi_i[x] \Bixy B_{i\cdot y'} \Cov{\Kiy, K_{i\cdot y'}} 
+ \O(n_i^{-2}) \\
&= \displaystyle \hspace{-0.5ex}
\sum_{i\in\I\setminus\II}
n_i
\sum_{(x, y)\in\D} \Di[x,y] \Bixy
\Bigl( B_{i\cdot y} - \hspace{-1.5ex}\sum_{(x',y')\in\D}\hspace{-1.5ex}B_{i\cdot y'} \Di[x',y'] \Bigr) \\
&~~+ \displaystyle
\sum_{i\in\II}
n_i \hspace{-0.5ex}\sum_{(x, y)\in\D}\hspace{-0.5ex}
\Dyi[y] \pi_i[x]\Bixy
\Bigl(
 \Biy - \hspace{-0.5ex} \sum_{y'\in\Y^+}\hspace{-1.0ex} B_{i\cdot y'}\Dyi[y']
\Bigr)
+ \O(n_i^{-2})
\end{align*}

Similarly, the covariance between $\hat{H}_\II(X, Y)$ and $\hat{H}_\II(X)$ is given by:
\begin{align*}
\Cov{\hat{H}_\II(X), \hat{H}_\II(X, Y)}
= \displaystyle \hspace{-0.5ex}
\sum_{i\in\I\setminus\II}
n_i
\sum_{(x, y)\in\D} \Di[x,y] \Bixy
\Bigl( B_{i x \cdot} - \hspace{-1.5ex}\sum_{(x',y')\in\D}\hspace{-1.5ex}B_{i x'\cdot} \Di[x',y'] \Bigr)
+ \O(n_i^{-2})
\end{align*}

The covariance between $\hat{H}_\II(X)$ and $\hat{H}_\II(Y)$ is given by:
\begin{align*}
\Cov{\hat{H}_\II(X), \hat{H}_\II(Y)}
&= \displaystyle \hspace{-0.5ex}
\sum_{i\in\I\setminus\II}\hspace{-1ex}
n_i \hspace{-1ex}
\sum_{(x, y)\in\D}\hspace{-1ex} \Di[x,y] B_{ix\cdot}
\!\Bigl( B_{i\cdot y} - \hspace{-1.5ex}\sum_{(x',y')\in\D}\hspace{-1.5ex}B_{i\cdot y'} \Di[x',y'] \Bigr)
\!+\!\O(n_i^{-2})
\end{align*}

Therefore the variance of the mutual information is as follows:
\begin{align*}
\hspace{-2ex}
\Variance{\hat{I}_\II(X; Y)}
&= \displaystyle
\Variance{ \hat{H}_\II(X) + \hat{H}_\II(Y) - \hat{H}_\II(X, Y) }
\\[0.0ex]
&= \displaystyle
\Variance{\hat{H}_\II(X)} + \Variance{\hat{H}_\II(Y)} + \Variance{\hat{H}_\II(X, Y)}
+ 2 \Cov{\hat{H}_\II(X), \hat{H}_\II(Y)}
\\
&~~~~
- 2 \Cov{\hat{H}_\II(X), \hat{H}_\II(X, Y)}
- 2 \Cov{\hat{H}_\II(Y), \hat{H}_\II(X, Y)}
\\[0.2ex]
&= \displaystyle \hspace{-1.0ex}
\sum_{~~i\in\I\setminus\II}\hspace{-1ex} n_i
\hspace{-2.0ex}\sum_{~~~~(x, y)\in\D}\hspace{-3ex} \Di[x,y]
\\[-0.5ex]
&\hspace{15ex}
\biggl(\!
B_{ix\cdot} \Bigl(\!B_{ix\cdot} - \hspace{-4ex}\sum_{~~~~~(x', y')\in\D}\hspace{-3.5ex} B_{ix'\cdot} \Di[x',y'] \Bigr)
+ \Biy \Bigl(\!\Biy - \hspace{-4ex}\sum_{~~~~~(x', y')\in\D}\hspace{-3.5ex} B_{i\cdot y'} \Di[x',y'] \Bigr)
\\[-0.5ex]
&\hspace{15ex}
+ B_{ixy} \Bigl(\!B_{ixy} - \hspace{-4ex}\sum_{~~~~~(x', y')\in\D}\hspace{-3.5ex} B_{ix'y'} \Di[x',y'] \Bigr)
+ 2B_{ix\cdot} \Bigl(\!\Biy - \hspace{-4ex}\sum_{~~~~~(x', y')\in\D}\hspace{-3.5ex} B_{i\cdot y'} \Di[x',y'] \Bigr)
\\[-0.5ex]
&\hspace{15ex}
- 2\Bixy \Bigl(\!B_{ix\cdot} - \hspace{-4ex}\sum_{~~~~~(x', y')\in\D}\hspace{-3.5ex} B_{ix'\cdot} \Di[x',y'] \Bigr)
- 2\Bixy \Bigl(\!\Biy - \hspace{-4ex}\sum_{~~~~~(x', y')\in\D}\hspace{-3.5ex} B_{i\cdot y'} \Di[x',y'] \Bigr)
 \biggr)
\\[0.0ex]
&
\hspace{1.5ex}
+
\sum_{i\in\II} n_i\hspace{-0.5ex}
\sum_{y\in\Y^+}\hspace{-0.5ex}
\Dyi[y]
\biggl(
\Biy \Bigl(\Biy - \hspace{-1.0ex}\sum_{y'\in\Y^+} B_{i\cdot y'} \Dyi[y'] \Bigr)
\\[-0.5ex]
&\hspace{21ex}
+ \hspace{-0.5ex}\sum_{x\in\D_y}\hspace{-0.5ex}
\pi_i[x]B_{ixy} \sum_{x'\in\X^+}\hspace{-1ex}\pi_i[x']\Bigl(B_{ix'y} - \hspace{-1.5ex}\sum_{y'\in\D_{x'}} B_{ix'y'} \Dyi[y'] \Bigr)
\\[-0.5ex]
&\hspace{21ex}
-2 \hspace{-0.5ex}\sum_{x\in\D_y}\hspace{-0.5ex}\pi_i[x]\Bixy
\Bigl(
 \Biy - \hspace{-0.5ex} \sum_{y'\in\Y^+}\hspace{-1.0ex} B_{i\cdot y'}\Dyi[y']
\Bigr)
\biggr)
+\, \O(n_i^{-2})
\\[0.0ex]
&= \displaystyle \hspace{-0.5ex}
\sum_{i\in\I\setminus\II} \frac{\theta_i^2}{n_i}
\Biggl(
\sum_{(x, y)\in\D}\hspace{-1.5ex} \Di[x,y]
\Bigl(1 + {\textstyle \log\!\frac{P_X[x]P_Y[y]}{P_{XY}[x,y]}} \Bigr)^{\!2}\hspace{-1ex} -
\biggl(\sum_{(x, y)\in\D}\hspace{-1.5ex} \Di[x,y]
\Bigl(1 + {\textstyle \log\!\frac{P_X[x]P_Y[y]}{P_{XY}[x,y]}} \Bigr) \biggr)^{\!2}
\Biggr)
\\[-0.5ex]
&\hspace{2ex}
+\hspace{-0.5ex}
\sum_{i\in\II} \frac{\theta_i^2}{n_i}
\Biggl(\!
\sum_{~~y\in\Y^+}\hspace{-1ex}
\Dyi[y]
\Bigl(
\log P_Y[y] - 
\sum_{x\in\X} \pi_i[x] \log P_{XY}[x,y]
\Bigr)^{\!2}
\\[-0.5ex]
&\hspace{12ex}
-
\biggl(
\sum_{~~y\in\Y^+}\hspace{-1ex}
\Dyi[y]
\Bigl(
\log P_Y[y] - 
\sum_{x\in\X} \pi_i[x] \log P_{XY}[x,y]
\Bigr)
\biggr)^{\!2}
\Biggr)\!
+\, \O(n_i^{-2})
\end{align*}
\end{proof}

\subsection{Proof for the Variance Estimation Using Only the Standard Sampling}
\label{subsec:proof:var:standard}

In this section we present the proofs for Theorem~\ref{lem:MI:variance:general} in Section~\ref{subsec:eval-accuracy} and Proposition~\ref{lem:entropy:variance} in Section~\ref{sec:other-measures}.

\resVarianceMI*

\begin{proof}
If $\II = \emptyset$, then $\Expect{\hat{I}(X; Y)} = \Expect{\hat{I}_\II(X; Y)}$.
Hence the claim follows from Theorem~\ref{thm:MI:variance:symbolic}.
\end{proof}

\resVarianceShannonEntropy*

\begin{proof}
Let $\II = \emptyset$ and $\B_{ix\cdot} = \frac{\theta_i}{n_i} \Bigl(1 + \log P_{X}[x] \Bigr)$.
Then by the proof of Theorem~\ref{thm:MI:variance:symbolic} in Appendix~\ref{subsec:proof:var:ATS}, we have:
\begin{align*}
\Variance{\hat{H}_\II(X)}
\displaystyle
&=
\sum_{i\in\I}
n_i \hspace{-0.5ex}\sum_{(x, y)\in\D}\hspace{-0.5ex}
\Di[x,y] B_{ix\cdot}
\Bigl(
 B_{ix\cdot} - \hspace{-0.5ex} \sum_{(x', y')\in\D}\hspace{-1.0ex} B_{i x'\cdot}\Di[x',y']
\Bigr)  + \O(n_i^{-2})
\\
&=
\sum_{i\in\I}
n_i \biggl(\,
\sum_{(x, y)\in\D}\hspace{-0.5ex}
\Di[x,y] B_{ix\cdot}^2
- \hspace{-0.5ex}
\Bigl( \sum_{(x, y)\in\D}\hspace{-0.5ex} \Di[x,y] B_{i x\cdot} \Bigr)^{\!2}
\biggr)
+ \O(n_i^{-2})
\\
&=
\sum_{i\in\I} \frac{\theta_i^2}{n_i}
\Biggl(
\hspace{-1.2ex}\sum_{\hspace{1.5ex}x\in\X^+}\hspace{-1.8ex} \Dxi[x]
\Bigl(1 + \log P_X[x] \Bigr)^{\!2}
\!-\!
\biggl(\hspace{-1.2ex}\sum_{\hspace{1.5ex}x\in\X^+}\hspace{-1.8ex} \Dxi[x]
\Bigl(1 + \log P_X[x] \Bigr) \biggr)^{\!2}
\Biggr)\!
+ \O(n_i^{-2})
{.}
\end{align*}

\end{proof}

\subsection{Proofs for Adaptive Analysis}\label{subsec:proof:adaptive}

In this section we present the proofs for the results in Section~\ref{sec:adaptive}.
To prove these it suffices to show the following proposition:

\begin{proposition}\label{prop:adaptive} \rm
Let $v_1, v_2, \ldots , v_m$ be $m$ positive real numbers.
Let $n, n_1, n_2, \ldots, n_m$ be $(m+1)$ positive real numbers such that $\sum_{i = 1}^{m} n_i = n$.
Then
\begin{align*}
\sum_{i = 1}^{m} \frac{v_i}{n_i} \ge
\frac{1}{n}\left( \sum_{i = 1}^{m} \sqrt{v_i} \right)^{\!2}.
\end{align*}
The equality holds when 
$\displaystyle n_i = \frac{\sqrt{v_i}n}{\sum_{j = 1}^{m} \sqrt{v_j}}$
for all $i = 1, 2, \ldots, m$.
\end{proposition}

\begin{proof}
The proof is by induction on $m$.
When $m = 1$ the equality holds trivially.
When $m = 2$ it is sufficient to prove 
\begin{align}\label{eq:adaptive:m2}
\frac{v_1}{n_1} + \frac{v_2}{n_2} \ge
\frac{\left( \sqrt{v_1} + \sqrt{v_2} \right)^2}{n_1 + n_2}.
\end{align}
By $n_1, n_2 > 0$, this is equivalent to 
$
(n_1 + n_2) (n_2 v_1 + n_1 v_2) \ge
n_1 n_2 \left( \sqrt{v_1} + \sqrt{v_2} \right)^2.
$
We obtain this by:
\begin{align*}
(n_1 + n_2) (n_2 v_1 + n_1 v_2) - n_1 n_2 \left( \sqrt{v_1} + \sqrt{v_2} \right)^2
=\, &
(n_1 + n_2) n_2 v_1 + (n_1 + n_2) n_1 v_2 - n_1 n_2 \left( v_1 + 2\sqrt{v_1v_2} + v_2 \right)
\\ =\, &
n_2^2 v_1 + n_1^2 v_2 - 2n_1 n_2 \sqrt{v_1v_2}
\\ =\, &
(n_2 \sqrt{v_1} - n_1 \sqrt{v_2})^2
\\ \ge\, &
0.
\end{align*}

Next we prove the inductive step as follows.
\begin{align*}
\sum_{i = 1}^{m} \frac{v_i}{n_i}
&=
\biggl(\, \sum_{i = 1}^{m-1} \frac{v_i}{n_i} \biggr) + \frac{v_m}{n_m}
\\[-2ex] & \ge
\frac{1}{n_1+\ldots+n_{m-1}}\left( {\textstyle \sum_{i = 1}^{m-1} \sqrt{v_i} } \right)^2 + \frac{\sqrt{v_m}^2}{n_m}
& \mbox{ (by induction hypothesis) }
\\[-0.0ex] & \ge
\frac{1}{(n_1+\ldots+n_{m-1})+n_{m}}\biggl( \sqrt{\Bigl( {\textstyle \sum_{i = 1}^{m-1} \sqrt{v_i} } \Bigr)^2} + \sqrt{v_m} \biggr)^{\!2}
& \mbox{ (by Equation (\ref{eq:adaptive:m2})) }
\\[-0.5ex] & =
\frac{1}{n_1+\ldots+n_{m}}\Bigl( {\textstyle \sum_{i = 1}^{m-1} \sqrt{v_i} + \sqrt{v_m} } \Bigr)^{\!2}
\\[-0.5ex] & =
\frac{1}{n}\Bigl( {\textstyle \sum_{i = 1}^{m} \sqrt{v_i} } \Bigr)^{\!2}.
\end{align*}

Finally, when $n_i = \frac{\sqrt{v_i}n}{\sum_{j = 1}^{m} \sqrt{v_j}}$
for all $i = 1, 2, \ldots, m$, 
then
$
\sum_{i = 1}^{m} \frac{v_i}{n_i}
=
\sum_{i = 1}^{m} \frac{v_i (\sum_{j = 1}^{m} \sqrt{v_j})}{\sqrt{v_i}n}
=
\frac{1}{n}\biggl( \sum_{i = 1}^{m} \sqrt{v_i} \biggr)^{\!2}.
$
\end{proof}

\resAdaptiveMISymbolic*

\begin{proof}
By Proposition~\ref{prop:adaptive} the variance $v = \sum_{i\in\I} \frac{v_i^{\star}}{n_i}$ of estimated mutual information is minimised when $n_i = \frac{\sqrt{v_i^{\star}}n}{\sum_{j = 1}^{m} \sqrt{v_j^{\star}}}$.
Hence the theorem follows.
\end{proof}

\resOptimalSS*

\begin{proof}
Let $\II = \emptyset$.
Then this theorem immediately follows from Theorem~\ref{thm:adaptive-sampling-ATS}.
\end{proof}

\resOptimalSSforShannon*

\begin{proof}
By Proposition~\ref{prop:adaptive} the variance $v = \sum_{i\in\I} \frac{v'_i}{n_i}$ of estimated Shannon entropy is minimised when $n_i = \frac{\sqrt{v'_i}n}{\sum_{j = 1}^{m} \sqrt{v'_j}}$.
Hence the proposition follows.
\end{proof}

\resSampleSizesKnownPrior*

\begin{proof}
By Proposition~\ref{prop:adaptive} the variance $v = \sum_{i\in\I} \sum_{x\in\X^+}\frac{v_{ix}}{n_i}$ of estimated mutual information is minimised when $n_i\lambda_i[x] = \frac{\sqrt{v_{ix}}n}{\sum_{j = 1}^{m} \sqrt{v_{jx}}}$.
Hence the theorem follows.
\end{proof}

\subsection{Proof for the Estimation Using the Knowledge of Priors}
\label{subsec:proof:est:prior}

In this section we present the proof for Propositions~\ref{lem:MI:expectation:KnownPrior} and~\ref{lem:MI:variance:KnownPrior} in Section~\ref{subsec:better-known-prior}, i.e., the result on variance estimation using the knowledge of the prior.

\resMeanMIKnownPrior*

\begin{proof}
Since the precise prior is provided to the analyst, we have
\begin{align*}
\Expect{\hat{I}(X; Y)} = \Expect{H(X)} + \Expect{\hat{H}(Y)} - \Expect{\hat{H}(X, Y)}.
\end{align*}
By using the results on $\Expect{\hat{H}(Y)}$ and $\Expect{\hat{H}(X, Y)}$ in the proof of Theorem~\ref{lem:MI:expectation:general}, we obtain the proposition.
\end{proof}

\resVarianceMIKnownPrior*

\begin{proof}
Since the precise prior is provided to the analyst, we have
\begin{align*}
\Variance{\hat{I}(X; Y)}
&= \displaystyle
\Variance{ H(X) + \hat{H}(Y) - \hat{H}(X, Y) } \\[0ex]
&= \displaystyle
\Variance{\hat{H}(Y)} + \Variance{\hat{H}(X, Y)}
- 2 \Cov{\hat{H}(Y), \hat{H}(X, Y)}.
\end{align*}
By using the results on $\Variance{\hat{H}(Y)}$, $\Variance{\hat{H}(X, Y)}$ $\Cov{\hat{H}(Y), \hat{H}(X, Y)}$ in the proof of Theorem~\ref{lem:MI:variance:general}, we obtain the proposition.
\end{proof}

\label{lastpage}

\end{document}